\documentclass[conference]{IEEEtran}
\usepackage{fancyhdr}

\usepackage{color-edits}
\usepackage{algorithm}
\usepackage[noend]{algpseudocode}

\usepackage{amsthm}
\usepackage{xcolor}
\definecolor{darkgrey}{RGB}{70,70,70}
\definecolor{lightgrey}{RGB}{200,200,200}
\definecolor{lyellow}{RGB}{255,255,100}
\definecolor{llyellow}{RGB}{250,250,180}
\definecolor{lgreen}{RGB}{144,238,144}
\definecolor{lblue}{RGB}{180,180,230}

\usepackage[customcolors]{hf-tikz}
\hfsetbordercolor{white}
\hfsetfillcolor{vlgray}

\definecolor{vlgray}{rgb}{0.77 0.77 0.77}
\definecolor{ablack}{rgb}{0.2 0.2 0.2}
\definecolor{vllgray}{rgb}{0.9 0.9 0.9}
\definecolor{bblue}{rgb}{0.5 0.5 0.99}

\usepackage{colortbl}

\newif\iftr     
\newif\ifall    
\newif\ifconf   
\newif\ifsq     
\newif\ifnonb   
\newif\iftodos  

\newif\ifsqCAP
\newif\ifsqVS
\newif\ifsqEN
\newif\ifsqTIT
\newif\ifsqEQ

\newcommand{\tr}[1]{\iftr #1 \fi}

\newcommand{\all}[1]{\ifall #1 \fi}

\sqtrue
\sqCAPtrue
\sqENtrue
\sqVStrue
\sqTITtrue
\sqEQtrue

\sqfalse
\sqCAPfalse
\sqENfalse
\sqVSfalse
\sqTITfalse
\sqEQfalse

\usepackage{etex}
\usepackage{balance}
\usepackage{epstopdf}
\usepackage{placeins}

\usepackage{cite}

\usepackage{graphicx}
\usepackage{float}
\usepackage{dblfloatfix}
\usepackage{multirow}
\usepackage{rotating}
\usepackage[switch]{lineno} 
\usepackage{makecell}
\usepackage{tabulary}
\usepackage{parcolumns}
\usepackage{tikz}
\usetikzlibrary{tikzmark}

\usepackage{xpatch}
\expandafter\xpatchcmd
\csname pgfk@/tikz/every picture/.@cmd\endcsname
{\thepage}{\arabic{page}}{}{}

\tikzstyle{comment} = [draw, fill=blue!70, text=white, text width=3cm, minimum height=1cm, rounded corners, align=left, font=\scriptsize]
\tikzstyle{background_alg} = [draw, fill=blue!20, opacity=0.4, inner sep=4pt, rounded corners=2pt]

\usetikzlibrary{shapes}
\usetikzlibrary{plotmarks}
\usetikzlibrary{calc, fit}

\usepackage{enumitem}

\usepackage{amsmath,amsfonts}
\usepackage{mathtools,mathrsfs}

\newtheorem{theorem}{Theorem}[section]
\newtheorem{prop}[theorem]{Proposition}

\newtheorem{lemma}[theorem]{Lemma}

\DeclareMathOperator{\E}{\mathbb{E}}

\usepackage{soul}
\usepackage{fontawesome}
\usepackage{pifont}
\usepackage{textcomp}
\usepackage{booktabs}
\usepackage{url}
\usepackage{pbox}
\usepackage[normalem]{ulem}
\usepackage[10pt]{moresize}

\usepackage{cleveref}
\usepackage[utf8]{inputenc}
\crefname{section}{§}{§§}
\Crefname{section}{§}{§§}

\newcommand{\macb}[1]{\textbf{{#1}}}

\ifsqCAP
\usepackage[font={normalfont, scriptsize}]{caption}
\usepackage[font={normalfont, scriptsize}]{subcaption}
\else
\usepackage[font={normalfont, footnotesize}]{caption}
\usepackage[font={normalfont, footnotesize}]{subcaption}
\fi

\newcommand{\vspaceSQ}[1]{\ifsqVS\vspace{#1}\fi}
\newcommand{\enlargeSQ}[1]{\ifsqEN\enlargethispage{\baselineskip}\fi}

\ifsqTIT
\usepackage[compact]{titlesec}
\titlespacing*{\section}{0pt}{6pt}{2pt}
\titlespacing*{\subsection}{0pt}{5pt}{1pt}
\titlespacing*{\subsubsection}{0pt}{5pt}{1pt}
\fi

\usepackage{inconsolata}
\usepackage{listings}

\ifsq
\else
\fi

\newcommand{\maciej}[1]{\textcolor{blue}{[Maciej: #1]}}

\newcommand{\cesar}[1]{\textcolor{teal}{[Cesare: #1]}}

\newcommand{\jakub}[1]{\textcolor{blue}{[Jakub: #1]}}

\definecolor{hlL}{rgb}{0.8 0.8 0.99}

\newcounter{highlight}

\newcounter{hlLR} 
\newcommand{\hlLR}[2]{%
\stepcounter{hlLR}\hfsetfillcolor{hlL}\tikzmarkin[set fill color=hlL, set border color=white, above offset=0.225, right offset=#1, below offset=-0.125]{\thehighlight}#2\tikzmarkend{\thehighlight}}

\newcounter{hlLIR}

\newcounter{hlLIIR}

\newcounter{Ahighlight}

\usepackage{scalerel,stackengine}
\stackMath
\newcommand\rwh[1]{%
\savestack{\tmpbox}{\stretchto{%
  \scaleto{%
        \scalerel*[\widthof{\ensuremath{#1}}]{\kern-.6pt\bigwedge\kern-.6pt}%
                  {\rule[-\textheight/2]{1ex}{\textheight}}
                              }{\textheight}%
}{0.5ex}}%
\stackon[1pt]{#1}{\tmpbox}%
}

\usepackage[hang,flushmargin]{footmisc}

\renewcommand{\epsilon}{\ensuremath\varepsilon}

\renewcommand{\phi}{\ensuremath{\varphi}}

\newcommand{\fRB}[1]{\left(#1\right)}
\newcommand{\fSB}[1]{\left[#1\right]}

\newcommand{\fVB}[1]{\left\vert#1\right\vert}

\usepackage{relsize}
\usepackage{comment}

\renewcommand{\marginpar}[1]{}
\renewcommand{\hl}[1]{#1}
\renewcommand{\enlargethispage}[1]{}

\newif\ifkdd
\newif\ifkddr
\newif\iffull
\newif\ifdebug
\newif\ifreg
\newif\iffullapp
\newif\ifacm
\newif\ifblind

\conftrue
\trfalse
\blindtrue

\conffalse
\trtrue
\blindfalse

\regtrue
\debugtrue
\fulltrue
\kddtrue
\kddrtrue
\allfalse
\fullappfalse
\todostrue

\acmfalse

\iftr
\else
\let\oldsubsubsection\subsubsection
\renewcommand{\subsubsection}[1]{\vspace{-0.4em}\oldsubsubsection{#1}}
\fi

\begin{document}


\title{ProbGraph: High-Performance and High-Accuracy Graph Mining with Probabilistic Set Representations}

\blindfalse

\ifblind
\else 

\author{\IEEEauthorblockN{Maciej Besta$^{1\dagger}$, Cesare Miglioli$^2$, Paolo Sylos Labini$^3$, Jakub Tětek$^4$, Patrick
Iff$^1$,\\ Raghavendra Kanakagiri$^5$, Saleh Ashkboos$^1$, Kacper Janda$^6$, Michał Podstawski$^{7,8}$,\\ Grzegorz
Kwaśniewski$^1$, Niels Gleinig$^1$, Flavio Vella$^9$, Onur Mutlu$^1$, Torsten Hoefler$^{1\dagger}$\vspace{0.5em}}
\IEEEauthorblockA{$^1$ETH Zurich, Zürich, Switzerland\quad
       $^2$Research Center for Statistics, University of Geneva, Geneva, Switzerland\\
       $^3$Free University of Bozen-Bolzano, Bolzano, Italy\quad
       $^4$BARC, University of Copenhagen, Copenhagen, Denmark\\
       $^5$UIUC, Champaign, USA\quad
       $^6$AGH-UST, Kraków, Poland\quad 
       $^7$Warsaw University of Technology, Warsaw, Poland\\
       $^8$TCL Research Europe, Warsaw, Poland\quad
       $^9$University of Trento, Trento, Italy\\
$^\dagger$Corresponding authors:
maciej.besta@inf.ethz.ch, htor@inf.ethz.ch
} 
}

\fi

\maketitle

\ifconf
\else

\thispagestyle{plain}
\pagestyle{plain}

\fi

\begin{abstract}
Important graph mining problems such as Clustering are computationally
demanding. To significantly accelerate these problems, we propose ProbGraph: a
graph representation that enables simple and fast approximate parallel graph
mining with strong theoretical guarantees on work, depth, \emph{and} \hl{result} accuracy.
The key idea is to represent sets of vertices using probabilistic set
representations such as Bloom filters. These representations are much faster to
process than the original vertex sets thanks to vectorizability and small size.
We use these representations as building blocks in important parallel graph
mining algorithms such as Clique Counting or Clustering. When enhanced with
ProbGraph, these algorithms significantly outperform tuned parallel exact
baselines (up to nearly 50$\times$ on 32 cores) while ensuring accuracy of more
than 90\% for many input graph datasets. Our novel bounds and algorithms based on
probabilistic set representations with desirable statistical properties are of
separate interest for the data analytics community.
\end{abstract}

\vspace{0.5em}

\ifconf
\else


\fi

\marginparsep=1em
\marginpar{\vspace{-15em}\colorbox{yellow}{\textbf{R-4}}}

\section{Introduction}
\label{sec:intro}

\enlargeSQ

\enlargeSQ
\enlargeSQ

Graph mining is an important part of the graph processing landscape, with
problems related to discovering patterns in graphs, for example Clustering,
Clique Counting, or Link Prediction~\cite{besta2021sisa, arora2020survey,
besta2021motif, besta2021graphminesuite}.
Accelerating graph mining is notoriously difficult because these problems are
hard to parallelize due to properties such as high irregularity or little
locality~\cite{DBLP:journals/ppl/LumsdaineGHB07, sakr2020future,
besta2021enabling, tate2014programming, besta2017push, besta2015accelerating}.
Simultaneously, graph mining underlies many important computational problems in
social network analysis, machine learning, computational sciences, and
others~\cite{cook2006mining, jiang2013survey, horvath2004cyclic,
chakrabarti2006graph}. 

In approximate computing, a certain (ideally small) amount of accuracy is
sacrificed, in exchange for speedups or reduced energy
consumption~\cite{mittal2016survey, xu2015approximate, han2013approximate}.
\iftr
It relaxes the need for full precision at the level of arithmetic blocks,
processing units, pertinent error and quality measures, algorithms, programming
models, and many others.
\fi
Traditionally, graph algorithms with provable approximation ratios were
developed to alleviate the hardness of various NP-Complete graph problems, such
as minimum graph coloring~\cite{halldorsson1993still, jones1993parallel}.
Still, these works are usually complex, specific to a particular graph problem,
and often need additional heuristics to be easily used in practice.

\marginparsep=2em
\marginpar{\vspace{1em}\colorbox{yellow}{\textbf{R-5}}\\ \colorbox{yellow}{(minor}\\ \colorbox{yellow}{comm-}\\ \colorbox{yellow}{-ent 20)}}

\sethlcolor{yellow} \hl{Moreover, there are many heuristics for approximating
graph properties such as betweenness
centrality~\mbox{\cite{riondato2016fast, borassi2016kadabra, riondato2018abra,
geisberger2008better, bader2007approximating, solomonik2017scaling}}, minimum
spanning tree weight~\mbox{\cite{chazelle2005approximating}}, maximum
matching~\mbox{\cite{besta2020substream}},
reachability~\mbox{\cite{dumbrava2018approximate}}, graph
diameter~\mbox{\cite{chechik2014better, roditty2013fast}}, and
others~\mbox{\cite{slota2014complex, roditty2013fast, boldi2011hyperanf,
echbarthi2017lasas, besta2020high}}.
\emph{Unfortunately, these schemes are all specific to a particular graph problem
or algorithm.}}\sethlcolor{yellow}

To alleviate the above-mentioned issues with accelerating
graph mining, we propose {ProbGraph} (PG), a {probabilistic
graph representation} for {{simple}, {versatile}, {fast}, and {{tunable}}
approximate graph mining}. 
We observe that the \emph{input graph and many auxiliary data structures are
effectively a collection of sets of vertices and
edges}~\cite{besta2021graphminesuite, besta2021sisa}. 
Here, {our key idea} is to encode such sets with carefully selected
\textbf{probabilistic set representations} (sometimes called set sketches) such
as Bloom filters~\cite{bloom1970space}.  This results in a ``probabilistic''
graph representation that -- as we will show -- can accelerate different graph
algorithms at a (tunable) accuracy-storage-performance tradeoff.
Importantly, {sets} and {set operations} are {common}
in graph problems, making PG applicable to many algorithms.


We first show that many time-consuming operations in different graph algorithms
can be expressed with \emph{set intersection cardinality} $|X \cap Y|$. For
example, deriving common parts of vertex neighborhoods takes more than 90\% of the time in
common Triangle Counting algorithms~\cite{han2018speeding,
besta2021graphminesuite, besta2021sisa}, and it can be expressed as a sum of
$|X \cap Y|$ (over different $X$ and $Y$). \emph{{We identify more such graph
algorithms}}. 

\enlargeSQ
\enlargeSQ

Second, we carefully select three probabilistic set representations: Bloom
filters (BF)~\cite{bloom1970space} and two types of MinHash
(MH)~\cite{broder1997resemblance}. We use these set representations to design
\emph{estimators}~$\widehat{|X \cap Y|}$ that approximate $|X \cap Y|$. 
Our central motivation is that these estimators are much faster to obtain than
the exact $|X \cap Y|$ thanks to performance-friendly properties.
We conduct a \emph{work-depth} analysis, \emph{showing formally that
ProbGraph-enhanced graph algorithms have abundant parallelism}.
For example, $X$ and $Y$, when represented with Bloom filters, are {bit
vectors}. Thus, $|X \cap Y|$ amounts to computing a bitwise AND, followed by a
reduction. Such an operation \emph{significantly benefits from vectorization}.
Moreover, thanks to its fixed-size set representations, ProbGraph \emph{exhibits
excellent load balancing properties}: all set intersections are conducted
over the same size bit vectors, annihilating issues related to intersecting 
neighborhoods of different sizes. This is particularly attractive as
modern graph datasets have very often high skews in degree distributions~\cite{besta2021enabling}.

Importantly, \emph{we ensure that our estimators have strong theoretical
guarantees on their accuracy (i.e., quality)}.
%
%
For this, we develop or adapt bounds on the quality of
estimators~$\widehat{|X \cap Y|}$. Here, we prove how far they deviate from
the true size of the set intersection, and we provide upper bounds for their mean squared error (MSE).
We also produce quality bounds that are better than past
works.
To derive these quality bounds, we use concentration inequalities and other
statistical concepts such as sub-Gaussian random
variables~\cite{boucheron2013concentration}. 
For example, the probability that our estimator based on MH deviates from the true value by more than a given distance~$t$, decreases \emph{exponentially} with $t$. 
%
%
%



\enlargeSQ

Moreover, we show that, for some representations, we offer \emph{Maximum
Likelihood Estimators} (MLE)~\cite{casella2002statistical}. Thus, these
estimators are asymptotically\footnote{with increasing sketch size, for a
fixed input.} unbiased (the expected estimator value converges to the true
parameter value) and efficient (asymptotically, \emph{no
other estimator} has lower variance).
%
%
%
Due to the prevalence of BF and MH in high-performance data mining, our novel
results on the theory of probabilistic set representations are of interest
beyond graph analytics.

\marginparsep=1em
\marginpar{\vspace{1em}\colorbox{yellow}{\textbf{R-5}}\\ \colorbox{yellow}{(minor}\\ \colorbox{yellow}{comm-}\\ \colorbox{yellow}{-ent 1)}}

Our fast parallel implementation of PG enables important graph mining problems
(TC, clustering~\cite{estivill2002so, aggarwal2010survey, schaeffer2007graph},
4--clique counting~\cite{cliques}, and vertex
similarity~\cite{jarvis1973clustering, besta2020communication}) \hl{to} achieve very
large performance advantages over tuned baselines, even up to 50$\times$ lower
runtimes when using 32 cores. This is caused by the fact that both work and
depth of ProbGraph-enhanced graph algorithms are lower than those of the exact
tuned baselines.  Simultaneously, using small fixed-size probabilistic
representations makes it much easier to load balance expensive set operations.
Simultaneously, ProbGraph achieves high accuracy of more than 90\% for many
inputs.

\marginparsep=2em
\marginpar{\vspace{4em}\colorbox{yellow}{\textbf{R-5}}\\ \colorbox{yellow}{(minor}\\ \colorbox{yellow}{comm-}\\ \colorbox{yellow}{-ent 2)}}

We also provide strong theoretical guarantees on work and depth of all the
considered graph mining algorithms. Finally, we use our $|X \cap Y|$ estimators
to derive novel estimators~$\widehat{TC}$ on the triangle count~$TC$ in an
arbitrary graph, achieving strong statistical properties such as MLE and
exponential concentration quality bounds. \emph{Our TC estimator based on
MH has better theoretical properties (e.g., is MLE) than past
theoretical results~\cite{tsourakakis2009doulion, pagh2012colorful,
bandyopadhyay2016topological, iyer2018asap, iyer2018bridging, besta2019slim,
eden2017approximately, assadi2018simple, tetek2021approximate}.}

Overall, PG enables trading a small amount of accuracy and storage for more
performance. These tradeoffs are \emph{tunable} by the user, who can select
which aspect is most important. 

\begin{figure*}[t]
\centering
\includegraphics[width=1.0\textwidth]{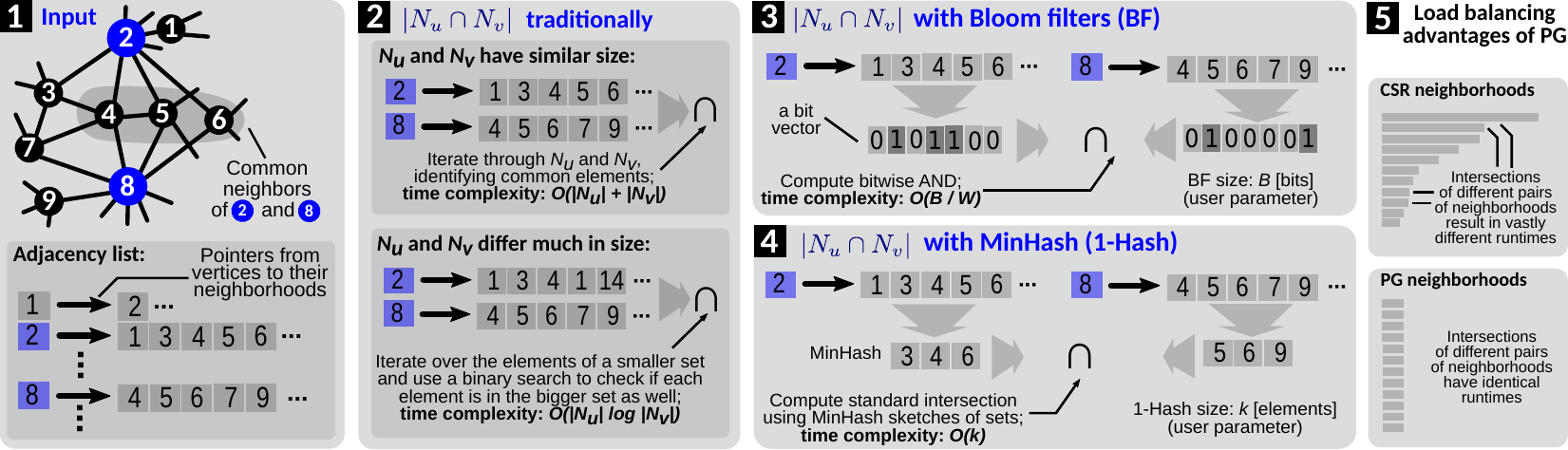}
\vspaceSQ{-1.5em}
\caption{\textmd{Overview of selected PG set representations,
how they are used to accelerate intersections
of vertex neighborhoods, and in alleviating load imbalance. 
\hl{In panel }``\protect\includegraphics[scale=0.12,trim=0 16 0 0]{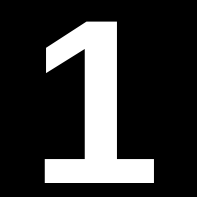}''\hl{, we show a part of an example input
graph.
In panel }``\protect\includegraphics[scale=0.12,trim=0 16 0 0]{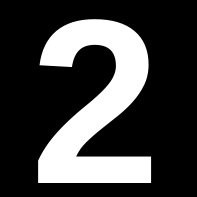}''\hl{, we
illustrate a traditional exact way to compute the count of shared neighbors
$|N_u \cap N_v|$ of any vertices $u$ and $v$. There are
two variants of this operation: ``merge'' (the upper sub-panel), where one
simply merges two sorted sets (it is more beneficial when sets have similar
sizes), and ``galloping'' (the lower sub-panel), where -- for each element from
a smaller set, one uses binary search to check if this element is in the larger
set (it works better when one set is much larger than the other one).}
\hl{In panel }``\protect\includegraphics[scale=0.12,trim=0 16 0
0]{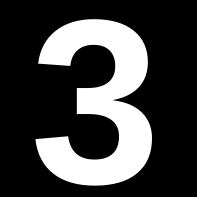}''\hl{, we show how to compute $|N_u \cap N_v|$ with BF. Here, $N_u$
and $N_v$ are first converted into bit vectors (cf.~wide vertical arrows).
Then, the resulting bit vectors are intersected with a very fast bitwise AND
operation (cf.~wide horizontal arrows).}
\hl{In panel }``\protect\includegraphics[scale=0.12,trim=0 16 0
0]{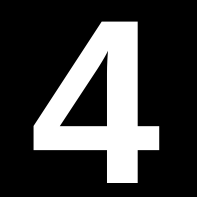}''\hl{, we show how to compute $|N_u \cap N_v|$ with MH. Here, $N_u$
and $N_v$ are first appropriately hashed into much smaller subsets of $N_u$ and
$N_v$ (cf.~wide vertical arrows). The resulting sets can be intersected fast 
because they are much smaller than the original $N_u$ and $N_v$ (indicated with wide horizontal
arrows).}
\hl{In panel }``\protect\includegraphics[scale=0.12,trim=0 16 0
0]{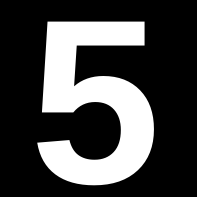}''\hl{, we show load balancing benefits from using PB
(i.e., it is easy to load balance intersections of same-sized PG neighborhoods).}
}}
\label{fig:pg-overview}
\vspaceSQ{-1em}
\end{figure*}

\all{Moreover, to the best of our knowledge, we provide
\emph{the most extensive analysis of different probabilistic set
representations}. We crystallize the outcomes of this analysis in a set of
\emph{\textbf{guidelines}} for algorithm designers, in which we precisely
indicate best use cases for proposed estimators and corresponding set
representations \emph{in the context of graph analytics}.}

\section{Fundamental Concepts}
\label{sec:fund}

We first present the used concepts ad symbols (see Table~\ref{tab:symbols}).

\enlargeSQ

\ifconf
\begin{table}[h]
\vspaceSQ{-1em}
\centering
\footnotesize
\scriptsize
\setlength{\tabcolsep}{1pt}
\ifsq\renewcommand{\arraystretch}{0.7}\fi
\else
\begin{table*}[b]
\centering
\small
\fi
\begin{tabular}{@{}l|ll@{}}
\toprule
\multirow{9}{*}{\begin{turn}{90}\shortstack{Graph, code}\end{turn}}
& $G = (V,E)$ & A graph $G$; $V, E$ are sets of vertices and edges, respectively.\\
%
%
& $n, m$ & Numbers of vertices and edges in $G$; $|V| = n, |E| = m$.\\
& $d_v, N_v$ & The degree and neighbors of $v \in V$.\\
& $d, \overline{d}$ & The maximum and the average degree in $G$ ($\overline{d} = m/n$).\\
\iftr
& $TC$ & Count of triangles in a given graph. \\
\fi
%
& $W$ & Size of a memory word [bits]. \\
& $s$ & Storage budget, i.e., space dedicated to PG structures. \\
%
%
\midrule
\multirow{5}{*}{\begin{turn}{90}\shortstack{BF}\end{turn}}
& $\mathcal{B}_X, \mathbf{B}_X, \mathbf{B}_X[i]$ & A Bloom filter; the associated bit vector; the $i$-th bit in $\mathbf{B}_X$.\\
& $B_X, B_{X,1}, B_{X,0}$ & The size of $\mathbf{B}_X$ (\#bits); number of ones and zeros in $\mathbf{B}_X$. \\
& $b$ & The number of hash functions used with a given Bloom filter. \\
& $h_i$ & An $i$-th associated hash function, $i \in \{1, ..., b\}$. \\
\iftr 
& $p_f$ & The probability of a false positive of a given Bloom filter. \\
\fi
\midrule
\if
\multirow{2}{*}{\begin{turn}{90}\shortstack{MH, KMV}\end{turn}}
\else
\multirow{4}{*}{\begin{turn}{90}\shortstack{MinHash}\end{turn}}
\fi
& $\mathcal{M}_X, M_X$ & $k$-Hash variant; the approximating set based on input $X$. \\
& $k$ & The number of elements stored in a given MinHash. \\
& $h_i$ & An $i$-th associated hash function, $i \in \{1, ..., k\}$. \\
& $\mathcal{M}^1_X, M^1_X$ & 1-hash variant; the approximating set based on input $X$. \\
%
\bottomrule
\end{tabular}
\vspaceSQ{-0.5em}
\caption{\textmd{Important used symbols. ``$X$'' denotes the
input set approximated by a respective representation (omitted
if it is clear from context).}}
\label{tab:symbols}
\vspaceSQ{-2em}
\ifconf
\end{table}
\else
\end{table*}
\fi

\subsection{Graph Model and Representation}
\label{sec:models_int}

A graph $G$ is modeled as a tuple $(V,E)$ with a set of vertices ($V, |V| = n$)
and edges ($E \subseteq V \times V, |E| = m$). 
We model vertices with their integer IDs ($V = \{1, ..., n\}$).
$N_v$ and $d_v$ denote the neighbors and the degree of a given vertex~$v$
($N_v \subset V$); $d$ is $G$'s maximum degree. 
We store the input (i.e., not sketched) graph $G$ using the standard
Compressed Sparse Row (CSR) format, in which all neighborhoods $N_v$ form a
contiguous array ($2m$ words if $G$ is undirected). Next, there is an array
with $n$ pointers to each representation of $N_v$. Each $N_v$ is stored as a
contiguous sorted array of vertex IDs.

\subsection{Work-Depth Analysis of Parallel Algorithms}

We use {work-depth (WD) analysis} for
bounding run-times of parallel algorithms, to analyze PG set
intersections (\cref{sec:intersect_summary-par}, Table~\ref{tab:queries-int}),
construction costs (\cref{sec:constr-costs}, Table~\ref{tab:constr}), and graph
algorithms (\cref{sec:par-algs-wd}, Table~\ref{tab:wd-algs}). The \emph{work}
of an algorithm is the total number of operations and the \emph{depth} is
the longest sequential chain of execution in the algorithm (assuming
infinite number of parallel threads executing the algorithm)~\cite{Bilardi2011,
blelloch2010parallel}.

%

\iftr

\subsection{Set Algebra}

When using arbitrary sets, we use symbols $X = \{x_1, ..., x_l\}$ and $Y =
\{y_1, ..., y_l\}$.  We use intersection ($X \cap Y$), union ($X \cup Y$),
cardinality ($|X|$), and membership ($\in X$).

\fi

\subsection{Probabilistic Set Representations}
\label{sec:prob-back}

We consider Bloom filters (BF) and two variants of MinHash (MH).
%
%
We pick different representations to better understand which ones
are best suited for accelerating graph mining problems with 
high accuracy.
\ifconf
\else
All the proofs of theorems in the following sections are 
in the appendix.
\fi
Figure~\ref{fig:pg-overview} illustrates \hl{the
representations considered}.

\marginpar{\vspace{-4em}\colorbox{yellow}{\textbf{R-5}}\\ \colorbox{yellow}{(minor}\\ \colorbox{yellow}{comm-}\\ \colorbox{yellow}{-ent 3)}}

\textbf{Bloom filters}
The Bloom filter (BF)~\cite{bloom1970space} is a space-efficient set
representation that answers \emph{membership queries} fast but with some 
probability of \emph{false positives}. 
\all{Intuitively, a BF represents a set~$X$ using (1) a bit vector of a
pre-selected length (usually much lower that $|X|$) and (2) a certain number of
hash functions. These hash functions map elements in $X$ to corresponding
bits.}
Formally, a Bloom filter~$\mathcal{B}_X$ representing a set~$X$ consists of an
$l$-element bit vector $\mathbf{B}_X$ (initialized to zeros) and $b$
hash functions $h_1, \ldots, h_b$  (usually assumed to be independent) that map
elements of $X$ to integers in $[l]$ ($[l] \equiv \{ 1, ..., l \}$).  The size
of $\mathbf{B}_X$ is also denoted with $B_X$ while the number of ones in
$\mathbf{B}_X$ with $B_{X,1}$.
Now, when {constructing} $\mathcal{B}_X$, for each element $x \in X$,
one computes the corresponding hashes $h_1(x), \ldots, h_b(x)$. Then, the bits
$\mathbf{B}_X[h_1(x)], ..., \mathbf{B}_X[h_b(x)]$ are set to one. 
Second, {verifying if} $x \in X$ is similar. First, all hash functions
are evaluated for $x$. If all bits at the corresponding positions are set,
i.e., $\forall i\in \{1, \ldots, b\} : \textbf{B}_X[h_i(x)] = 1$, then $x$ is
considered to be in $X$. It is possible that some of these bits are 
set to one while adding other elements due to hash collisions, and
$x$ might be falsely reported as being an element of $X$. 
\iftr
Minimizing the number
of such \emph{false positives} was addressed in many research
works~\cite{bose2004on}.
\fi

\all{\macb{Why Bloom Filters?}
We use BF because it is \emph{space-efficient} and it enables \emph{fast},
\emph{vectorizable}, and \emph{parallelizable} intersection of the
representations of the input sets.}


\textbf{MinHash ($k$-Hash variant)}
MinHash (MH)~\cite{broder2000} ``sketches'' a set~$X$ by hashing its elements
to integers and keeping $k$ elements with smallest hashes. An MH
representation~$\mathcal{M}_X$ of $X$ consists of a set~$M_X$ with elements
from~$X$ ($\forall_{x \in M_X}\ x \in X$) and $k$ hash functions~$h_i$, $i \in
\{1, ..., k\}$. To construct~$\mathcal{M}_X$, one computes all hashes $h_i(x)$
for each $x \in X$. Then, for \emph{each} hash function~$h_i$ separately, one
  selects an element $x_{i,min} \in X$ that has the smallest hash~$h_i(x)$.
  These elements form the final set $M_X = \{x_{1,min}, ..., x_{k,min}\}$. Note
  that $M_X$ may be a multi set: if $x_{i,min} = x_{j,min}$ for $i \neq j$,
  then $M_X$ contains $x_{i,min}$ twice.
MinHash was designed to approximate the Jaccard similarity index $J(X,Y) =
\frac{|X \cap Y|}{|X \cup Y|}$ that assesses the similarity of two sets $X$ and
$Y$.  We have $J(X,Y) \in [0;1]$; $J(X,Y) = 0$ indicates no similarity while
$J(X,Y) = 1$ means $X = Y$.

\all{If $X$ and $Y$ are disjoint then there are no elements in their
intersection and the Jaccard coefficient is zero.  If they are the same, their
intersection is the same as their union and their Jaccard coefficient is 1.

which in turn can be used to calculate the size of their intersection.  The
MinHash of a set $A$ is calculated in two steps. First, all elements in the set
are transformed by a hash function $h$, $A' = \{h(a) | a\in A\}$.  Then, the
MinHash $h_k(A)$ of $A$ is the set of elements whose hash value belongs to the
$k$ smallest values in $A'$ (see Figure \ref{fig:mh_expl} for an illustration).

After the calculations the MinHash structures of two sets can be used to
estimate their Jaccard coefficient. 

To estimate the Jaccard coefficient of two sets $A$ and $B$ using their MinHash
representation we intersect their MinHash structures $h_k(A)$ and $h_k(B)$.
The estimate of the Jaccard coefficient $J(A, B)$ is then calculated as $J(A,B)
= \frac{|h_k(A) \cap h_k(B)|}{k}$.  For a derivation see appendix
\ref{app:deriv_minhash}.  Combining this with the inclusion exclusion identity
$|A\cup B| = |A| + |B| - |A\cap B|$ we can calculate the intersection size
$|A\cap B| = \frac{J}{1+J}\cdot (|A|+|B|)$.}

\textbf{MinHash (1-Hash variant)}
The $k$-hash variant may be computationally expensive as it computes
$k$ different hash functions for all elements of~$X$. To alleviate this, one
can use a variant called 1-hash ($\mathcal{M}^1_X$)
with only one hash function $h$. After computing hashes $h(x)$ for all $x
\in X$, $M^1_X$ will contain $k$~elements from~$X$ that resulted in
$k$~smallest hash values. Unlike $k$-hash, $M_X^1$ by definition never
contains duplicates.

\all{\maciej{all good here?} the assumption is that $h(x) \in X$ for any $x \in
X$, which is not necessarily true when using a MinHash. \jakub{I don't get the
last part -- I mean, $h(x)$ is from $[0,1]$. And what is the connection to the
first part of the sentence anyway?}\cesar{I did not write this part, Maciej
could you please answer? Maybe Jakub you can rewrite the sentence or other kmv
parts as you think they are more appropriate.}}

\enlargeSQ

\subsection{Estimators}
\label{sec:estimators}
\label{sec:single-sets-BF}

A central concept in PG is an {estimator}: a rule for calculating an
estimate of a given quantity based on observed data. We develop estimators of
set sizes~$\widehat{|X|}$ and set intersection cardinalities~$\widehat{|X \cap
Y|}$, and we use them to approximate $|N_v \cap N_u|$ for any two vertices
$u,v$ in considered graph algorithms.

We describe in more detail a specific BF estimator by Swamidass et
al.~\cite{swamidass2007mathematical}. We later generalize it and also prove
novel bounds on its quality. 
Given a set $X$ represented by a Bloom filter $\mathcal{B}_X$, one can estimate
$|X|$ as

\vspaceSQ{-0.5em}
\ifsqEQ\small\fi

\begin{equation}\label{eq:est_ss_bf}
\widehat{|X|}_S = - \frac{B_X}{b} \log \left ( 1 - \frac{B_{X,1}}{B_X} \right ).
\end{equation}{}

\normalsize
\vspaceSQ{-0.6em}

To derive this estimator, consider inserting a single element into a given BF
$\mathcal{B}_X$. The probability that some single bit equals~0, in
$\mathcal{B}_X$ with one hash function (that maps elements \emph{uniformly}
over the bit array), is $1-{1}/{B_X}$. When applying $b$ hash functions (which
follow the usual assumption of being independent), this probability is
$\fRB{1-{1}/{B_X}}^b$. After inserting $|X|$ elements into $\mathcal{B}_X$, the
probability that this one bit is~0 is $\fRB{1-{1}/{B_X}}^{b |X|} \approx
\text{exp}\fRB{-b |X| / B_X}$ (from the established identity for $e^{-1}
\approx (1 - 1/x)^x$). Thus, the probability that this bit is set to~1, is $1 -
\fRB{1-{1}/{B_X}}^{b |X|} \approx 1 - \text{exp}\fRB{-b |X| / B_X}$.
Then, observe that the number of bits set to~1 in~$\mathcal{B}_X$, can be
estimated as $B_{X,1} \approx B_X \cdot \fRB{1 - \text{exp}\fRB{-b |X| / B_X}}$
given the binomial density approximation used in Swamidass et
al.~\cite{swamidass2007mathematical}. When resolving this equation for~$|X|$,
we obtain Eq.~(\ref{eq:est_ss_bf}).

\enlargeSQ
\enlargeSQ

\marginpar{\vspace{-25em}\colorbox{yellow}{\textbf{ALL}}}

\subsection{Properties of Estimators}
\label{sec:back-properties}

\marginpar{\vspace{1em}\colorbox{yellow}{\textbf{ALL}}}

Figure~\ref{fig:estimators} \hl{provides
an intuitive description of the desirable statistical properties of PG}.
These properties enable highly accuracy empirical results 
(Section~\ref{sec:eval}) and attractive theoretical 
results for triangle count (Section~\ref{sec:theory}).

\ifconf
\begin{figure}[t]
\else
\begin{figure*}[t]
\fi
\vspaceSQ{-1em}
\centering
\includegraphics[width=1.0\columnwidth]{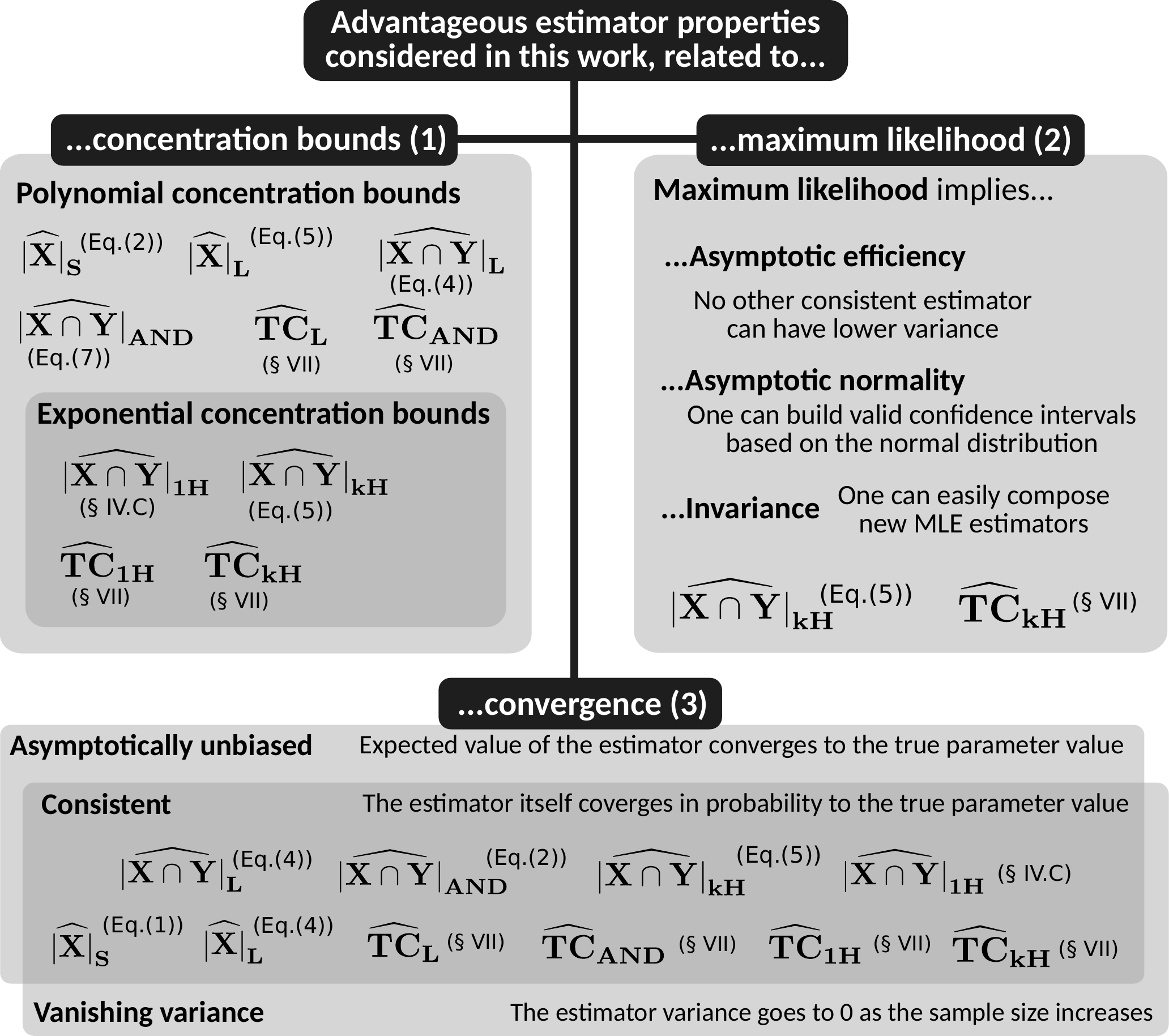}
\vspaceSQ{-1.5em}
\caption{\textmd{Considered estimator properties and the associated
PG estimators.
\hl{The advantageous properties of ProbGraph estimators belong to three classes
of properties: \textbf{(1) having concentration bounds}, \textbf{(2) being
maximum likelihood}, and \textbf{(3) convergence}.
\mbox{\underline{\textbf{Intuitively}}}, \mbox{``(1)''} means that we can bound the probability that a PG estimator deviates from the true parameter value for every sketch size.
\mbox{``(2)''} means that a PG estimator identifies the parameter value which is more likely to have produced the data we observed. \mbox{``(3)''} means that a PG estimator converges to the true parameter
value as the sketch size increases. In the figure, we list different desirable properties implied by
\mbox{(1)}, \mbox{(2)}, and \mbox{(3)}. Formal definitions of the considered
properties can be found in Section~\mbox{\ref{sec:fund}}.}
}}
\label{fig:estimators}
\vspaceSQ{-1em}
\ifconf
\end{figure}
\else
\end{figure*}
\fi


All PG estimators are \textbf{asymptotically unbiased}. In such an
estimator~$\widehat{\theta}$, the difference between $\widehat{\theta}$'s
expected value and the true value of the parameter being estimated~$\theta$
converges to 0 for a fixed input and the size of the sketch that we use going to infinity
(i.e., the bias of $\widehat{\theta}$ goes to 0, or, on average,
$\widehat{\theta}$ hits~$\theta$ when the sample size approaches the limit). Unbiased
estimators are usually more desirable than biased ones: intuitively,
they ensure zero average error (when estimating~$\theta$) after a
given amount of trials.
Next, each PG estimator of $\widehat{\theta}$ is also \textbf{consistent},
i.e, the sampling distribution
of~$\widehat{\theta}$ becomes increasingly more concentrated at~$\theta$ with
the increasing number of samples. Hence, if there are enough
observations (in our case when the sketches are large enough),
one can find~$\theta$ with arbitrary precision. Asymptotic unbiasedness alone does not imply consistency; it requires also a vanishing
variance (i.e., that the estimator variance converges to 0 with the increasing
sample size).  
%

We also verify if PG estimators are \textbf{maximum likelihood}. This class of estimators provides several powerful and useful properties, and is among the most important tools for estimation and inference in statistics~\cite{myung2003tutorial, fisher1922mathematical, white1982maximum}.
%
%
Specifically, a maximum likelihood estimator (MLE)~$\widehat{\theta}_{MLE}$ is
an estimator that maximizes the \emph{likelihood function}~$L$, i.e.,
$\widehat{\theta}_{MLE} = \text{argmax}_{\theta \in \Theta} \; L(\mathbf{x};
\theta)$ (cf.,~Chapter~7 in~\cite{casella2002statistical}) where $\Theta$ is
the parameter space. 
\marginparsep=1em
\marginpar{\vspace{-15em}\colorbox{yellow}{\textbf{ALL}}}
%
%
Here, the likelihood function $L$ is defined as the probability of observing a
given sample $\mathbf{x} = (x_1, ..., x_n)$ as a function of~$\theta$, i.e., $L
\equiv P(X_1 = x_1, ..., X_n = x_n; \theta)$, where $X_1, ..., X_n$ represent a
random sample from a given population. Thus, $\widehat{\theta}_{MLE}$ is the
value of the parameter $\theta$ for which the observed sample is the most
likely.  This intuitive choice of an estimator leads under mild conditions to
strong optimality properties such as {consistency} (discussed above),
\enlargeSQ
\enlargeSQ
\textbf{invariance}, and \textbf{asymptotic efficiency}. An
estimator~$\widehat{\theta}$ is invariant if, whenever $\widehat{\theta}$ is
the MLE of $\theta$, then for any function $\tau(\cdot)$,
$\tau(\widehat{\theta})$ is the MLE of $\tau(\theta)$.  This is useful if
complicated functions of the parameter $\theta$ are of interest since
$\tau(\widehat{\theta}_{MLE})$ inherits automatically all the properties of
$\widehat{\theta}_{MLE}$. Finally, the asymptotic efficiency certifies that
$\widehat{\theta}_{MLE}$ attains, under mild regularity conditions, the
Cramer-Rao bound~\cite{casella2002statistical}.  To understand the importance
of this result, we need to introduce the Mean Squared Error (MSE) of an
estimator. MSE measures the average squared difference between the estimator
and the parameter, i.e., $MSE(\widehat{\theta},\theta) =
E_{\theta}[(\widehat{\theta} - \theta)^2]$).  Thus, the asymptotic efficiency
implies that there exists no consistent estimator of $\theta$ that achieves a
lower MSE than $\widehat{\theta}_{MLE}$. Next, given the well known
decomposition of the MSE into bias squared plus variance (i.e.,
$MSE(\widehat{\theta}) = Bias(\widehat{\theta},\theta)^{2} +
Var_{\theta}(\widehat{\theta})$), we can also conclude that among all the
consistent estimators of $\theta$, no one has a strictly smaller variance than
$\widehat{\theta}_{MLE}$. Another interesting implication of this result is the
\textbf{asymptotic normality} of MLE. This property allows to build valid
\emph{confidence intervals} for the parameter of interest based on the normal
distribution (asymptotically in sketch size, for a fixed input). 

\vspaceSQ{-0.2em}
\subsection{Concentration Bounds}
\vspaceSQ{-0.2em}

\marginparsep=2em
\marginpar{\vspace{-2em}\colorbox{yellow}{\textbf{R-5}}\\ \colorbox{yellow}{(minor}\\ \colorbox{yellow}{comm-}\\ \colorbox{yellow}{-ent 4)}}

\hl{We use} the notion of a \emph{concentration inequality}. Overall, such
an inequality \emph{bounds the deviation of a given random variable~$X$ from
some value} (usually the expectation~$E[X]$). 
\iftr
In this work, we mainly use the Chebyshev~\cite{boucheron2013concentration},
Hoeffding~\cite{bentkus2004hoeffding} and Serfling inequalities~\cite{greene2017exponential}.
\fi


%

%
%


%
%


\section{Sets \& $|X \cap Y|$ In Graph Algorithms}
\label{sec:algs}

We first identify algorithms that
use $|X \cap Y|$.
A graph itself can be modeled as a collection of sets: each vertex neighborhood
$N_v$ is essentially a set. In PG, we use this observation to
approximate the graph structure and operations by using probabilistic set representations in
place of $N_v$ and $|N_v \cap N_u|$, for any vertices $u$ and $v$.
In the following listings, the ``$X$'' and ``$Y$'' general symbols are replaced
with specific sets. Operations approximated by PG are marked with the
~~\tikzmarkin[set fill color=hlL, set border color=white, above offset=0.27,
right offset=5em, left offset=-0.1em, below
offset=-0.1]{mot1}\textcolor{black}{blue}\tikzmarkend{mot1}~ color.
``\texttt{[in par]}'' means that a loop is parallelized. We ensure
that the parallelization does not involve conflicting memory accesses. For
clarity, we focus on \emph{formulations} and we discuss parallelization details
in Section~\ref{sec:design}.

In \textbf{Triangle Counting (TC)}~\cite{shun2015multicore, strausz2022asynchronous}
(Listing~\ref{lst:tc}), one counts the total number of triangles $tc$ in an
undirected graph. An example application of TC is computing clustering
coefficients~\cite{al2018triangle}.
For each vertex~$u$, one computes the cardinalities of the intersections of
$N_u$, the set of neighbors of $u$, with the sets of the neighbors of each
neighbor of~$u$ (Lines~\ref{ln:tc-main-1}-\ref{ln:tc-main-2}).


\begin{lstlisting}[float=h, aboveskip=0em,belowskip=-0.5em,abovecaptionskip=0.0em,label=lst:tc,caption=Triangle Counting (Node Iterator).]
|\vspace{0.5em}|/* |\textbf{Input:}| A graph $G$. |\textbf{Output:}| Triangle count $tc \in \mathbb{N}$. */
//Derive a vertex order $R$ s.t. if $R(v) < R(u)$ then $d_v \leq d_u$:
for $v \in V$ [in par] do: $N^+_v = \{ u \in N_v | R(v) < R(u) \}$
|\label{ln:tc-main-1}|$tc$ = $0$; //Init $tc$; for all neighbor pairs, increase $tc$:
//Now, derive the actual count of triangles: 
|\label{ln:tc:s}|$v \in V$ [in par] do:
|\vspace{0.25em}|  |\label{ln:tc-main-2}|for $u \in N^+_v$ [in par] do: $tc$ += |\hlLR{8.5em}{ $\mid N^+_v \cap N^+_u\mid$ }| 
\end{lstlisting}

\marginparsep=2em
\marginpar{\vspace{-3em}\colorbox{yellow}{\textbf{R-5}}\\ \colorbox{yellow}{(minor}\\ \colorbox{yellow}{comm-}\\ \colorbox{yellow}{-ent 5)}}

We also consider \hl{higher-order} \textbf{Clique Counting}, a problem important for dense
subgraph discovery~\cite{danisch2018listing}. Listing~\ref{lst:fcl} contains
4-clique listing. We reformulated the original scheme~\cite{danisch2018listing}
(without changing its time complexity) to expose $|X \cap Y|$. The
\hl{algorithm presented} generalizes TC. 

\marginparsep=2em
\marginpar{\vspace{-2em}\colorbox{yellow}{\textbf{R-5}}\\ \colorbox{yellow}{(minor}\\ \colorbox{yellow}{comm-}\\ \colorbox{yellow}{-ent 6)}}

\begin{lstlisting}[float=h,belowskip=-0.5em,aboveskip=0em,label=lst:fcl,caption=Reformulated 4-Clique Counting.]
|\vspace{0.5em}|/* |\textbf{Input:}| A graph $G$. |\textbf{Output:}| Number of 4-cliques $ck \in \mathbb{N}$. */
/Derive a vertex order $R$ s.t. if $R(v) < R(u)$ then $d_v \leq d_u$:
for $v \in V$ [in par] do: $N^+_v = \{ u \in N_v | R(v) < R(u) \}$
$ck$ = $0$;
for $u \in V$ [in par] do:
  for $v \in N^+_u$ [in par] do:
    $C_3$ = $N^+_u \cap N^+_v$ //Find 3-cliques
|\vspace{0.25em}|    for $w \in C_3$ do: //For each 3-clique...
|\vspace{0.25em}|      $ck$ += |\hlLR{7em}{ $\vert N^+_w \cap C_3 \vert$ }| //Find 4-cliques
\end{lstlisting}

\enlargeSQ

\textbf{Vertex Similarity} measures, used in graph databases and
others~\cite{robinson2013graph, robinson2015graph, lissandrini2017evaluation,
neo4j_book, besta2019demystifying}, assess how similar two vertices $v$ and $u$
are, see Listing~\ref{lst:sim}. They can be used on their own, or as a building
block of more complex algorithms such as clustering. Many of these schemes use
the cardinality of set intersection. This includes Jaccard, Common Neighbors,
Total Neighbors, or Adamic Adar.
Vertex Similarity is the basic building block of \textbf{Link Prediction}
and {Clustering}.

\begin{lstlisting}[float=h!, aboveskip=0em,belowskip=-0.5em,abovecaptionskip=0em,label=lst:sim,caption=Example vertex similarity measures~\cite{leicht2006vertex}.]
/* |\textbf{Input:}| A graph $G$. |\textbf{Output:}| Similarity $S \in \mathbb{R}$ of sets $A,B$.
* Most often, $A$ and $B$ are neighborhoods $N_u$ and $N_v$
|\vspace{0.5em}| * of vertices $u$ and $v$. */
|\vspace{0.5em}|//Jaccard similarity:
|\vspace{0.5em}|$S_J(A,B)$ = |\hlLR{5.5em}{ $\vert A \cap B\vert$ }| / |\hlLR{5.5em}{ $\vert A \cup B\vert$ }| = |\hlLR{5.5em}{ $\vert A \cap B\vert$ }| / ($\vert A \vert$ + $\vert B \vert$ - |\hlLR{5.2em}{ $\vert A \cap B \vert$ }|)
|\vspace{0.5em}|//Overlap similarity:
|\vspace{0.5em}|$S_O(A,B)$ = |\hlLR{5.3em}{ $\vert A \cap B\vert$ }| / min($\vert A\vert$, $\vert B\vert$)
|\vspace{0.5em}|//Certain measures are only defined for neighborhoods:
|\vspace{0.5em}|$S_A(v,u)$ = $\sum_{w} (1 / \log|N_w|)$ //where $w \in $ |\hlLR{6em}{ $N_v \cap N_u$ }|; Adamic Adar
|\vspace{0.5em}|$S_R(v,u)$ = $\sum_{w} (1 / |N_w|)$ //where $w \in $ |\hlLR{6em}{ $N_v \cap N_u$ }|; Resource Alloc.
|\vspace{0.5em}|$S_C(v,u)$ = |\hlLR{7em}{ $\vert N_v \cap N_u \vert$ }| //Common Neighbors
$S_T(v,u)$ = $\vert N_v \cup N_u \vert$ = $|N_v|$ + $|N_u|$ - |\hlLR{7em}{ $\vert N_v \cap N_u \vert$ }| //Total Neighbors
\end{lstlisting}

\marginparsep=1em
\marginpar{\vspace{-2em}\colorbox{yellow}{\textbf{R-5}}\\ \colorbox{yellow}{(minor}\\ \colorbox{yellow}{comm-}\\ \colorbox{yellow}{-ent 7)}}

\textbf{Graph Clustering} \hl{is} a broadly studied
problem~\cite{schaeffer2007graph}. Listing~\ref{lst:cl} shows Jarvis-Patrick
clustering~\cite{jarvis1973clustering}, a scheme that uses vertex similarity to
determine whether these two vertices are in the same cluster, and relies
heavily on $|X \cap Y|$.


\begin{lstlisting}[float=h, aboveskip=0em, belowskip=-0.5em,label=lst:cl,caption=Jarvis-Patrick clustering.]
/* |\textbf{Input:}| A graph $G = (V,E)$. |\textbf{Output:}| Clustering $C \subseteq E$
|\vspace{0.5em}| * of a given prediction scheme. */
|\vspace{0.5em}|//Use a similarity $S_C(v,u)$ = |\hlLR{7em}{ $\vert N_v \cap N_u\vert$ }| (see Listing |\ref{lst:sim}|).
|\vspace{0.5em}|for $e = (v,u) \in E$ [in par] do: //$\tau$ is a user-defined threshold
  if |\hlLR{7em}{ $\vert N_v \cap N_u\vert$ }| $> \tau$: $C$ $\cup$= $\{e\}$ 
//Other clustering schemes use other similarity measures.
\end{lstlisting}

\iftr

\textbf{Link Prediction}
There are many schemes for predicting whether two non-adjacent vertices can
become connected in the future in the context of evolving
networks~\cite{lu2011link}. Assessing the accuracy of a specific link
prediction scheme~$S$ is done with a simple algorithm~\cite{wang2014robustness}
shown in Listing~\ref{lst:lp}.
We start with some graph with \emph{known} links (edges). We derive $E_{sparse}
\subseteq E$, which is $E$ with random links removed; $E_{sparse} = E \setminus
E_{rndm}$.  $E_{rndm} \subseteq E$ are randomly selected \emph{missing} links
from $E$ (\emph{links to be predicted}).  We have $E_{sparse} \cup E_{rndm} =
E$ and $E_{sparse} \cap E_{rndm} = \emptyset$.
Now, we apply the link prediction scheme~$S$ (that we want to test) to each
edge $e \in (V \times V) \setminus E_{sparse}$. The higher a value $S(e)$, the
more probable $e$ is to appear in the future (according to $S$).  Now, the
effectiveness $ef$ of $S$ is computed by verifying how many of the edges with
highest prediction scores ($E_{predict}$) actually are present in the original
dataset~$E$: $ef = |E_{predict} \cap E_{rndm}|$.

\begin{lstlisting}[float=h, aboveskip=0em, belowskip=-0.5em,label=lst:lp,caption=Link prediction testing.]
/* |\textbf{Input:}| A graph $G = (V,E)$. |\textbf{Output:}| Effectiveness $ef$
|\vspace{0.5em}| * of a given prediction scheme. */
$E_{rndm}$ = /* Random subset of $E$ */
|\vspace{0.5em}|$E_{sparse}$ = $E \setminus E_{rndm}$ /* Edges in $E$ after removing $E_{rndm}$ */
//For each $e \in (V \times V) \setminus E_{sparse}$, derive score $S(e)$ that
//determines the chance that $e$ appears in future. Here,
|\vspace{0.25em}|//one can use any vertex similarity scheme $S$.
|\vspace{0.25em}|for $e = (v,u) \in (V \times V) \setminus E_{sparse}$ [in par] do: compute $S(v,u)$
|\vspace{0.25em}|$E_{predict}$ = /* Pick selected top edges with highest $S$ scores.*/
$ef$ = |\hlLR{11em}{ $\vert E_{predict} \cap E_{rndm}\vert$ }| //Derive the effectiveness.
\end{lstlisting}

\fi

\marginparsep=1em
\marginpar{\vspace{2em}\colorbox{yellow}{\textbf{R-3}}\\ \colorbox{yellow}{comm-}\\ \colorbox{yellow}{-ent 1}}

\subsection{\hl{Real-World Applications}}

\enlargeSQ
\enlargeSQ

\hl{Graph problems targeted by ProbGraph have numerous real-world
applications because the underlying
operation~$|X \cap Y|$, used to find the counts of the shared 
neighbors, is a common building block in many real-world problems
in domains ranging from network science or sociology, through
chemistry or biology, to the Internet
studies~\mbox{\cite{schaeffer2007graph}}.}

\hl{Triangle counting is used to obtain the \emph{network cohesion}, an
important measure of connectedness and ``togetherness'' of a group of
vertices~\mbox{\cite{friggeri2011triangles, prat2012shaping}}. Specifically,
for any subgraph $S \subseteq V$ , $S$'s cohesion is $TC[S] / \binom{|S|}{3}$,
  where $TC[S]$ is the triangle count of $S$; note that $S$ may also form $V$
  (in which case we obtain the cohesion of the whole graph).
Another example is discovering communities~\mbox{\cite{becchetti2008efficient,
palla2005uncovering}}, by computing the clustering coefficient
defined as $3 \cdot TC[S] / \binom{|S|}{3}$
Other use cases include \emph{spam detection} (standard and spam sites differ
in the respective counts of triangles that they belong to), optimization of
query planning in databases~\mbox{\cite{bar2002reductions}}, uncovering hidden
thematic layers in WWW~\mbox{\cite{eckmann2002curvature}}, or studying
differences between gene interactomes of various
species~\mbox{\cite{tran2013counting}}.}

\hl{The considered Jarvis-Patrick clustering can be used in adaptive web search
based on automatic construction of user profiles. A critical step in this use
case is generation of clusters of users, which is directly achieved using the
clustering scheme addressed in PG~\mbox{\cite{sugiyama2004adaptive}}. Other
selected examples are drug design (by predicting plasma protein
bindings~\mbox{\cite{kratochwil2002predicting}}), screening and generating
overviews of chemical databases (by computing clusters of related
molecules)~\mbox{\cite{willett1998chemical}}, or analyzing single-cell RNA
sequences (by approximating smooth low-dimensional surfaces that model states
of cells)~\mbox{\cite{kharchenko2021triumphs, tyser2021single}}.}

\marginparsep=2em
\marginpar{\vspace{0em}\colorbox{yellow}{\textbf{R-3}}\\ \colorbox{yellow}{comm-}\\ \colorbox{yellow}{-ent 1}}

\hl{Other considered problems also have numerous applications. 
In short, clique counting is used in social network analysis (cf.~the
established textbooks~\mbox{\cite[Chapter 11]{hanneman2005introduction}}
and~\mbox{\cite[Chapter 2]{jackson2010social}}) to find large and dense network
regions~\mbox{\cite{lu2018community, mitzenmacher2015scalable,
sariyuce2015finding, tsourakakis2015k, tsourakakis2017scalable}} or in
topological approaches to network
analysis~\mbox{\cite{sizemore2017classification}}. 
Link prediction and vertex similarity are used throughout the whole graph data
mining in many parts of network science and others, as illustrated in numerous
surveys and textbooks~\mbox{\cite{liben2007link, lu2011link, al2006link,
taskar2004link, cook2006mining, jiang2013survey, besta2021graphminesuite}}.}

\section{Approximating $|X \cap Y|$}
\label{sec:reps_int}


We now show how to derive approximate set intersection \emph{cardinality}~$|X
\cap Y|$ both \emph{fast} and \emph{with high accuracy}. 
\ifconf
\else
We provide selected results for BF and MH (less competitive outcomes are
in the Appendix).
\fi
\all{As in Section~\ref{sec:single_sets}, we just state key results and include
derivations and all outcomes in the Appendix or extended report.}
%
%
In this section, we assume arbitrary sets $X$ and $Y$, to ensure that our
outcomes are of interest beyond graph mining. From Section~\ref{sec:theory}
onwards, we focus on graph mining by applying the results from this section to
$\widehat{|N_u \cap N_v|}$.

\marginparsep=2em
\marginpar{\vspace{2em}\colorbox{yellow}{\textbf{ALL}}}

\subsection{\hl{Section Overview and Intuition}}

\hl{We first outline the section structure and provide the intuition behind the
key parts. 
We first present estimators for $|X \cap Y|$, designed using BF
(\mbox{\cref{sec:int_bf_est}}), $k$-Hash (\mbox{\cref{sec:int_kh_est}}), and
1-Hash (\mbox{\cref{sec:int-1h}}).
Then, we compare the obtained estimators regarding their accuracy
(\mbox{\cref{sec:intersect_summary-acc}}) and the amount of parallelism
(\mbox{\cref{sec:intersect_summary-par}}).}

\hl{We provide concentration bounds for all the estimators. We present here
selected ones, the others are in the supplementary material together with the proofs of all the propositions of this section. A generic form of a concentration bound from this section is $P(|\text{estimator} - \text{true\_value}| \ge t)
\le f(t)$. Intuitively, this means that we can bound (i.e., by the function $f(t)$) the probability that a PG estimator deviates (i.e., more than $t$) from the true parameter value for every sketch size. The function $f(t)$ (either polynomial or exponential in $t$ for all PG estimators) determines the speed at which a given PG estimator concentrates around the true parameter value.}   


\subsection{Approximating $|X \cap Y|$ with Bloom Filters}
\label{sec:int_bf_est}

We introduce a new estimator $\widehat{|X \cap Y|}_{AND}$ and we give a bound
on its accuracy. 
Specifically, for two sets $X$ and $Y$ represented by $\mathcal{B}_X$ and
$\mathcal{B}_Y$, we apply the estimator from Eq.~(\ref{eq:est_ss_bf}) to
$\mathcal{B}_{X \cap Y}$, obtaining

\enlargeSQ

\vspaceSQ{-0.5em}
\begin{equation}\label{eq:bf_int}
\widehat{|X \cap Y|}_{AND} = - \frac{B_{X \cap Y}}{b} \log \left( 1 - \frac{B_{X \cap
Y,1}}{B_{X \cap Y}} \right)    
\end{equation}
\vspaceSQ{-0.5em}
\normalsize

\noindent
where $B_{X \cap Y} = B_X = B_Y$ is the BF size (cf.~Table~\ref{tab:symbols}). 
%
%
Next, we prove an important property of~$\widehat{|X \cap Y|}_{AND}$. Note that
the following property also holds for the estimator by
Swamidass~\cite{swamidass2007mathematical} from Eq.~(\ref{eq:est_ss_bf}).

\begin{prop}\label{bound_int_BF} Let $\widehat{|X \cap Y|}_{AND}$ be the
estimator defined in Eq.~(\ref{eq:bf_int}). For $B_{X \cap Y}, b \in
\mathbb{N}$ such that $b = o(\sqrt{B_{X \cap Y}})$, and a set $X \cap Y$ such
that $b |X \cap Y| \leq 0.499 B_{X \cap Y} \cdot \log B_{X \cap Y}$ the
following holds: 
	
\vspaceSQ{-1.5em}
\small
\begin{align}
E\fSB{\fRB{\widehat{|X \cap Y|}_{AND} - |X \cap Y|}^2} & \leq \nonumber \\ 
(1+o(1))\left(e^{|X \cap Y| b / (B_{X \cap Y}-1)} \frac{B_{X \cap Y}}{b^2} - \frac{B_{X \cap Y}}{b^2} - \frac{|X \cap Y|}{b}\right) & \nonumber
\end{align}
\normalsize
\end{prop}

\noindent
%
Overall, Proposition~\ref{bound_int_BF} shows that we can bound the mean
squared error (MSE) of $\widehat{|X \cap Y|}_{AND}$ (and also
$\widehat{|X|}_{S}$ from Eq.~(\ref{eq:est_ss_bf})). By Chebyshev's
inequality\footnote{\scriptsize We apply the inequality on the MSE to derive a
bound for $P \left( \left| \widehat{|X \cap Y|}_{AND} - |X \cap Y| \right| \geq
t \right)$ rather than for $P \left( \left| \widehat{|X \cap Y|}_{AND} -
E(\widehat{|X \cap Y|}_{AND}) \right| \geq t \right)$ (as is usually done).},
we obtain the following concentration result:

\vspaceSQ{-1.5em}
\small
\begin{align}\label{bound_int_BF-eq}
P\fRB{\fVB{\widehat{|X \cap Y|}_{AND} - |X \cap Y|} \ge t} & \leq \nonumber \\
(1+o(1))\frac{\left(e^{|X \cap Y| b / (B_{X \cap Y}-1)} \frac{B_{X \cap Y}}{b^2} - \frac{B_{X \cap Y}}{b^2} - \frac{|X \cap Y|}{b}\right)}{t^{2}} &
\end{align}
\normalsize

\marginparsep=1em
\marginpar{\vspace{1em}\colorbox{yellow}{\textbf{R-3}}\\ \colorbox{yellow}{\textbf{R-5}}}

\hl{We can strengthen the intuition on the behavior of $\widehat{|X \cap Y|}_{AND}$ by taking the limit for $B_{X \cap Y} \rightarrow \infty$ in Eq.~}(\ref{eq:bf_int})\hl{. We call $\widehat{|X \cap Y|}_L$ this limiting estimator:}


\vspaceSQ{-0.5em}
\small

\if0

\begin{align}
  \widehat{|X \cap Y|}_L &\equiv \lim_{B_{X \cap Y}\to\infty} \widehat{|X \cap Y|}_{AND} \nonumber \\ &= \lim_{B_{X \cap Y}\to\infty}- \frac{B_{X \cap Y}}{b} \log \left ( 1 - \frac{B_{X \cap Y,1}}{B_{X \cap Y}} \right ) \nonumber \\
    &= \log \left ( \lim_{B_{X \cap Y}\to\infty} \left(1 - \frac{B_{X \cap Y,1}}{B_{X \cap Y}}\right)^{-\frac{B_{X \cap Y}}{b}} \right) \nonumber\\
      &= \log\left(\exp \left(\frac{B_{X \cap Y,1}}{b}\right)\right) =\frac{B_{X \cap Y,1}}{b} \label{eq:limit_xs}
      \end{align}

\else

\begin{align}
\widehat{|X \cap Y|}_L &\equiv \lim_{B_{X \cap Y}\to\infty} \widehat{|X \cap Y|}_{AND} =  \frac{B_{X \cap Y,1}}{b} \label{eq:limit_xs}
      \end{align}

\fi

\normalsize

Hence, as $B_{X \cap Y}$ increases, $\widehat{|X \cap Y|}_{AND}$ \emph{rescales
the number of ones in the BF} by $\frac{1}{b}$ because
$\widehat{|X \cap Y|}_{AND} \sim \frac{B_{X \cap Y,1}}{b}$ for $X,Y,b$ fixed
and $B_{X \cap Y} \rightarrow \infty$. 
\hl{In Section~\mbox{\ref{sec:eval}}, we will show that -- depending on the
choice of the scaling factor $\frac{1}{b}$ which impacts the bias-variance
trade-off} -- there are cases where $\widehat{|X \cap Y|}_{L}$ is better than
$\widehat{|X \cap Y|}_{AND}$.
%

Note that $\widehat{|X \cap Y|}_{AND}$ uses the count of ones in
a BF $\mathbf{B}_{X \cap Y}$. This number cannot be computed from individual
BFs $\mathbf{B}_X$ and $\mathbf{B}_Y$. In practice, we use $\mathbf{B}_{X
\cap Y} \approx \mathbf{B}_X \text{ AND } \mathbf{B}_Y$ (where ``AND''
indicates a logical bitwise AND operation) and use the result of AND to obtain
$B_{X \cap Y, 1}$.
This may somewhat increase the false positive probability, but -- as the results in
Section~\ref{sec:eval} show -- does not prevent high accuracy.
%
%

\marginparsep=1em
\marginpar{\vspace{-6em}\colorbox{yellow}{\textbf{R-3}}}

\marginparsep=1em

\enlargeSQ

\subsection{Approximating $|X \cap Y|$ with $k$-Hash}
\label{sec:int_kh_est}

To estimate $|X \cap Y|$ with MinHash, one first uses the definition of the
Jaccard similarity index $J_{X,Y} = {|X \cap Y|}/{|X \cup Y|}$
(cf.~\cref{sec:prob-back}) and, together with the well-known set algebraic
expression $|X \cup Y| = |X| + |Y| - |X \cap Y|$, rewrites it to obtain the
following estimator: 

\vspaceSQ{-0.5em}
\begin{equation}\label{est_int_mh}
\widehat{|X \cap Y|}_{kH} = \frac{\widehat{J_{X,Y}}_{kH}}{1 + \widehat{J_{X,Y}}_{kH}}(|X|+|Y|)
\end{equation}

where $\widehat{J_{X,Y}}_{kH} = \frac{|M_X \cap M_Y|}{k}$ is itself an unbiased
estimator of $J_{X,Y}$ (see \cite{ertl2018bagminhash} for a proof). 
If we assume that the $k$ hash functions are independent and perfectly random (a
usual assumption\tr{\footnote{\scriptsize To satisfy this assumption, one could just store
perfectly random permutations on the set of vertices without increasing
asymptotic complexity.}}), we have $|M_X \cap M_Y| \sim Bin (k,J_{X,Y})$, i.e.,
$|M_X \cap M_Y|$ follows the binomial distribution, where the number of trials
equals the number of hash functions $k$ and the probability of success is the
true Jaccard coefficient (this is valid by the construction of $k$-Hash,
see~\cite{broder2000min}). Thus, we can derive the expectation and the variance
of $\widehat{|X \cap Y|}_{kH}$ adapting the formulas for the moments of a binomial
random variable
\ifconf
\else
(provided in the Appendix).
\fi

We develop the following concentration bound ({this is the first exponential bound for $\widehat{|X \cap Y|}_{kH}$}):
%

\begin{prop}\label{bound_int_mh}
Let $\widehat{|X \cap Y|}_{kH}$ be the estimator from Eq.~(\ref{est_int_mh}).
Then, an upper bound for the probability of deviation from the true~$|X \cap Y|$,
at a given distance $t \geq 0$, is:

\vspaceSQ{-0.5em}
\begin{equation}\label{eq:mh_int_conc_union}
P \left( \left| \widehat{|X \cap Y|}_{kH} - |X \cap Y| \right| \geq t \right) \leq 2 e^{-\frac{2\;k\;t^2}{(|X| + |Y|)^2}}
\end{equation}
\vspaceSQ{-0.5em}
\end{prop}
\vspaceSQ{-0.5em}


We stress that $\widehat{|X \cap Y|}_{kH}$ derived with $k$-hash
can also be interpreted as a \emph{maximum likelihood estimator (MLE)}
(cf.~\cref{sec:estimators}) for $|X \cap Y|$ because of the invariance property
outlined in~\cref{sec:back-properties} (details are provided together with the
proof). Thus, our estimator inherits all the properties of MLE such as
consistency and asymptotic efficiency. 
Moreover, the bound is \emph{exponential}, i.e., the distance between the
estimator and the true value of $|X \cap Y|$ decreases exponentially.
Finally, we stress that $\widehat{|X \cap Y|}_{kH}$ \emph{is
asymptotically efficient}, i.e., no other estimator (using only this sketch)
can have lower variance (for a fixed input and asymptotically for $k \rightarrow \infty$).



\subsection{Approximating $|X \cap Y|$ with $1$-Hash}
\label{sec:int-1h}




The 1-Hash estimator is similar to $k$-Hash in that we first estimate the
Jaccard similarity itself as $\widehat{J_{X,Y}}_{1H} = \frac{|M^1_X \cap
M^1_Y|}{k}$.
Similarly to the estimator used in $k$-Hash, this is itself an unbiased
estimator of $J_{X,Y}$.
Then, we use it to estimate $|X \cap Y|$: $\widehat{|X \cap Y|}_{1H} =
\frac{\widehat{J_{X, Y}}_{1H}}{1 + \widehat{J_{X, Y}}_{1H}}(|X|+|Y|)$.
\iftr

\else
\fi
\iftr
Recall that the 1-Hash representation of~$X$ differs qualitatively from the
$k$-Hash variant in that (1) $M_X^1$ does not contain duplicates, and (2)
$\mathcal{M}_X^1$ uses only one hash function.
The $k$ elements maintained in a 1-Hash are not independent, as we are in a
\textit{sampling without replacement} scheme\footnote{\scriptsize Contrarily,
$k$-Hash is a \textit{sampling with replacement scheme} and explains why $|M_X
\cap M_Y| \sim Bin (k,J_{X,Y})$ for $k$-Hash}. This also means that $k$-Hash
can have duplicates, which is not possible with 1-Hash. Thus, $| M_X^1 \cap
M_Y^1| \sim Hypergeometric (|X \cup Y|,|X \cap Y|,k)$ where $|X \cup Y|$ is the
population size, $|X \cap Y|$ is the number of success states in the
population, and $k$ is the number of draws. This implies that we can derive the
expectation and the variance of $\widehat{|X \cap Y|}_{1H}$ by adapting the
formulas for the moments of an hypergeometric random variable. We provide the
formulas for the expectation in  
the Appendix. 
\fi
We now provide the same concentration bound as in the case of
$k$-Hash.

\vspaceSQ{-0.5em}
\begin{prop}\label{bound_int_one_h}
Consider $\widehat{|X \cap Y|}_{1H}$. Then, an upper bound for the probability of deviation from the true intersection set size, at a given distance $t \geq 0$, is:

\vspaceSQ{-0.5em}
\vspaceSQ{-0.5em}
\begin{equation}\label{eq:m1h_int_conc_union}
P \left(\fVB{\widehat{|X \cap Y|}_{1H} - |X \cap Y|} \geq t \right) \leq 2 e^{-\frac{2\;k\;t^2}{(|X| + |Y|)^2}}
\end{equation}
\vspaceSQ{-0.5em}
\vspaceSQ{-0.5em}

\end{prop}

\marginparsep=2em
\marginpar{\vspace{0em}\colorbox{yellow}{\textbf{R-5}}\\ \colorbox{yellow}{(minor}\\ \colorbox{yellow}{comm-}\\ \colorbox{yellow}{-ent 10)}}

The bound suggests that \hl{1-Hash can be better in practice than $k$-Hash. They both have exponential bounds but 1-Hash requires hashing elements using only
\emph{one} hash function. Thus, it is faster to compute}.
%

\enlargeSQ

\marginparsep=2em
\marginpar{\vspace{-1em}\colorbox{yellow}{\textbf{ALL}}}

\subsection{\hl{Analysis of Accuracy of $\widehat{|X \cap Y|}$}}
\label{sec:intersect_summary-acc}

\marginparsep=2em
\marginpar{\vspace{-1em}\colorbox{yellow}{\textbf{R-5}}\\ \colorbox{yellow}{(minor}\\ \colorbox{yellow}{comm-}\\ \colorbox{yellow}{-ent 11)}}

We summarize \hl{our} theory developments into estimating $|X \cap Y|$ in
Table~\ref{tab:estimators-summary} (estimators) and
Table~\ref{tab:bounds-summary} (bounds). These results are also applicable to
general estimators of $|X|$
(cf.~\cref{sec:estimators}) and we also show them in the table.
We provide deviation bounds for all PG estimators. The estimator for $k$-hash
is an MLE. Moreover, the $k$-Hash and $1$-Hash
bounds are exponential. This means that the estimates are unlikely to deviate much from the true value.  
%

\begin{table}[h]
\vspaceSQ{-0.5em}
\centering
\setlength{\tabcolsep}{5pt}
\scriptsize
%
\begin{tabular}{@{}llllllll@{}}
\toprule
\makecell[c]{\textbf{Result}} &
\makecell[c]{\textbf{Where}} &
\makecell[c]{\textbf{Class}} &
\makecell[c]{\textbf{AU}} &
\makecell[c]{\textbf{CN}} &
\makecell[c]{\textbf{ML}} &
\makecell[c]{\textbf{IN}} &
\makecell[c]{\textbf{AE}} \\
\midrule
$\widehat{|X|}_S$ & Eq.~(\ref{eq:est_ss_bf}) & BF & \faThumbsOUp\ \faStar & \faThumbsOUp\ \faStar & \faTimes & \faTimes & \faTimes \\
$\widehat{|X \cap Y|}_{AND}$ \faStar & Eq.~(\ref{eq:bf_int}) & BF & \faThumbsOUp\ \faStar & \faThumbsOUp\ \faStar & \faTimes & \faTimes & \faTimes  \\
%
%
$\widehat{|X \cap Y|}_{L}$ \faStar & \cref{sec:int_bf_est} & BF & \faThumbsOUp\ \faStar & \faThumbsOUp\ \faStar & \faTimes & \faTimes & \faTimes  \\
$\widehat{|X \cap Y|}_{kH}$ & Eq.~(\ref{est_int_mh}) & $k$-Hash & \faThumbsOUp\ \faStar & \faThumbsOUp\ \faStar & \faThumbsOUp\ \faStar & \faThumbsOUp\ \faStar & \faThumbsOUp\ \faStar  \\
$\widehat{|X \cap Y|}_{1H}$ & \cref{sec:int-1h} & $1$-Hash & \faThumbsOUp\ \faStar & \faThumbsOUp\ \faStar & \faTimes & \faTimes & \faTimes  \\
\bottomrule
\end{tabular}
%
\vspaceSQ{-0.5em}
\caption{\textmd{Summary of theoretical results (estimators) related to
$\widehat{|X|}$ and $\widehat{|X \cap Y|}$.  ``\faStar'': a new result provided
in this work (a new estimator or proving a certain novel property of a given
estimator).
``\textbf{CN}'': a consistent estimator.
``\textbf{AU}'': an asymptotically unbiased estimator.
``\textbf{ML}'': an MLE estimator.
``\textbf{IN}'': an invariant estimator.
``\textbf{AE}'': an asymptotically efficient estimator.
}}
\vspaceSQ{-0.5em}
\label{tab:estimators-summary}
\end{table}

\begin{table}[h]
\centering
\scriptsize
%
\begin{tabular}{@{}llllll@{}}
\toprule
\makecell[c]{\textbf{Result}} &
\makecell[c]{\textbf{Where}} &
\makecell[c]{\textbf{Class}} &
\makecell[c]{\textbf{Q}} & 
\makecell[c]{\textbf{MS}} & 
\makecell[c]{\textbf{CO}} \\ 
\midrule
$\widehat{|X|}_S$ \faStar & Eq.~(\ref{eq:est_ss_bf}) & BF & P \faStar & \faThumbsOUp & \faThumbsOUp \\
$\widehat{|X \cap Y|}_{AND}$ \faStar & Eq.~(\ref{bound_int_BF-eq}) & BF & P \faStar & \faThumbsOUp & \faThumbsOUp \\
$\widehat{|X \cap Y|}_{L}$ \faStar & \cref{sec:int_bf_est} & BF & P \faStar & \faThumbsOUp & \faThumbsOUp \\
$\widehat{|X \cap Y|}_{kH}$ \faStar & Eq.~(\ref{eq:mh_int_conc_union}) & $k$-Hash & E \faStar & \faTimes & \faThumbsOUp \\
$\widehat{|X \cap Y|}_{1H}$ \faStar & Eq.~(\ref{eq:m1h_int_conc_union}) & $1$-Hash &  E \faStar & \faTimes & \faThumbsOUp \\
\bottomrule
\end{tabular}
%
\vspaceSQ{-0.5em}
\caption{\textmd{Summary of theoretical results (bounds) related to
$\widehat{|X|}$ and $\widehat{|X \cap Y|}$ .
``\faStar'': a new result provided in this work.
``\textbf{Q}'': the quality of a given bound, ``\textbf{P}'': polynomial, ``\textbf{E}'': exponential.
``\textbf{MS}'': an MSE bound.
``\textbf{CO}'': a concentration bound.
}}
\vspaceSQ{-1.5em}
\label{tab:bounds-summary}
\end{table}

\enlargeSQ

\marginparsep=2em
\marginpar{\vspace{-1em}\colorbox{yellow}{\textbf{ALL}}}

\subsection{\hl{Analysis of Parallelism in $\widehat{|X \cap Y|}$}}
\label{sec:intersect_summary-par}

In Table~\ref{tab:queries-int}, we provide a work-depth analysis of parallelism
in different estimators, when applied to intersecting vertex neighborhoods.
For the exact intersection applied to CSR, we use two variants: merge (more
advantageous when neighborhoods are similar in size) and galloping (used when
neighborhoods vary in size by a large factor); the exact schemes and work/depth
are provided in numerous works~\cite{shun2015multicore, besta2021sisa,
besta2021graphminesuite}.
\emph{Importantly, using PG gives asymptotic advantages in both work
and depth over CSR}.
Work and depth in BF are dominated by -- respectively -- the bitwise AND over participating bit
vectors (taking $O(B_X / W)$ work) and the final sum of 1s over the resulting bitvector
(taking $O(\log B_X/W)$ depth using a simple binary tree reduction).
Note that $B_X$ is always expressed in bits and thus we divide it with the SIMD
width (or plain memory word size)~$W$ to obtain the actual counts of operations.
MH representations are series of up to $k$ vertex IDs and thus they use standard intersections.
Both BF and MH based intersection have attractive work and depth
as $\log k$ and $\log (B_X / W)$ are in practice very small.

\begin{table}[h]
\vspaceSQ{-1em}
\centering
\setlength{\tabcolsep}{1pt}
\scriptsize
%
\begin{tabular}{@{}llllll@{}}
\toprule
& \makecell[c]{\textbf{CSR} (merge)} &
\makecell[c]{\textbf{CSR} (gallop.)} &
\makecell[c]{\textbf{BF}} &
\makecell[c]{\textbf{$k$--Hash}} &
\makecell[c]{\textbf{1--Hash}} \\ 
\midrule
\textbf{Work:} &
\makecell[l]{$O(d_u+d_v)$} & 
\makecell[l]{$O(d_u \log d_v)$} & 
$O\fRB{\frac{B_X}{W}}$ & 
$O(k)$ & 
$O(k)$ \\ 
\textbf{Depth:} &
$O(\log (d_u + d_v))$ &
$O(\log (d_u + d_v))$ &
$O\fRB{\log \fRB{\frac{B_X}{W}}}$ &
$O(\log k)$ &
$O(\log k)$ \\
\bottomrule
\end{tabular}
\vspaceSQ{-0.5em}
\caption{\textmd{Work and depth of simple parallel algorithms for deriving $|N_u \cap N_v|$
(cardinality of the result of intersecting neighborhoods of vertices $u, v$).
}}
\vspaceSQ{-1.5em}
\label{tab:queries-int}
\end{table}

\section{Using ProbGraph with Graph Algorithms}
\label{sec:using}

\enlargeSQ
\enlargeSQ
%
We carefully design and implement PG as an easy-to-use library 
offering different set representations.
To use PG, the user (1) creates selected probabilistic representations
of vertex neighborhoods (BF, $k$-Hash, or 1-Hash), (2) plugs in PG 
routines for obtaining $\widehat{|X \cap Y|}$ in place of the exact set
intersections. For example, to use PG with graph algorithms from
Section~\ref{sec:algs}, one replaces the operations indicated with blue
color with PG routines.
As an example, in Listing~\ref{lst:example}, we compare obtaining Jaccard
similarity of two neighborhoods with an exact scheme and with a PG 
routine based on BF.
We ensure that one can flexibly select an arbitrary estimator
$\widehat{|X \cap Y|}$ because -- as our evaluation 
(Section~\ref{sec:eval}) shows -- no single representation works best in all
cases.

\begin{lstlisting}[language=c++, float=h, aboveskip=-0.5em,belowskip=-2.5em,abovecaptionskip=0.0em,label=lst:example, caption=Obtaining exact and approximate Jaccard (see Listing \ref{lst:sim})]
//|\textbf{Input}|: Graph $G$, two vertices $u$ and $v$ 
//Create a standard CSR graph with $G$ as the input graph
CSRGraph g = CSRGraph($G$); 
//Create a ProbGraph representation of $G$ based on Bloom filters
ProbGraph pg = ProbGraph(g, BF, 0.25); //Use the 25% storage budget

//Derive the exact intersection cardinality $|N_u \cap N_v|$
int interEX = pg.int_card(g.N(u), g.N(v)); 
//Derive the estimator $\widehat{|N_u \cap N_v|}_{AND}$
int interBF = pg.int_BF_AND(pg.N(u), pg.N(v)); 

//Compute the exact Jaccard coefficient between $u$ and $v$
double jacEX = interEX / (g.N(u).size() + g.N(v).size() - interEX)
//Compute the approximate Jaccard coefficient based on BF 
double jacBF = interBF / (g.N(u).size() + g.N(v).size() - interPG)
\end{lstlisting}

\subsection{Tradeoffs Between Storage, Accuracy, \& Performance}
\label{sec:storage-b}

\marginparsep=0.5em
\marginpar{\vspace{0em}\colorbox{yellow}{\textbf{R-5}}\\ \colorbox{yellow}{(minor}\\ \colorbox{yellow}{comm-}\\ \colorbox{yellow}{-ent 12)}}

\hl{Each probabilistic set representation considered in PG} offers a tradeoff between
performance, storage, and accuracy. In general, the smaller a representation is, the
faster to process it becomes and the less storage it needs, but also
the less accurate it becomes.
\hl{To control this tradeoff, we introduce a generic parameter~$s$ that enables
explicit control of the storage budget. 
$s \in [0;1]$ specifies how much additional memory (on top
of the storage needed for the default CSR graph representation) is needed to
maintain the PG estimators.}
In evaluation, we do not exceed more than 33\% of the additional needed storage.

\marginparsep=1em
\marginpar{\vspace{-5em}\colorbox{yellow}{\textbf{R-3}}\\ \colorbox{yellow}{\textbf{R-4}}}

%

\section{Design \& Implementation}
\label{sec:design}

\enlargeSQ

%
Each BF is implemented as a
simple bit vector.
$\widehat{|X \cap Y|}$ can then be computed using bitwise AND over $X$ and
$Y$, with Eq.~(\ref{eq:bf_int}).
Computing $B_{X \cap Y}$ can easily be {parallelized}
and accelerated with {vectorization}~\cite{besta2017slimsell}: the problem is
embarrassingly parallel and the bitwise AND is supported with SIMD
technologies such as AVX deployed in Intel CPUs, GPUs, and others. 
We also use the \texttt{popcnt} CPU
instructions~\cite{mula2017faster} to speed up deriving the number of ones in a
bit vector (1-bits), needed to obtain $B_{X \cap Y,1}$ in Eq.~(\ref{eq:bf_int});
\texttt{popcnt} counts
the number of 1-bits in one memory word \emph{in one CPU cycle}.

%
%
1--Hash and $k$--Hash are both series of integers. The estimators for $|X
\cap Y|$ based on 1--Hash and $k$--Hash are dominated by intersecting
sets of $k$ numbers. As $k \ll d$, it is
much faster than the corresponding operations on exact neighborhoods. 

\enlargeSQ

\subsection{Parallel Construction}
\label{sec:constr-costs}

\marginparsep=2em
\marginpar{\vspace{1em}\colorbox{yellow}{\textbf{R-5}}\\ \colorbox{yellow}{(minor}\\ \colorbox{yellow}{comm-}\\ \colorbox{yellow}{-ent 13)}}

\hl{Table~\mbox{\ref{tab:constr}} provides work and depth of constructing all probabilistic set
representations used in PG}.
\hl{As with the intersection computation, the construction process} is
also parallelizable, exhibiting very low depth.
During evaluation (Section~\ref{sec:eval}), we show that
it also does not pose a bottleneck in practice.

\begin{table}[h]
\vspaceSQ{-0.5em}
\centering
\setlength{\tabcolsep}{3pt}
\scriptsize
\footnotesize
%
\begin{tabular}{@{}llll@{}}
\toprule
\makecell[c]{\textbf{Representation}\\ \textbf{of $N_v$}} &
\makecell[c]{\textbf{$\ddagger$ Size}\\ {[bits]}} &
\makecell[c]{\textbf{Construction}\\ \textbf{(work)}} &
\makecell[c]{\textbf{Construction}\\ \textbf{(depth)}} \\ 
\midrule
%
%
%
BF & $B_X$ & $O(b d_v)$ & $O(\log(b d_v))$ \\
$k$-Hash  & $W k$  & $O(k d_v)$ & $O(\log d_v)$ \\ 
1-Hash  & $W k$  & $O(d_v)$ & $O(\log d_v)$ \\
\bottomrule
\end{tabular}
%
\vspaceSQ{-0.5em}
\caption{\textmd{Work/depth of simple algorithms for constructing
a probabilistic PG set representation of a given neighborhood $N_v$.
In BF, one must iterate over all $b$ hash functions and all $d_v$
neighbors, thus the work is dominated by $b d_v$ (cf.~\cref{sec:prob-back}). All the hash function
evaluations can run in parallel, but -- in the worse case -- they
may write to the same cell in the BF bit vector, giving depth $O(\log (b d_v))$
(parallelization with a binary tree reduction).
Derivations for MH are similar; the work and depth are dominated by evaluations of hash functions
and by finding $k$ smallest elements among $d_v$ ones, respectively.}}
\vspaceSQ{-0.0em}
\label{tab:constr}
\end{table}

\subsection{Parallelism in ProbGraph-Enhanced Graph Algorithms}
\label{sec:par-algs-wd}

Parallelization of graph algorithms enhanced with ProbGraph is straightforward
and is based on the listings from Section~\ref{sec:algs}.  Specifically, all
the loops marked with \texttt{[in par]} can be executed in parallel. Then, all
the instances of set intersection cardinality are executed using a
user-specified PG estimator.
The parallel execution of these estimators (cf.~\cref{sec:intersect_summary-par}
and Table~\ref{tab:queries-int}) enables better work and depth of graph mining
algorithms than with the default CSR implementation. We illustrate this in
Table~\ref{tab:wd-algs}.
Here, work and depth of CSR based routines are standard results known from
extensive works in parallel algorithm design~\cite{besta2017push, besta2021enabling,
shun2015multicore, blelloch1990pre, blelloch2010parallel}.
For example, in TC, the two outermost loops can be executed fully in parallel,
and the nested set intersection dominates depth ($d$ is the maximum degree in
a graph).
Both work and depth for PG baselines are derived by replacing the
nested exact $|X \cap Y|$ operation with the corresponding PG schemes and
results from Table~\ref{tab:queries-int}.
These asymptotic advantages are supported with empirical outcomes detailed
in Section~\ref{sec:eval}.

\begin{table}[t]
\vspaceSQ{-1em}
\centering
\setlength{\tabcolsep}{2pt}
\scriptsize
\footnotesize
%
\begin{tabular}{@{}llll@{}}
\toprule
& \makecell[c]{\textbf{CSR}} &
\makecell[c]{\textbf{PG (BF)}} &
\makecell[c]{\textbf{PG (MH)}} \\
\midrule
\textbf{Triangle Counting (work):} &
$O\fRB{n d^2}$ &
$O\fRB{\frac{nd B_X}{W}}$ &
$O\fRB{ndk}$ \\
\textbf{Triangle Counting (depth):} &
$O\fRB{\log d}$ &
$O\fRB{\log \fRB{\frac{B_X}{W}}}$ &
$O\fRB{\log k}$ \\
\textbf{4-Clique Counting (work):} &
$O\fRB{n d^3}$ &
$O\fRB{\frac{nd^2 B_X}{W}}$ &
$O\fRB{n d^2 k}$ \\
\textbf{4-Clique Counting (depth):} &
$O\fRB{\log^2 d}$ &
$O\fRB{\log d \log \fRB{\frac{B_X}{W}}}$ &
$O\fRB{\log^2 k}$ \\
\textbf{Clustering~(work):} &
$O\fRB{n d^2}$ &
$O\fRB{\frac{nd B_X}{W}}$ &
$O\fRB{ndk}$ \\
\textbf{Clustering~(depth):} &
$O\fRB{\log d}$ &
$O\fRB{\log \fRB{\frac{B_X}{W}}}$ &
$O\fRB{\log k}$ \\
\textbf{Vertex sim.~(work):} &
$O\fRB{d^2}$ &
$O\fRB{\frac{B_X}{W}}$ &
$O\fRB{k}$ \\
\textbf{Vertex vim.~(depth):} &
$O\fRB{\log d}$ &
$O\fRB{\log \fRB{\frac{B_X}{W}}}$ &
$O\fRB{\log k}$ \\
\bottomrule
\end{tabular}
\vspaceSQ{-0.5em}
\caption{\textmd{Advantages of ProbGraph in work and depth over exact baselines.}}
\vspaceSQ{-1.5em}
\label{tab:wd-algs}
\end{table}

\enlargeSQ

\subsection{Implementation Details and Infrastructure}


We use the GMS platform~\cite{besta2021graphminesuite}, a
high-performance parallel graph mining infrastructure, for implementing
the baselines. {Loading graphs from disk} and building the 
CSR representations is done with the GAP Benchmark Suite~\cite{beamer2015gap}.
We use the MurmurHash3 {hash function}~\cite{Appleby:2016}, well-known
for its speed and simplicity. We use the current time in milliseconds as a
  random seed.
For {parallelization}, we use OpenMP~\cite{chandra2001parallel}.
%
%
\hl{The whole implementation is available online.}\footnote{\hl{Link will be available upon publication due to double blindness.}}

\marginparsep=1em
\marginpar{\vspace{-2.5em}\colorbox{yellow}{\textbf{R-3}}}

\marginparsep=1em
\marginpar{\vspace{25em}\colorbox{yellow}{\textbf{R-3}}}

\enlargeSQ

\section{Theoretical Analysis of Accuracy}
\label{sec:theory}

We now illustrate that ProbGraph enables obtaining not only strong theoretical
accuracy guarantees on the set intersection cardinality, but also on graph
properties.
As an example, we now use our estimators $\widehat{|X \cap Y|}$ to develop
estimators~$\widehat{TC}$ for triangle count $TC$, and to derive its
concentration bounds. 

As shown in Listing~\ref{lst:tc}, TC can be obtained by summing intersections
$|N_u \cap N_v|$ of neighborhoods for each pair of adjacent vertices $u$ and
$v$.  Hence, to estimate TC, we simply sum cardinalities $\widehat{|N_u \cap
N_v|}$ for each edge~$(u,v)$ in a given graph.
This gives the following estimator:

\vspaceSQ{-1em}
\begin{gather*}
\widehat{TC}_{\star} = \frac{1}{3} \sum_{(u,v) \in E} \widehat{|N_u \cap N_v|}_{\star}
\end{gather*}
\vspaceSQ{-0.5em}

\noindent
where $\star$ indicates a specific $\widehat{|X \cap Y|}_{\star}$
estimator (cf.~Table~\ref{tab:estimators-summary}).

%

\begin{theorem} \label{thm:tc_deviations} Let $\widehat{TC}_\star$ be the
estimator of the number of triangles.
(cf.~Section~\ref{sec:algs}). Then, depending on the underlying estimator
$\widehat{|X \cap Y|}_{\star}$, we have the following cases: 

For the \textbf{Bloom Filter} AND estimator, if $b \Delta \leq 0.499 B_X \log B_X$, then we have the following bound

\small
\footnotesize
\[
P\left(\left|TC - \widehat{TC}_{AND}\right| \geq t\right) \leq  \frac{2\;m^2 (1+o(1)) \left(e^{\frac{\Delta b }{ B_X-1}} \frac{B_X}{b^2} - \frac{B_X}{b^2} - \frac{\Delta}{b}\right)}{9\;t^{2}}
\]
\normalsize

In the case of both \textbf{1-Hash} and \textbf{$k$-Hash} (below, we use
the notation for 1-Hash), we have
%
%
\small
\begin{gather*}
	P\left(\left|TC - \widehat{TC}_{1H}\right| \geq t\right) \leq 2 \exp \left( - \frac{18\;k\;t^2}{\left(\sum_{v \in V} d(v)^2\right)^2} \right)
\end{gather*}
\normalsize

Moreover, if the maximum degree is $\Delta$, then
%
%
\small
\begin{gather*}
	P\left(\left|TC - \widehat{TC}_{1H}\right| \geq t\right) \leq 2 \exp \left( - \frac{9\;k\;t^2}{4\left(\Delta+1\right) \sum_{v \in V} d(v)^3} \right)
\end{gather*}
\normalsize
\end{theorem}

\iftr
We provide a detailed proof of each statement of
Theorem~\ref{thm:tc_deviations} in the Appendix. 
\fi


%

Consistency of all the TC estimators follows from consistency of the individual
estimators (cf.~\cref{sec:back-properties}). The fact that
$\widehat{TC}_{kH}$ is MLE follows from $\widehat{|X \cap Y|}_{kH}$ being MLE.

\enlargeSQ

\if 0
\begin{table*}[t]
\else
\begin{table}[t]
\fi
\vspaceSQ{-1em}
\centering
\if 0
\footnotesize
\else
\setlength{\tabcolsep}{1pt}
\scriptsize
\ssmall
\fi
%
\begin{tabular}{lllllllllll}
\toprule
& \textbf{Reference} &  \makecell[c]{\textbf{Constr.}\\\textbf{time}} &
\makecell[c]{\textbf{Memory}\\\textbf{used}} &
\makecell[c]{\textbf{Estimation}\\\textbf{time}} & \textbf{AU} & \textbf{CN} & \textbf{ML} &
\textbf{IN} & \textbf{AE} & \makecell[c]{\textbf{B}} \\ 
\midrule
\multirow{12}{*}{\begin{turn}{90}\textbf{Past results}\end{turn}} &
\makecell[l]{Doulion~\cite{tsourakakis2009doulion}} & $O(m)$ & $O(pm)$ & $O(T(pm))$ & \faThumbsOUp & \faThumbsOUp & \faTimes & \faTimes & \faTimes & \faTimes \\
& Colorful~\cite{pagh2012colorful} & $O(m)$ & $O(pm)$ & $O(T(pm))$ & \faThumbsOUp & \faThumbsOUp & \faTimes & \faTimes & \faTimes & \faThumbsOUp\ (P) \\ 
%
%
%
& Sketching~\cite{bandyopadhyay2016topological} & $O(km)$ & $O(kn)$ & \makecell[l]{$O(T(k^2n))$} & \faThumbsOUp & \faThumbsOUp &  \faTimes  &  \faTimes  & \faTimes & \faTimes \\
& ASAP~\cite{iyer2018asap} & n/a & $O(n+m)$ & $O(1)$ / sample & \faTimes & \faTimes & \faTimes & \faTimes & \faTimes & \faTimes \\
& GAP~\cite{iyer2018bridging} & $O(m)\dagger$ & $O(m')\dagger$ & $O(T(m'))\dagger$ & \faTimes & \faTimes & \faTimes & \faTimes & \faTimes & \faTimes \\  
& \makecell[l]{Slim Gr.~\cite{besta2019slim}} & $O(m)$ & $O(pm)$ & $O(T(pm))$ & \faThumbsOUp & \faThumbsOUp & \faTimes & \faTimes & \faTimes & \faTimes \\
& Eden et al.~\cite{eden2017approximately} & n/a & $O\fRB{\frac{n}{TC^{1/3}}}$ & $O\fRB{\frac{n}{TC^{1/3}} + \frac{m^{3/2}}{TC}}$ & \faThumbsOUp & \faThumbsOUp & \faTimes & \faTimes & \faTimes & \faThumbsOUp \\
& Assadi et al.~\cite{assadi2018simple} & n/a & $O(1)$ & $O(m^{3/2}/TC)$ & \faThumbsOUp & \faThumbsOUp & \faTimes & \faTimes & \faTimes & \faThumbsOUp \\
& Tětek~\cite{tetek2021approximate} & n/a & $\fRB{\frac{m^{1.41}}{TC^{0.82}}}$& $\fRB{\frac{m^{1.41}}{TC^{0.82}}}$& \faThumbsOUp & \faThumbsOUp & \faTimes & \faTimes & \faTimes & \faThumbsOUp \\

\midrule
\multirow{4}{*}{\begin{turn}{90}\textbf{PG}\end{turn}} &
$\widehat{TC}_{AND}$ (BF) & $O(b m)$ & $O(n+m)$ & $O(m B / W)$ & \faThumbsOUp & \faThumbsOUp & \faTimes & \faTimes & \faTimes &  \faThumbsOUp\ (P) \\
& $\widehat{TC}_{kH}$ (MH) & $O(k m)$ & $O(n+m)$ & $O(km)$  & \faThumbsOUp & \faThumbsOUp &  \faThumbsOUp & \faThumbsOUp & \faThumbsOUp & \faThumbsOUp\ (E) \\
& $\widehat{TC}_{1H}$ (MH) & $O(k m)$ & $O(n+m)$ & $O(km)$  & \faThumbsOUp & \faThumbsOUp & \faTimes&  \faTimes& \faTimes & \faThumbsOUp\ (E) \\
%
%
\bottomrule
\end{tabular}
\vspaceSQ{-0.5em}
\caption{
\textmd{
\textbf{ProbGraph vs.~existing results for estimating TC} (sorted chronologically).
\textbf{``Constr.~time'':} time to construct a given estimator.
\textbf{``Memory'':} the amount of storage needed to construct a given estimator.
\textbf{``Estimation time'':} time needed to estimate TC.
\textbf{``AU'' (asymptotically unbiased), ``CN'' (consistent)}, \textbf{``ML'' (maximum likelihood estimator)},
\textbf{``IN'' (invariant)}, \textbf{``AE'' (asymptotically efficient)}: properties
of estimators (explained in~\cref{sec:back-properties}).
\textbf{``B'' (concentration bounds)}: whether a given scheme is supported with
concentration bounds.
\textbf{``P'' (polynomial)} or \textbf{``E'' (exponential)}: bound quality.
``\faThumbsOUp'': supported, provided.
``\faTimes'': not available, not provided.
%
%
%
``$\dagger$'': the original work does \emph{not}
explicitly provide a given result and it was derived in this work.
%
%
%
\textbf{\ul{Symbols used in related work} (different from ones in
Table~\ref{tab:symbols}):}
$p$: probability of keeping an edge,
$m'$: \#sampled edges.
}}
\vspaceSQ{-2em}
\label{tab:tc}
\if 0
\end{table*}
\else
\end{table}
\fi

\enlargeSQ

\subsection{Comparison to Existing Estimators}

We compare our $\widehat{TC}$ estimators to others in Table~\ref{tab:tc}.
We consider the estimators from Doulion~\cite{tsourakakis2009doulion},
topological graph sketching~\cite{bandyopadhyay2016topological},
GAP~\cite{iyer2018bridging}, ASAP~\cite{iyer2018asap},
Slim Graph~\cite{besta2019slim}, MCMC~\cite{Rahman:2014}, and
the ``colorful'' TC analysis~\cite{pagh2012colorful}, as well as several more
recent results from the theory community~\cite{eden2017approximately,assadi2018simple}.
%
%
In the comparison, we consider construction time, used memory, TC estimation time,
and whether an estimator is asymptotically unbiased, consistent, maximum likelihood,
invariant, and whether it offers concentration bounds, and -- if yes -- 
are they polynomial or exponential (cf.~Section~\ref{sec:back-properties}
for an explanation on the relevance of these properties).
%

All
ProbGraph's estimators offer polynomial or exponential
concentration bounds.
Importantly, $\widehat{TC}$ based on MinHash is the only one to offer
exponential concentration bounds so far.  This means that -- for both
$\widehat{TC}_{1H}$ and $\widehat{TC}_{kH}$ -- any deviation from the true
value of $TC$ goes to zero exponentially fast with the increasing size of the
potential deviation.
Moreover, we observe that $\widehat{TC}_{kH}$ has all the desirable estimator
properties mentioned above.
Thus, it is particularly attractive whenever high accuracy is of the uttermost importance.

%
%
%
%
\if 0
These bounds, which we denote as $\widehat{TC}_{col}$, are as follows:
\else
\fi

\section{Evaluation} \label{sec:eval}

We now show that ProbGraph enables large speedups in graph mining while
maintaining high accuracy of outcomes.
%
%
We do not advocate a
single way of deriving $|X \cap Y|$, but we
illustrate pros and cons of different classes of schemes, and underline when each
scheme is best applicable.

\begin{figure*}[t]
\centering
%
%
\centering
\includegraphics[width=1.0\textwidth]{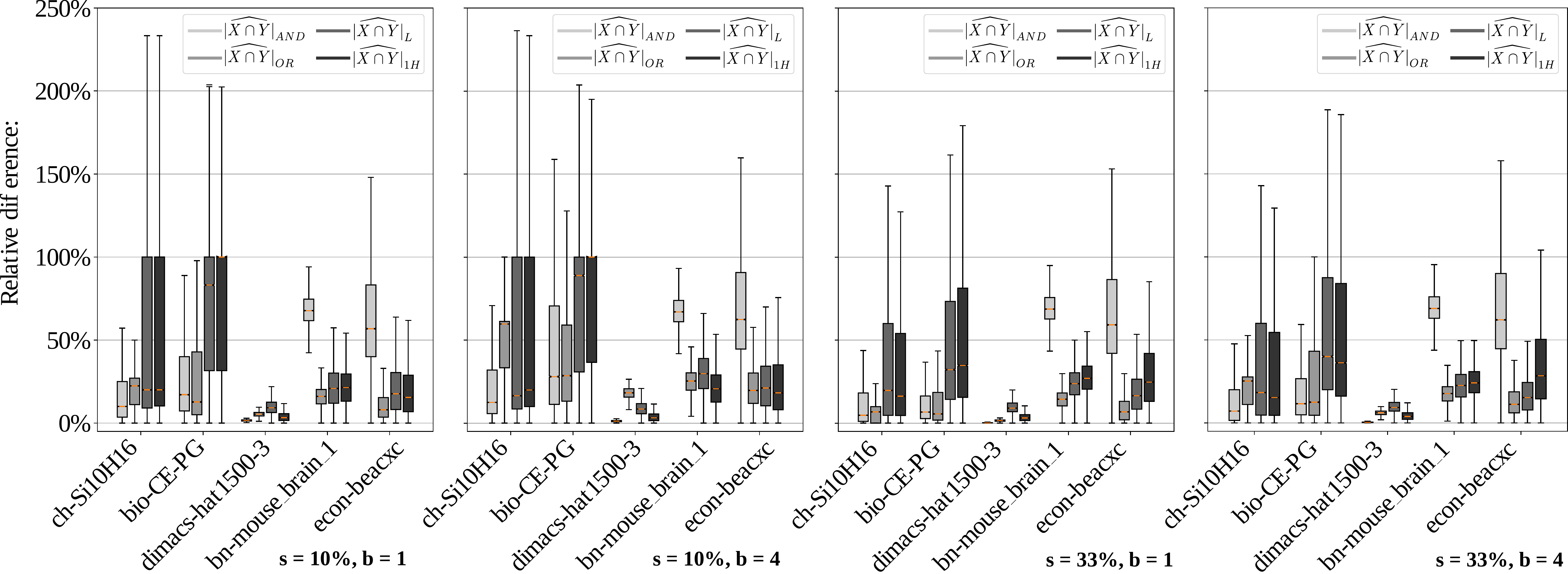}
\vspaceSQ{-1.0em}
\caption{\textmd{Analysis of the accuracy of PG estimators of $|X \cap Y|$.}}
\label{fig:int-analysis}
\vspaceSQ{-1em}
\end{figure*}

\subsection{Datasets, Methodology, Architectures}

\marginparsep=0.5em
\marginpar{\vspace{0em}\colorbox{yellow}{\textbf{R-5}}\\ \colorbox{yellow}{(minor}\\ \colorbox{yellow}{comm-}\\ \colorbox{yellow}{-ent 14)}}

\hl{We follow a recent set of recommendations on the
benchmarking parallel
applications~\mbox{\cite{hoefler2015scientific}}}. For example, we omit the first 1\%
of performance data as warmup. We derive enough data for the mean and 95\%
non-parametric confidence intervals. 
%

\textbf{Comparison Baselines}
%
%
%
We compare PG-based approximate graph algorithms to
tuned state-of-the-art implementations (Triangle and 4-Clique Counting,
Clustering, and Vertex Similarity) from the GAP~\cite{beamer2015gap} and
GMS~\cite{besta2021graphminesuite} graph benchmarking suites.
Moreover, when analyzing our estimators of $|N_u \cap N_v|$, we
consider an existing BF estimator~\cite{papapetrou2010cardinality,
harmouch2017cardinality}, given by the expression $\widehat{|X|} = -
\frac{\log\fRB{1 - {B_{X,1}}/{B_X}}}{b \log\fRB{1 - {1}/{B_X}}}$.
We also consider another existing estimator~\cite{swamidass2007mathematical},
given by the expression $\widehat{|X \cap Y|}_{OR} = |X| + |Y| + \frac{B_{X
\cup Y}}{b} \log \left( 1 - \frac{B_{X \cup Y,1}}{B_{X \cup Y}} \right)$; this
estimator uses the single set estimator evaluated on the set union.
Finally, when evaluating TC, we compare to the established TC
estimators: Doulion~\cite{tsourakakis2009doulion} (representing schemes based
on edge sampling) colorful TC~\cite{pagh2012colorful} (representing schemes
based on sophisticated combinatorial pruning).
\hl{We also consider state-of-the-art heuristics that do not come with
theoretical guarantees: Reduced Execution~\mbox{\cite{singh2018scalable}},
Partial Processing~\mbox{\cite{singh2018scalable}}, Auto-Approximate (two
variants)~\mbox{\cite{shang2014auto}}.}
%

\marginparsep=0.5em
\marginpar{\vspace{-3em}\colorbox{yellow}{\textbf{R-1}}}

\enlargeSQ
\enlargeSQ
\enlargeSQ

\textbf{Datasets}
We consider {SNAP (S)}~\cite{snapnets}, {KONECT (K)}~\cite{kunegis2013konect},
{DIMACS (D)}~\cite{demetrescu2009shortest}, Network Repository (N)~\cite{nr},
and {WebGraph (W)}~\cite{BoVWFI} datasets.
For broad analysis, we follow the recommendations of the GMS graph
mining benchmark~\cite{besta2021graphminesuite}, and we use networks of different
origins (biology, chemistry, economy, etc.), sizes, densities
($m/n$), degree distribution skews, and even {higher-order
characteristics} (e.g., counts of cliques).
We illustrate the real-world datasets in Table~\ref{tab:graphs}.
We also use synthetic graphs power-law (the Kronecker
model~\mbox{\cite{leskovec2010kronecker}}) degree distribution.  Using such
synthetic graphs enables systematically changing a specific single graph
property such as \mbox{$n$}, \mbox{$m$}, or \mbox{$m/n$}, which is not possible
with real-world datasets.
\emph{This entails a very large evaluation space and we only include representative
findings for selected graphs.}

\begin{table}[h]
\centering
\footnotesize
\scriptsize
%
\setlength{\tabcolsep}{2pt}
\renewcommand{\arraystretch}{1}
\begin{tabular}{l}
\toprule
\makecell[l]{
\textbf{\ul{Biological}.} 
 Gene functional associations: 
   ({\emph{bio-SC-GT}}, 1.7K, 34K), 
   ({\emph{bio-CE-PG}},\\1.9K, 48K), 
   ({\emph{bio-CE-GN}}, 2.2K, 53.7K), 
   ({\emph{bio-DM-CX}}, 4K, 77K), 
   ({\emph{bio-DR-CX}}, 3.3K,\\ 85K),
   ({\emph{bio-HS-LC}}, 4.2K, 39K), 
   ({\emph{bio-HS-CX}}, 4.4K, 108.8K), 
   ({\emph{bio-SC-HT}}, 2K, 63K),\\
   ({\emph{bio-WormNetB3}}, 2.4K, 79K).
   ({\emph{bio-WormNet-v3}}, 16.3K, 762.8K).
 Human gene \\regulatory network: 
   ({\emph{bio-humanGene}}, 14K, 9M),
   ({\emph{bio-mouseGene}}, 45K, 14.5M).
%
}\\
%
%
\makecell[l]{\textbf{\ul{Interaction}.}
 Animal networks:  
    ({\emph{int-antCol3-d1}}, 161, 11.1K), 
    ({\emph{int-antCol5-d1}},\\153, 9K),
    ({\emph{int-antCol6-d2}}, 165, 10.2K),
    ({\emph{intD-antCol4}}, 134, 5K). 
 Human contact\\network:
  ({\emph{int-HosWardProx}}, 1.8k, 1.4k).
 Users-rate-users:
    ({\emph{int-dating}}, 169K, 17.3M),\\
({\emph{edit-enwiktionary}}, 2.1M, 5.5M).
 Collaboration:
    ({\emph{int-citAsPh}}, 17.9K, 197K).
}\\
%
%
\makecell[l]{\textbf{\ul{Brain}.}
  ({\emph{bn-flyMedulla}}, 1.8K, 8.9K),
  ({\emph{bn-mouse}}, 1.1K, 90.8K),\\
  ({\emph{bn-mouse\_brain\_1}}, 213, 21.8K).
}\\
%
\makecell[l]{\textbf{\ul{Economic}.}
  ({\emph{econ-psmigr1}}, 3.1K, 543K),
  ({\emph{econ-psmigr2}}, 3.1K, 540K),\\
  ({\emph{econ-beacxc}}, 498, 50.4K),
  ({\emph{econ-beaflw}}, 508, 53.4K),
  ({\emph{econ-mbeacxc}}, 493, 49.9K),\\
  ({\emph{econ-orani678}}, 2.5K, 90.1K).
}\\
%
%
\makecell[l]{\textbf{\ul{Social}.}
  Facebook: ({\emph{soc-fbMsg}}, 1.9k, 13.8k),
  Orkut: (3.1M, 117M).
}\\
%
%
\makecell[l]{\textbf{\ul{Scientific computing}.}
  ({\emph{sc-pwtk}}, 217.9K, 5.6M),
  ({\emph{sc-OptGupt}}, 16.8K, 4.7M),\\
  ({\emph{sc-ThermAB}}, 10.6K, 522.4K).
}\\
%
%
\makecell[l]{\textbf{\ul{Discrete math}.}
  ({\emph{dimacs-c500-9}}, 501, 112K),
  ({\emph{dimacs-hat1500-3}}, 1.5K, 847K).
}\\
%
%
\makecell[l]{\textbf{\ul{Chemistry}.}
  ({\emph{ch-SiO}}, 33.4K, 675.5K),
  ({\emph{ch-Si10H16}}, 17K, 446.5K).
}\\
%
%
%
\bottomrule
\end{tabular}
\vspaceSQ{-0.5em}
\caption{\textmd{Used
graphs.
For each graph, we show its ``(\#vertices, \#edges)''.
}}
%
\label{tab:graphs}
\vspaceSQ{-1em}
\end{table}


\textbf{Parametrizing Set Representations}
We use the generic storage budget parameter~$s$ (cf.~\cref{sec:storage-b}) to
set the maximum allowable amount of memory than can be used by PG. Then, the
parameters specific to each probabilistic set representation ($b, B_X, k$)
enable fine tuning the tradeoff between storage, accuracy, and performance. In
the following, we will also illustrate how to pick these parameters to maximize
performance and accuracy for a given storage budget.

\marginparsep=1em
\marginpar{\vspace{-17em}\colorbox{yellow}{\textbf{ALL}}}

\enlargeSQ

\textbf{Architectures}
We use a 
a Dell PowerEdge R910 server with an Intel Xeon X7550 CPUs @ 2.00GHz with
18MB L3 cache, 1TiB RAM, and 32 cores per CPU (grouped in four sockets).
We also use XC50 compute nodes in the Piz Daint Cray
supercomputer (one such node comes with 12-core Intel Xeon E5-2690 HT-enabled
CPU 64 GiB RAM).

\iftr
\textbf{Parallelism}
Unless stated otherwise, we run algorithms on the maximum number of cores
available in a system.
\fi

\textbf{Assessing Accuracy}
To measure the accuracy of
algorithms that return some \emph{counts} (e.g., clique count, count of
clusters), we use expression $\frac{|cnt_{PG} -
cnt_{EX}|}{cnt_{EX}}$, where $cnt_{PG}$ and $cnt_{EX}$ are ProbGraph and
exact counts, respectively.
Note that the ProbGraph counts may be lower but also higher than the exact ones (due to false positives in BFs). 

\begin{figure*}[t]
\centering
\vspaceSQ{-0.5em}
\includegraphics[width=1.0\textwidth]{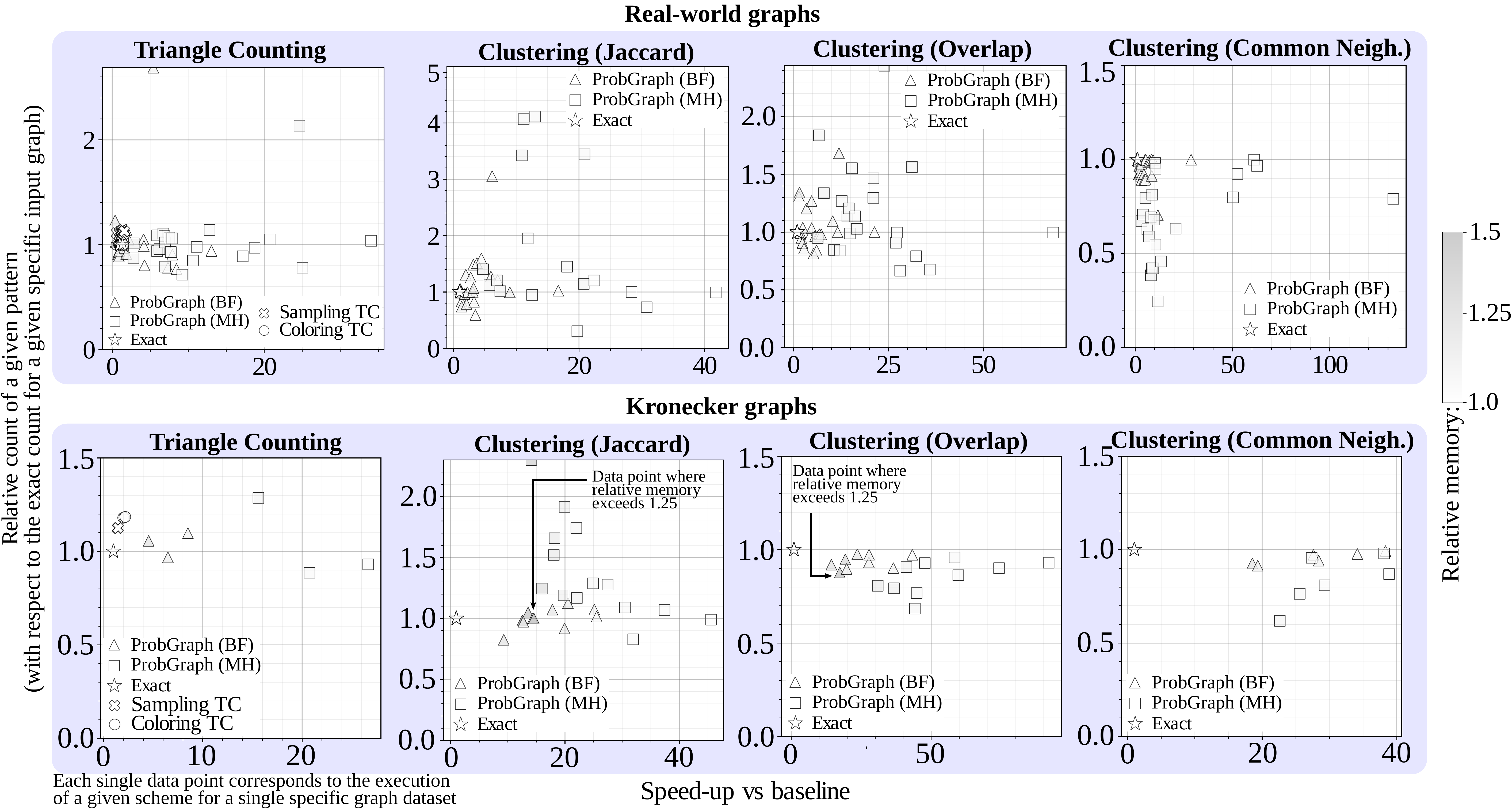}
\vspaceSQ{-1.5em}
\caption{\textmd{Summary of advantages of PG \hl{for real-world (top panel) and
Kronecker (bottom panel) graphs, for Triangle Counting and Clustering}. 
All 32 cores are used. \hl{Note that most data points are white or
almost white because they come with very low amounts of additional memory
(we annotate a few data points that come with more than 25\% additional
relative memory amounts). Other graph problems come with the same performance/memory/accuracy patterns
for the used comparison baselines.}}}
\label{fig:pg-benefits-1}
\vspaceSQ{-1em}
\end{figure*}

\enlargeSQ

\subsection{Estimating $|N_v \cap N_u|$}

We first assess specific PG estimators of $|N_v \cap N_u|$ in terms of their
accuracy. 
For each graph, we derive the BF and MH representations of its vertex
neighborhoods, and then the intersections of neighborhoods of
adjacent vertices. Finally, we compute the relative differences between these
PG intersection cardinalities and the CSR related cardinalities
$|(\widehat{|X \cap Y|}_\bullet
- |X \cap Y|)| / |X \cap Y|$ where $\bullet \in \{AND, L, 1H, kH\}$ We
  summarize these differences, for each graph, using boxplots. 
We use the storage budget~$s = 33\%$ and $b \in \{1, 4\}$.

Representative results are in Figure~\ref{fig:int-analysis}.
While medians are low (less than $\approx$25\% for most cases), there is a
certain spread in outliers. This is because we consider
\emph{all} adjacent vertices, and there is a high chance that at
least some pairs will result in low accuracy.
Overall, the results illustrate that there is no single winner among the
estimators, and the outcomes depend on the graph structure. One observation is
that the BF based on AND tends to perform worse on very dense graphs,
and comparably to marginally better than L on sparser graphs.
Similarly, $k$--Hash is marginally worse than 1--Hash on very dense
graphs; sparse graphs entail a reverse pattern.
%

\enlargeSQ

\subsection{Estimating Outcomes of Graph Algorithms}

\marginparsep=0.5em
\marginpar{\vspace{1em}\colorbox{yellow}{\textbf{R-5}}\\ \colorbox{yellow}{(minor}\\ \colorbox{yellow}{comm-}\\ \colorbox{yellow}{-ent 16)}}

We now illustrate that PG estimators enable high performance and high
accuracy at a small additional storage budget when applied to parallel
graph mining.
We first conduct an analysis using all 32 cores.
\hl{For each graph problem considered}, we illustrate the exact baseline and the
schemes based on BF and MH estimators. We use $\widehat{|N_u \cap N_v|}_{AND}$
with $b = 2$ and $\widehat{|N_u \cap N_v|}_{1H}$ that represent BF and MH
schemes; they offer high accuracy while being fast to compute as they need few
hash functions (other estimators come with similar accuracy outcomes but 
are slower to compute).
%
%
Figures~\mbox{\ref{fig:pg-benefits-1}} and~\mbox{\ref{fig:pg-benefits-2}}
show the results for real-world and Kronecker graphs. Each single
plot is dedicated to a specific graph problem and it compares different
estimators \hl{(indicated with different shapes of data points)} across three
dimensions: performance (speedup, X axis), accuracy (relative count, Y axis),
and memory budget (relative memory size with respect to the default CSR, shades
of B\&W).
\sethlcolor{yellow}\hl{Each plot corresponds to a specific graph problem. Each
data point corresponds to the execution of a given scheme for a specific
graph dataset. Thus, each plot shows collectively how different baselines
behave for different input graphs.}

\marginpar{\vspace{-9em}\colorbox{yellow}{\textbf{R-2}}}

\marginpar{\vspace{-3em}\colorbox{yellow}{\textbf{R-2}}}

\enlargeSQ

\marginpar{\vspace{9em}\colorbox{yellow}{\textbf{R-3}}\\ \colorbox{yellow}{\textbf{R-5}}}

In general, the results follow the insights from the analysis of estimating
set intersections. BFs offer high accuracy
and high speedups, sometimes even as high as 20$\times$
(Clustering using Overlap), or nearly 30$\times$
(Clustering using Common Neighbors), while keeping the accuracy more than
98\%. Speedups can be as high as 50$\times$, with the accuracy more
than 90\% (e.g., 4-Clique Counting for Kronecker graphs).
\hl{On the other hand, MH usually gives consistently higher speedups as well as
lower memory requirements, \emph{but its accuracy is in most cases worse than
BF}.}
\hl{We conjecture this is because MH estimators preserve only
specific subsets of vertex neighborhoods (selected using hash functions),
explicitly eliminating other vertices. In contrast, when using BF estimators,
each vertex is hashed to certain bit(s) in the final bit vector, and is thus to
some extent ``reflected'' in PG estimators.}

\marginpar{\vspace{3em}\colorbox{yellow}{\textbf{R-1}}\\ \colorbox{yellow}{\textbf{R-2}}}

\hl{In terms of memory efficiency, the mostly very light shades indicate very low
additional storage overheads. Except for a few cases where light gray indicates
that -- for a given graph -- a given estimator needs around 20-50\% additional space,
\emph{all the cases require at most 25\% more storage}.
We additionally indicate this in the plots with the appropriate annotations.}

We then provide more details on all the problems using detailed bar plots; due to space
constraints, we only show 
TC in Figure~\ref{fig:tc-dets}
for real-world graphs (all other problems and Kronecker graphs come with similar insights). For the
storage budget of at most 25\%, for the majority of graphs -- regardless of
their origin -- PG enables high ($>$80\%) accuracy combined with high speedups.
The only cases of low accuracy (i.e., the number of clusters detected being
much higher than the exact one) is when using MH. This is because MH based
schemes explicitly remove a (usually high) number of edges, resulting in
possibly significant increase in cluster counts.

\begin{figure}[t]
\centering
%
\includegraphics[width=1.0\columnwidth]{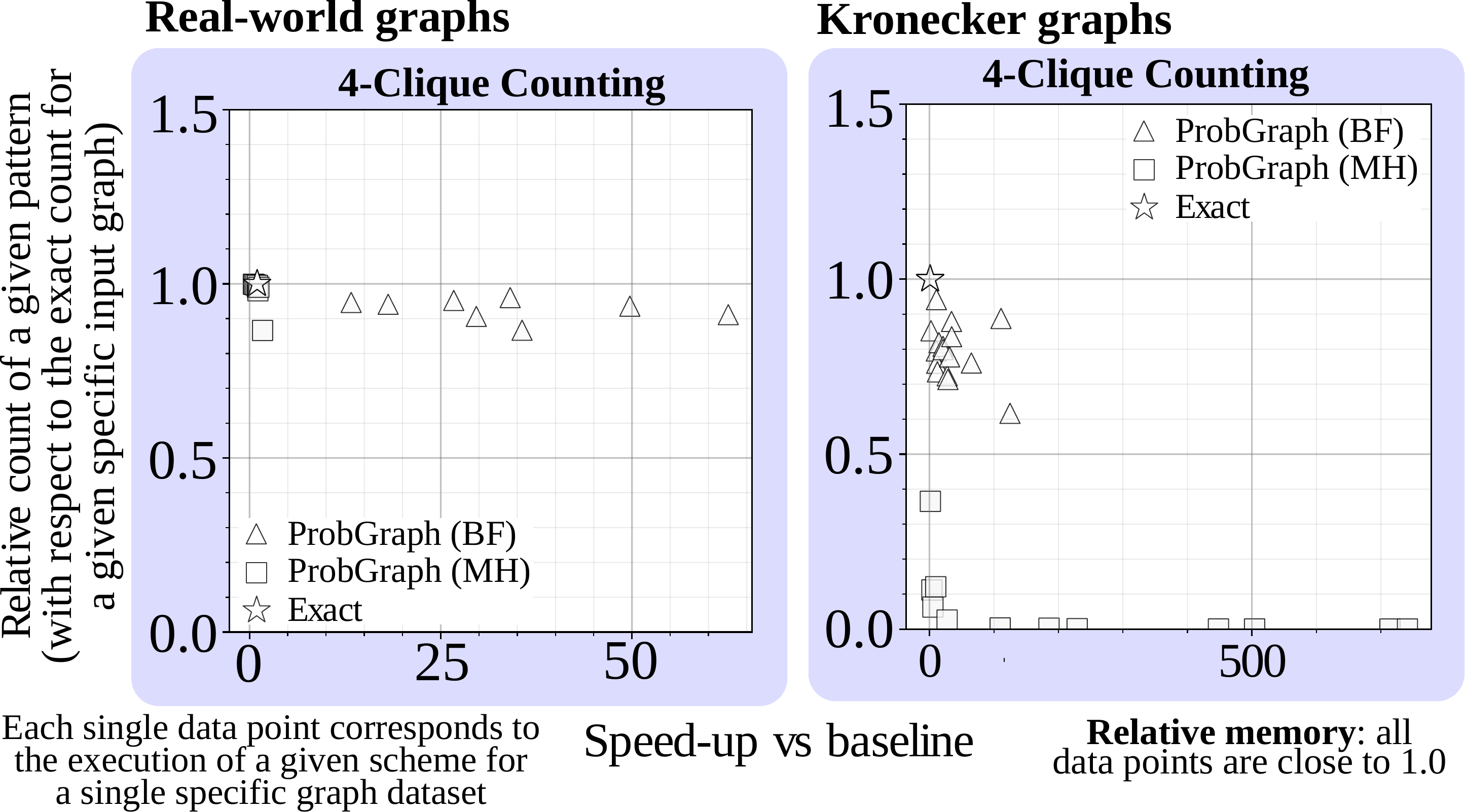}
\vspaceSQ{-1.5em}
\caption{\textmd{Summary of advantages of PG \hl{for 4-clique counting for real-world (left panel) and
Kronecker (right panel) graphs}. All 32 cores are used.}}
\label{fig:pg-benefits-2}
\vspaceSQ{-1.5em}
\end{figure}

We conclude that BF-based PG estimators consistently deliver high accuracy as
well as speedups at small memory budget, for a broad set of graphs and
problems. 1--Hash based schemes may provide much more performance, but require
more careful parametrization and input selection.

\marginpar{\vspace{-45em}\colorbox{yellow}{\textbf{ALL}}}  
\marginpar{\vspace{-20em}\colorbox{yellow}{\textbf{ALL}}}

\marginparsep=2em
\marginpar{\vspace{3em}\colorbox{yellow}{\textbf{R-1}}\\ \colorbox{yellow}{comm-}\\ \colorbox{yellow}{-ent 2}}

\enlargeSQ
\enlargeSQ
\enlargeSQ

\subsection{\hl{Comparison to Heuristics}}

\hl{We also compare to heuristics for approximate graph computations that do
not come with guarantees on the quality of outcomes.
One scheme called ``Reduced Execution''~\mbox{\cite{singh2018scalable}}
reduces the count of iterations of the outermost loop.
Another scheme, ``Partial Graph Processing''~\mbox{\cite{singh2018scalable}},
processes -- for each vertex~$v$ -- a randomly selected subset of $v$'s
neighbors.}
%
%
%
\hl{Moreover, we use two variants of sampling-based ``Auto-Approximation''
that addresses a purely vertex-centric model of
computation~\mbox{\cite{shang2014auto}}.} 
%
%
\hl{We show these heuristics in Figure~\mbox{\ref{fig:tc-dets}} (we exclude them
from Figures~\mbox{\ref{fig:pg-benefits-1}} and~\mbox{\ref{fig:pg-benefits-2}}
to preserve clarity; they always achieve much worse results that obscure
plots). 
The advantage of heuristics is that they do not need additional memory, as
shown in the lowest panel in Figure~\mbox{\ref{fig:tc-dets}}.  However, PG
always achieves much better accuracy, by at least 25\%, up to $\approx$75\%.
This is because the heuristics are not based on theoretical developments that
ensure high PG's accuracy.  Moreover, the heuristics are also slower than PG.
``AutoApproximate'' schemes introduce particularly large overheads due to their
purely vertex-centric abstraction, which makes them even slower than the exact
tuned baselines that we compare to. The results for all other graph problems
are similar.}

\begin{figure*}[t]
\centering
\vspaceSQ{-0.5em}
\includegraphics[width=1.0\textwidth]{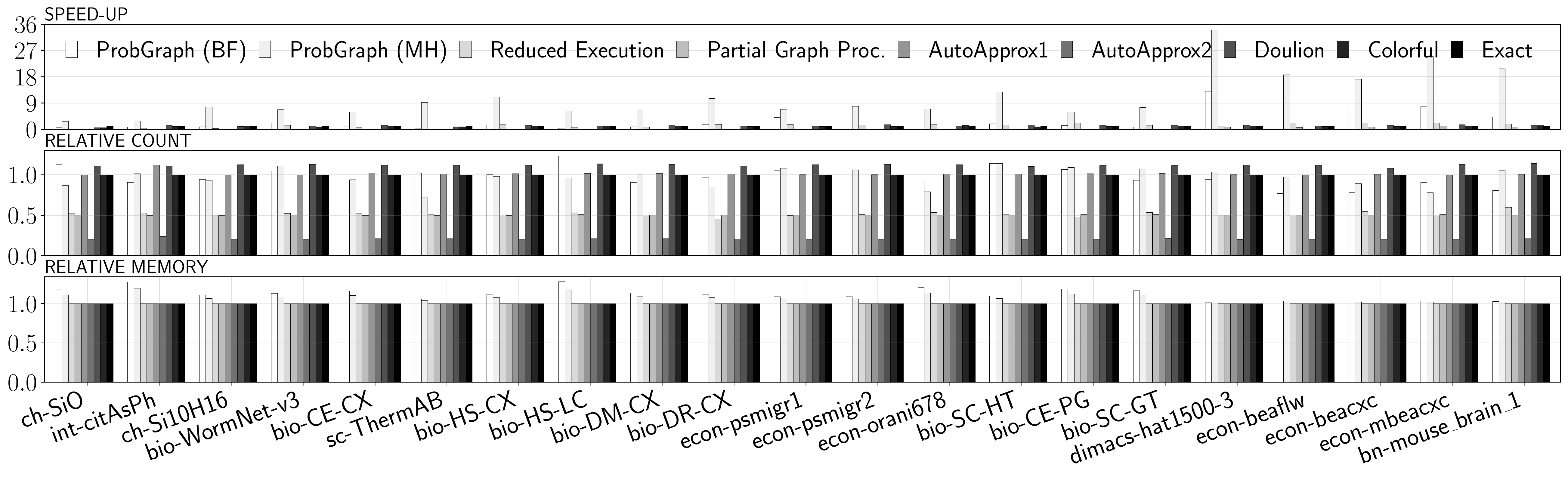}
\vspaceSQ{-2em}
\caption{\textmd{Analysis of performance/accuracy/memory of ProbGraph for
Triangle Counting, \hl{illustrating advantages of ProbGraph over baselines with theoretical underpinning (Sampling, Colorful) and over heuristics (Reduced Execution, Partial Graph Processing, AutoApprox1, AutoApprox2)}.}}
\label{fig:tc-dets}
\vspaceSQ{-0.5em}
\end{figure*}

\enlargeSQ

\enlargeSQ

\subsection{Analysis of Scaling}

\marginparsep=1em
\marginpar{\vspace{0em}\colorbox{yellow}{\textbf{R-5}}\\ \colorbox{yellow}{(minor}\\ \colorbox{yellow}{comm-}\\ \colorbox{yellow}{-ents}\\ \colorbox{yellow}{17-18)}}

We also consider strong/weak scaling; \hl{Figure~\mbox{\ref{fig:scale-analysis}}
offers representative results}. 
We observe nearly ideal strong scaling of all \hl{the baselines compared}. PG schemes feature much
lower runtimes.
To investigate in what regime of parameters PG baselines exhibit better scaling
behavior, we also analyze weak scaling, see Figure~\ref{fig:scaling-weak-tc}.
We use Kronecker graphs. We increase the number of
edges along with the number of threads, from $m \approx 4\text{M}$ to $m
\approx 1.8\text{B}$ for a fixed $n = 1\text{M}$. The largest graphs fill the
whole available memory (1TB).
In this experiment, we increase the number of edges at a \emph{rate twice as
large as the thread count} (cf.~the X axis). As we use Kronecker graphs,
this \emph{stresses load balancing capabilities of the compared baselines},
as most vertices have small neighborhoods, but some neighborhoods
grow particularly fast, making it very challenging to load balance set
intersections (cf.~the right side of Figure~\ref{fig:pg-overview}). The results
illustrate that all PG baselines scale \emph{much better than} all competition
baselines. It becomes particularly visible beyond a certain point, where the PG
scaling curves become gradually flatter.
This is enabled by the PG design, in which set representations are of the same
(usually very small) size. Hence, load imbalance is less of an issue,
while -- as shown earlier in this section -- accuracy loss is
  negligible. 
\hl{Finally, Figure~\mbox{\ref{fig:scale-analysis-cn}} shows that the difference between
BF and MH in scaling also depends on the targeted problem. For Clustering based
on Common Neighbors, BF becomes comparable, or marginally better, than MH,
for large thread counts. This is because
  the algorithm for Clustering is almost completely
  dominated by $|X \cap Y|$, hence benefiting from BF's very fast bitwise AND
  set intersections.}

\marginpar{\vspace{-25em}\colorbox{yellow}{\textbf{ALL}}\\ \colorbox{yellow}{\textbf{   }}\\ \colorbox{yellow}{\textbf{   }}\\ \colorbox{yellow}{\textbf{   }}\\ \colorbox{yellow}{\textbf{   }}\\ \colorbox{yellow}{\textbf{   }}\\ \colorbox{yellow}{\textbf{   }}\\ \colorbox{yellow}{\textbf{   }}\\ \colorbox{yellow}{\textbf{   }}\\ \colorbox{yellow}{\textbf{   }}}

\marginpar{\vspace{-5em}\colorbox{yellow}{\textbf{R-3}}\\ \colorbox{yellow}{\textbf{R-5}}}

\iftr
\begin{figure}[t]
\centering
\vspaceSQ{-1em}
\includegraphics[width=1.0\columnwidth]{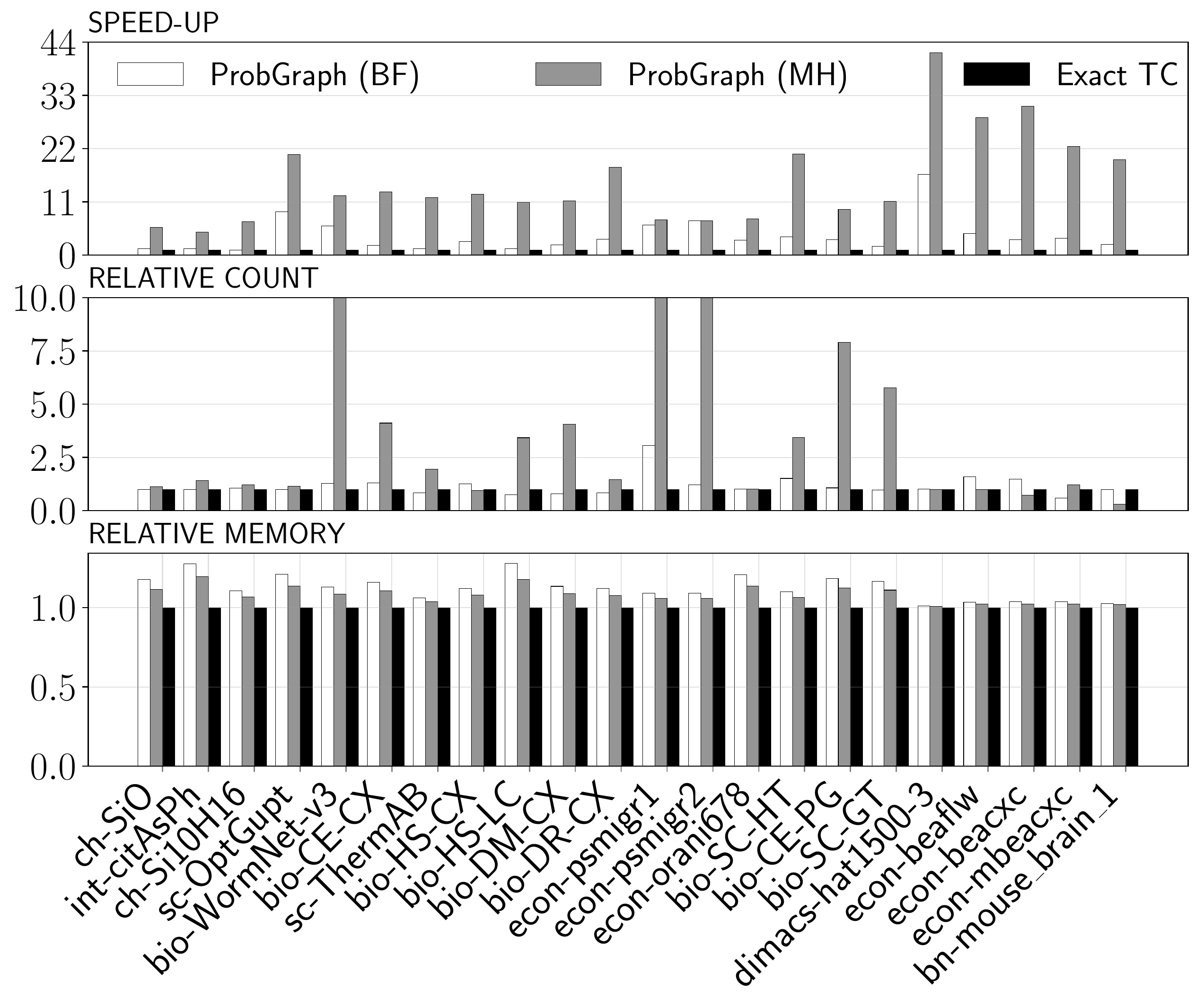}
\vspaceSQ{-0.5em}
\caption{\textmd{Analysis of performance/accuracy/memory of ProbGraph, for
Clustering based on the Jaccard Coefficient score for Vertex
Similarity. For relative counts of clusters, we set a cutoff for the value of 10 
for clarity of plots.}}
\label{fig:cl-dets}
\vspaceSQ{-1em}
\end{figure}
\fi

\ifconf
\else

\begin{figure*}[t]
\centering
\vspaceSQ{-1em}
%
\begin{subfigure}[t]{0.22 \textwidth}
\centering
\includegraphics[width=\textwidth]{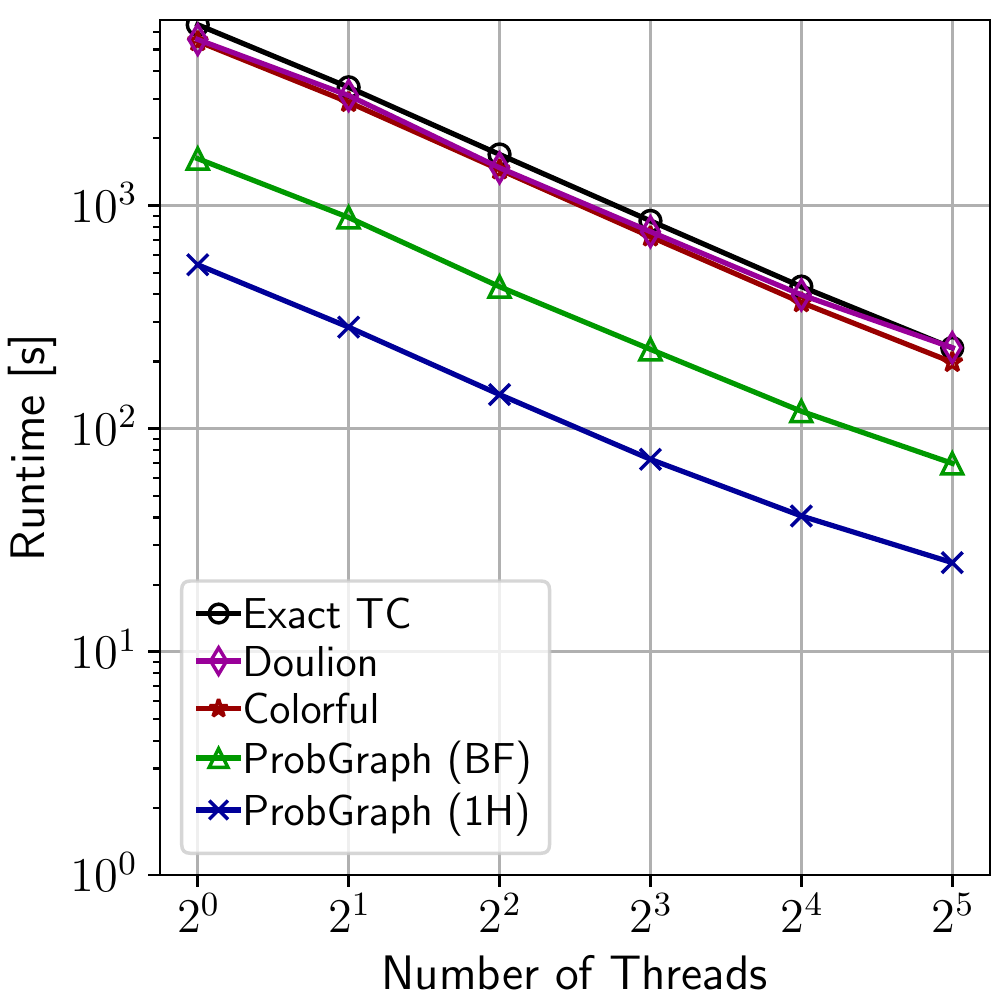}
\vspaceSQ{-1.0em}
\caption{\textmd{Strong scaling (TC).}}
\label{fig:scaling-strong-tc}
\end{subfigure}
\begin{subfigure}[t]{0.22 \textwidth}
\centering
\includegraphics[width=\textwidth]{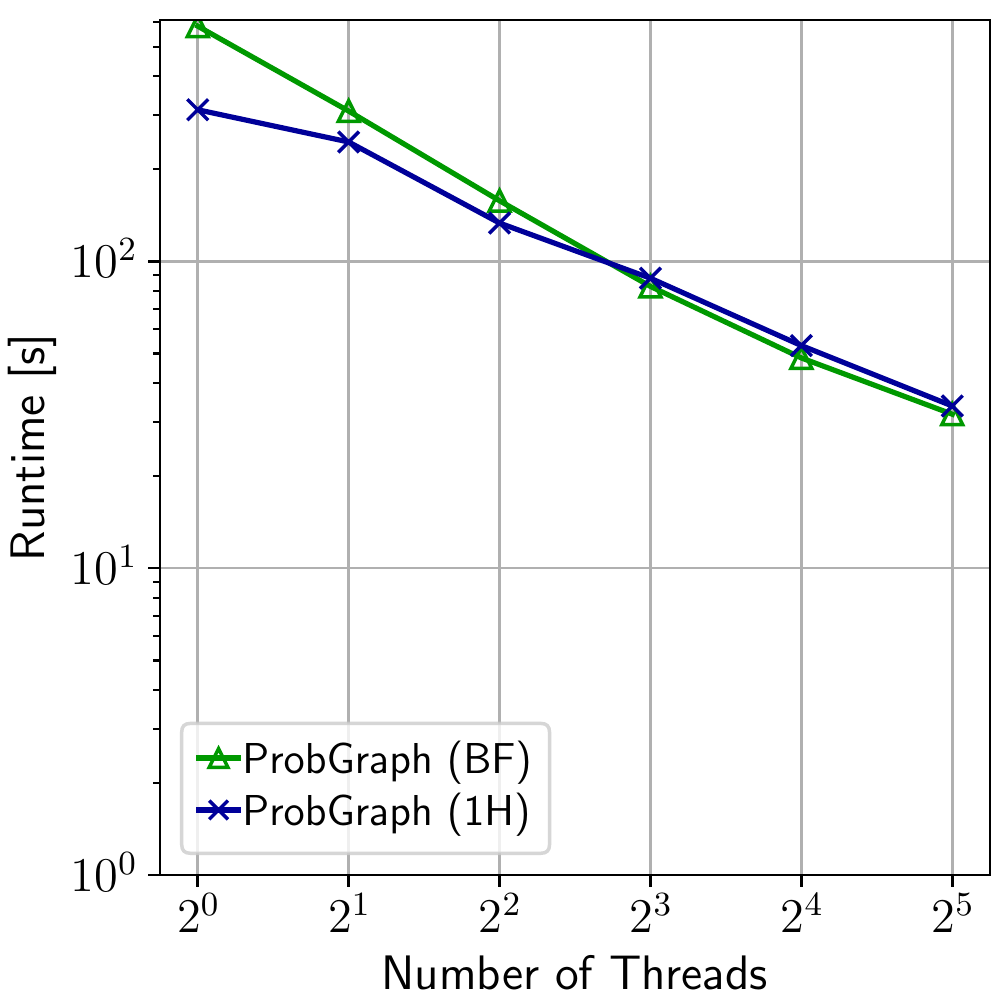}
\vspaceSQ{-1.0em}
\caption{\textmd{Strong scaling (Clustering, Common Neighbors Vertex Similarity).}}
\label{fig:scaling-strong-jp-cn}
\end{subfigure}
\begin{subfigure}[t]{0.22 \textwidth}
\centering
\includegraphics[width=\textwidth]{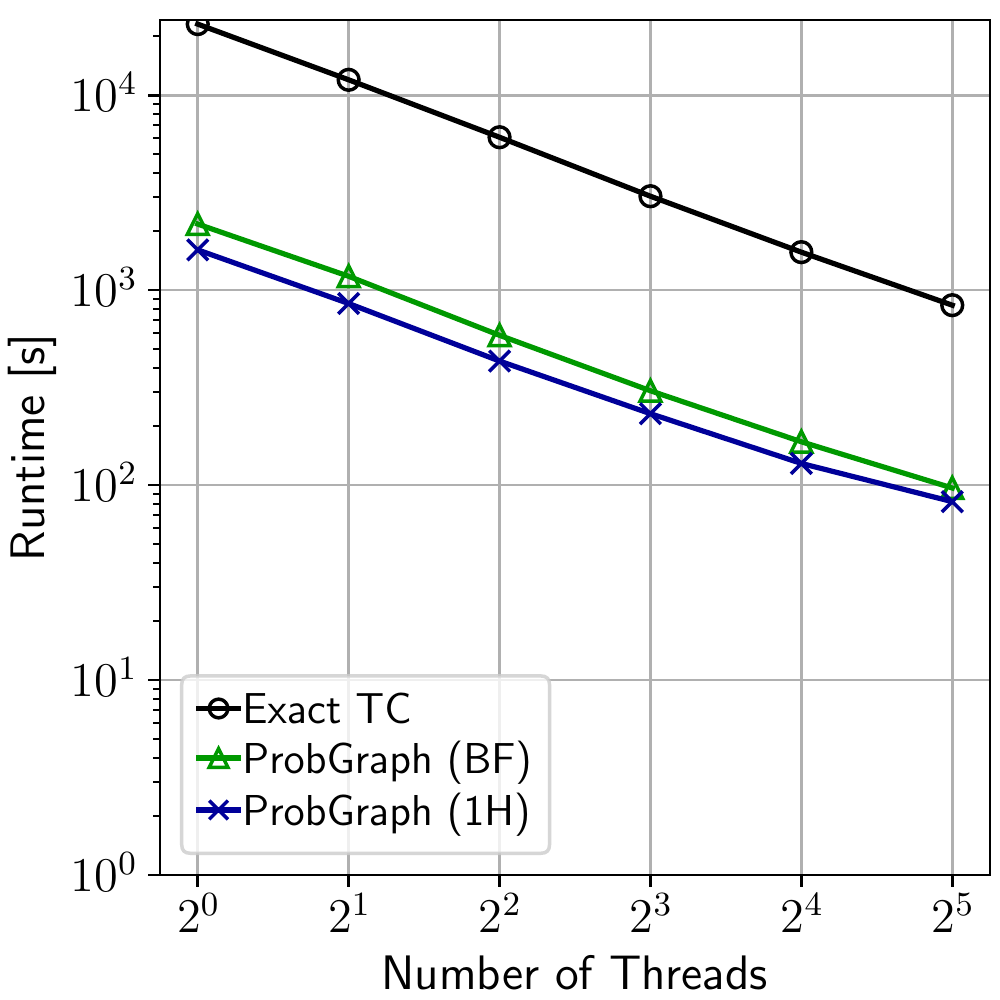}
\vspaceSQ{-1.0em}
\caption{\textmd{Strong scaling (Clustering, Jaccard Vertex Similarity).}}
\label{fig:scaling-strong-jp-jc}
\end{subfigure}
\begin{subfigure}[t]{0.22 \textwidth}
\centering
\includegraphics[width=\textwidth]{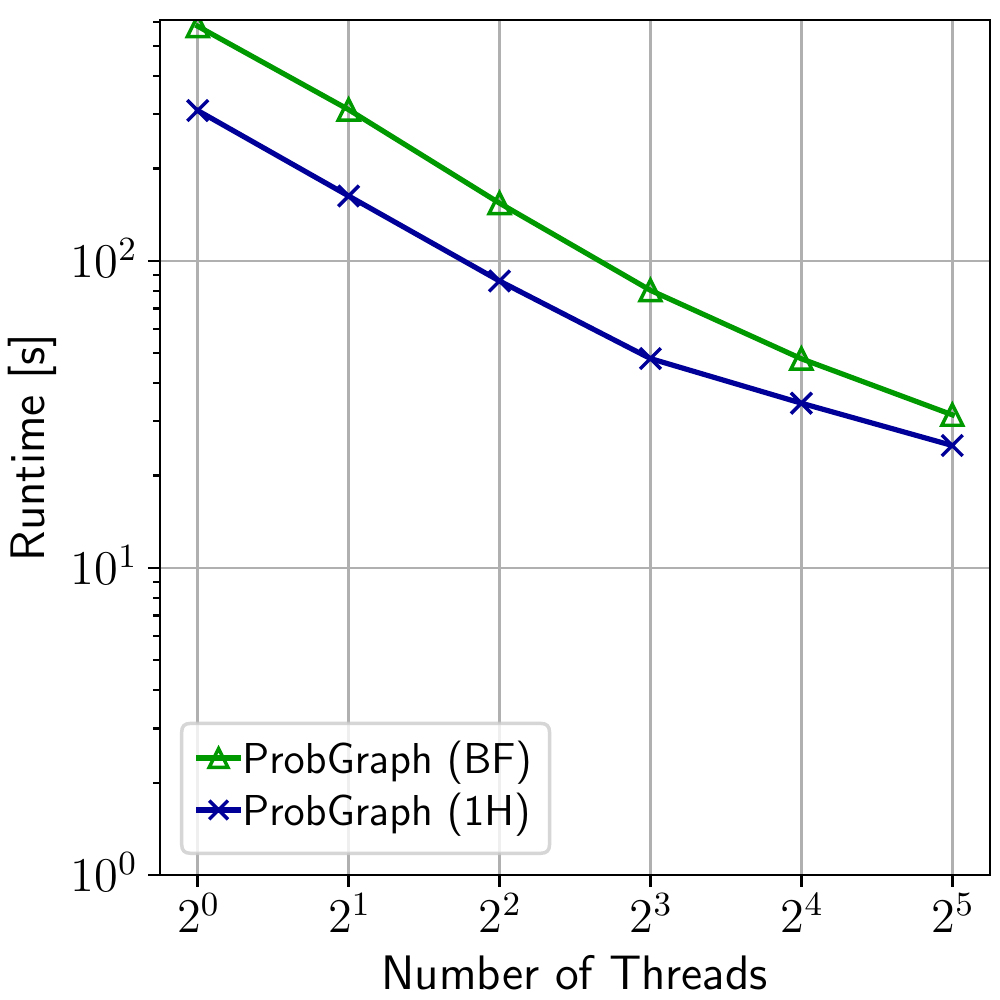}
\vspaceSQ{-1.0em}
\caption{\textmd{Strong scaling (Clustering, Overlap Vertex Similarity).}}
\label{fig:scaling-strong-jp-ov}
\end{subfigure}
%
%
%
\begin{subfigure}[t]{0.22 \textwidth}
\centering
\includegraphics[width=\textwidth]{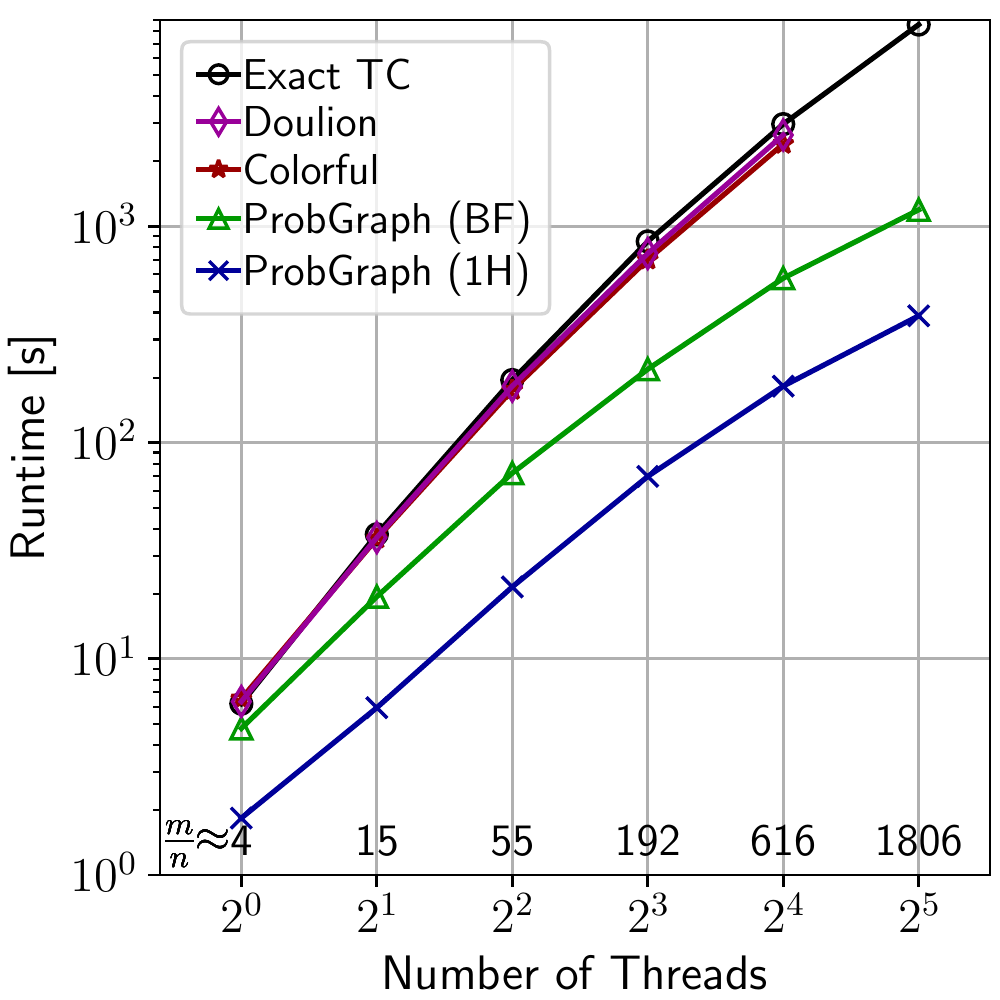}
\vspaceSQ{-1.0em}
\caption{\textmd{Weak scaling (TC).}}
\label{fig:scaling-weak-tc}
\end{subfigure}
\begin{subfigure}[t]{0.22 \textwidth}
\centering
\includegraphics[width=\textwidth]{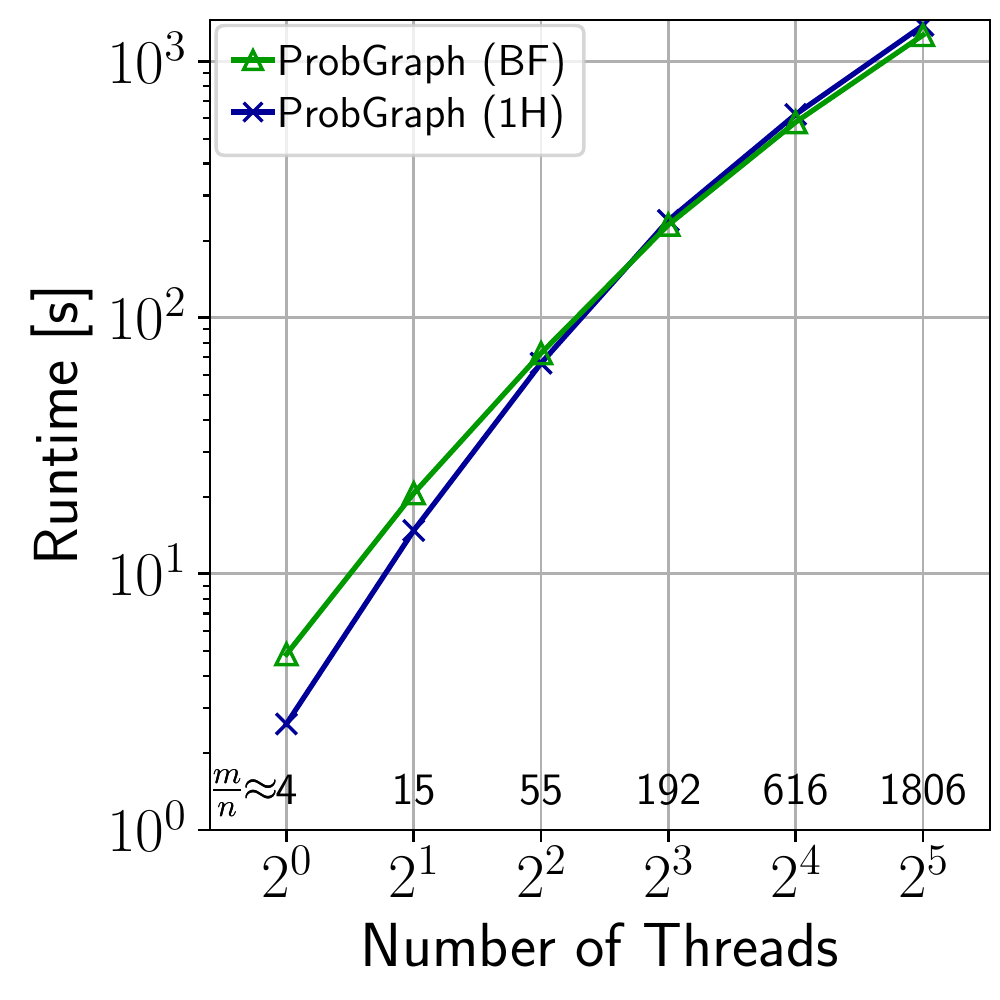}
\vspaceSQ{-1.0em}
\caption{\textmd{Weak scaling (Clustering, Common Neighbors Vertex Similarity).}}
\label{fig:scaling-weak-jp-cn}
\end{subfigure}
\begin{subfigure}[t]{0.22 \textwidth}
\centering
\includegraphics[width=\textwidth]{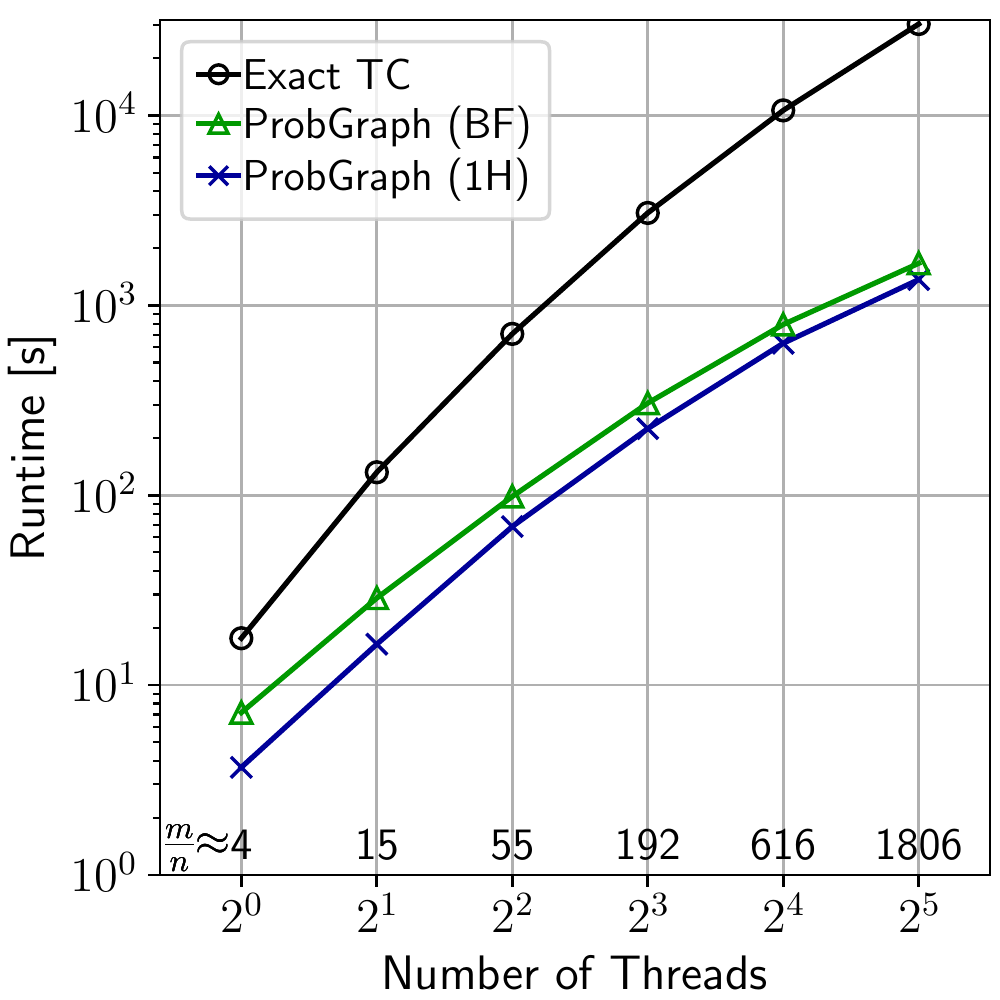}
\vspaceSQ{-1.5em}
\caption{\textmd{Weak scaling (Clustering, Jaccard Vertex Similarity).}}
\label{fig:scaling-weak-jp-jc}
\end{subfigure}
\begin{subfigure}[t]{0.22 \textwidth}
\centering
\includegraphics[width=\textwidth]{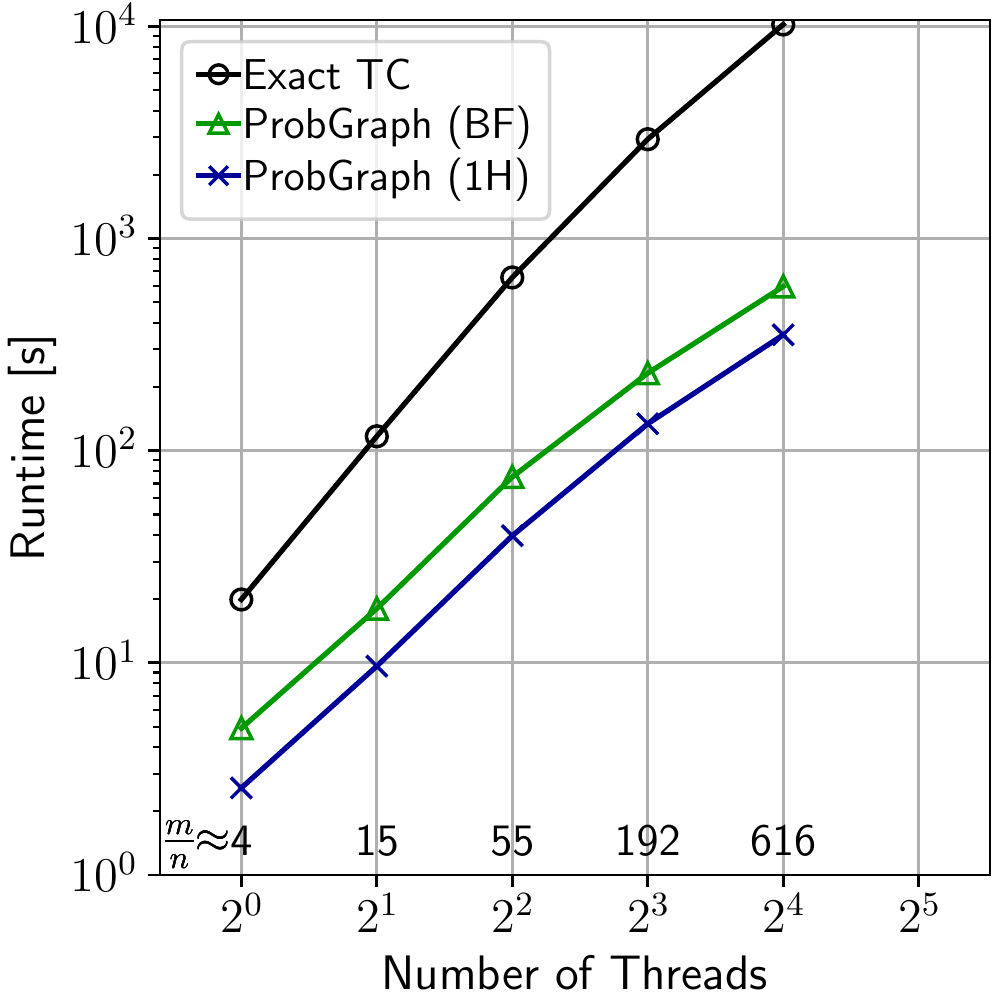}
\vspaceSQ{-1.5em}
\caption{\textmd{Weak scaling (Clustering, Overlap Vertex Similarity).}}
\label{fig:scaling-weak-jp-ov}
\end{subfigure}
\vspaceSQ{-0.5em}
\caption{\textmd{Analysis of scaling of representative baselines.}}
\label{fig:scale-analysis}
\vspaceSQ{-1em}
\end{figure*}

\fi

\begin{figure}[h]
\centering
\vspaceSQ{-1em}
%
\begin{subfigure}[t]{0.22 \textwidth}
\centering
\includegraphics[width=\textwidth]{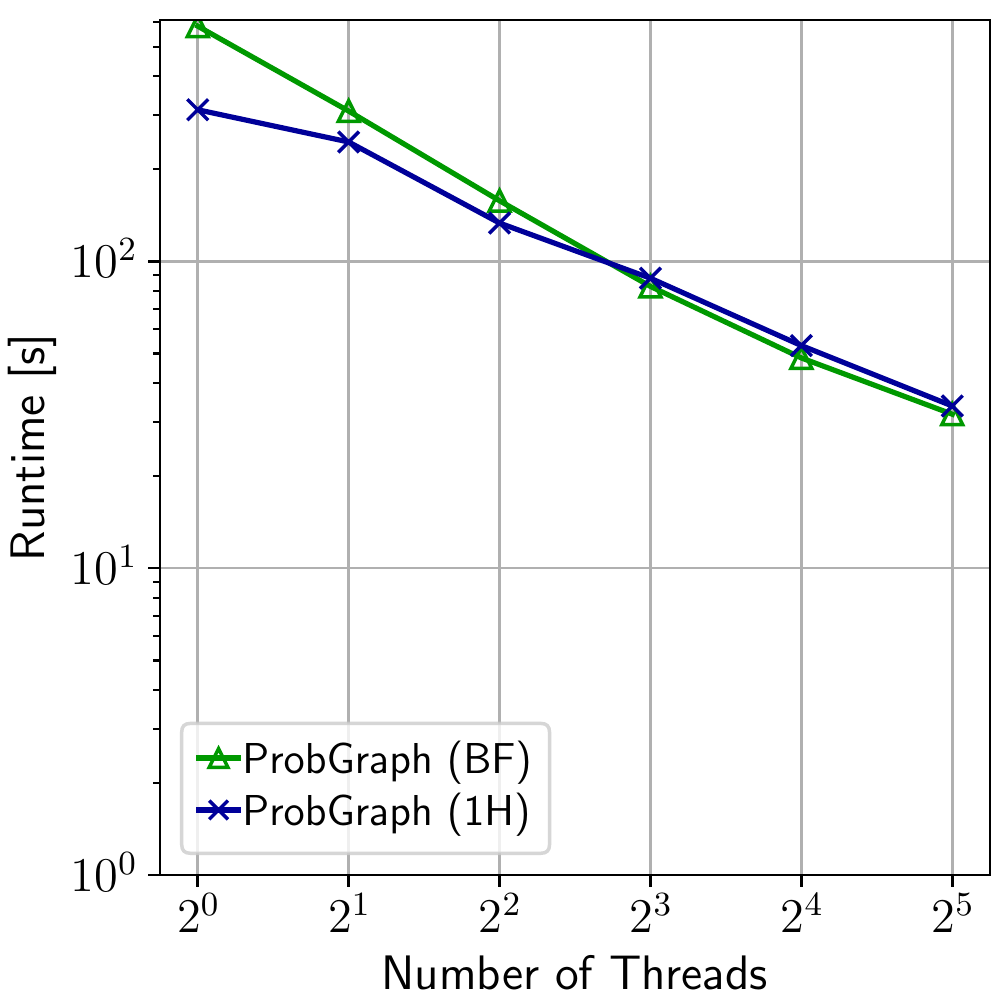}
\vspaceSQ{-1.5em}
\caption{\textmd{\hl{Strong scaling.}}}
\label{fig:scaling-strong-jp-cn}
\end{subfigure}
%
\begin{subfigure}[t]{0.22 \textwidth}
\centering
\includegraphics[width=\textwidth]{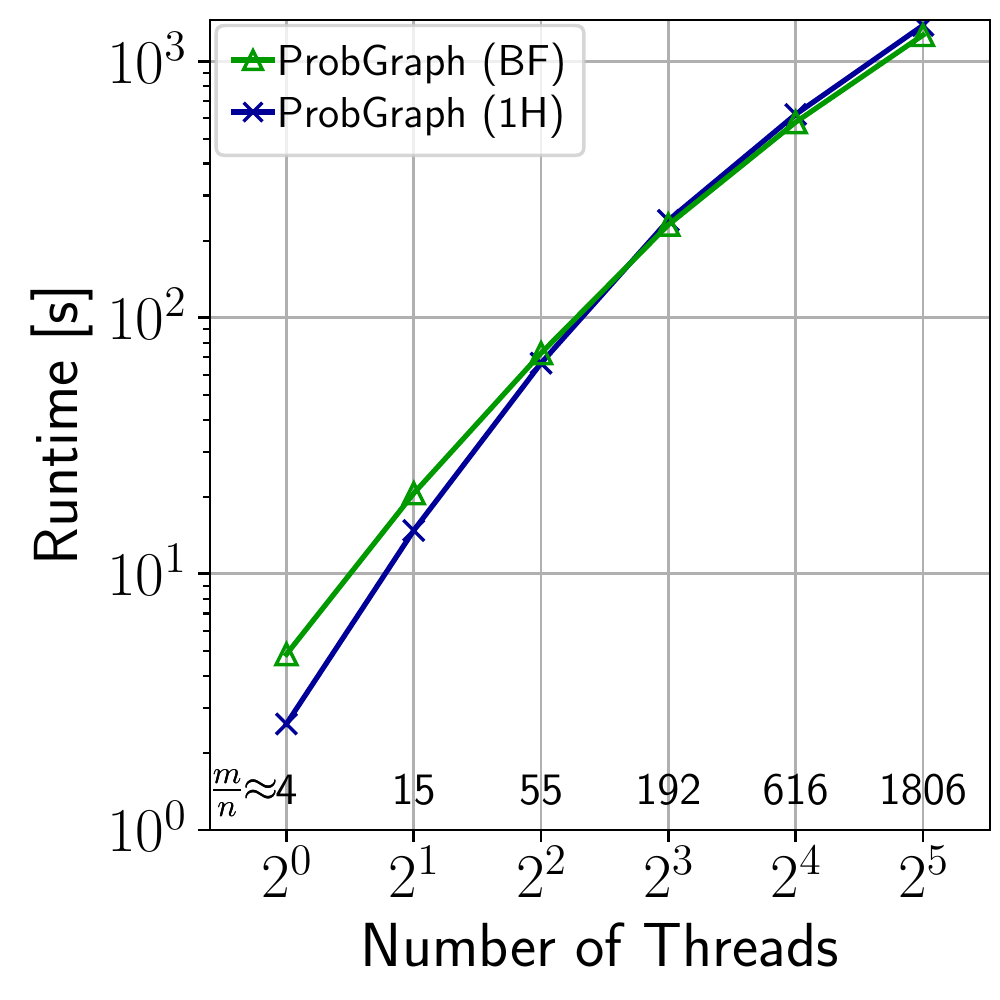}
\vspaceSQ{-1.5em}
\caption{\textmd{\hl{Weak scaling.}}}
\label{fig:scaling-weak-tc}
\end{subfigure}
\vspaceSQ{-0.5em}
\caption{\textmd{\hl{Scaling results for Clustering (Common Neighbors), illustrating
comparable scaling performance of both BF and MH.}}}
\label{fig:scale-analysis-cn}
\vspaceSQ{-1em}
\end{figure}

\marginpar{\vspace{-12em}\colorbox{yellow}{\textbf{R-3}}\\ \colorbox{yellow}{\textbf{   }}\\ \colorbox{yellow}{\textbf{   }}\\ \colorbox{yellow}{\textbf{   }}\\ \colorbox{yellow}{\textbf{   }}\\ \colorbox{yellow}{\textbf{   }}\\ \colorbox{yellow}{\textbf{   }}\\ \colorbox{yellow}{\textbf{   }}\\ \colorbox{yellow}{\textbf{   }}}


\enlargeSQ
\enlargeSQ
\enlargeSQ

\subsection{Analysis of Distributed-Memory Executions}

ProbGraph is seamlessly applicable to both shared- and distributed-memory
settings.  Due to the small sizes of neighborhood sketches, we never have to
distribute any sketch across two compute nodes. Intersecting two sketches
located on two different nodes can be done using different approaches. We
currently employ a straightforward scheme in which sketches are transferred
across the network using point-to-point message passing, in order to conduct
intersections on a single node.  This offers significant reductions in overall
communication times, compared to standard baselines, of up to 4$\times$ for
different graphs.  These reductions are a direct consequence of small sizes of
used sketches.  However, more advanced schemes based one-sided
communication~\cite{besta2015active, fompi-paper, besta2014fault} could easily be employed; we leave
them for future work.

\subsection{Analysis of Construction Costs}

We also analyze the \emph{construction costs} of PG.  Time to
construct a single neighborhood follows asymptotic complexities in
Table~\ref{tab:constr}; it is not a bottleneck and is lower than 50\% of
the algorithm execution time in the majority of cases. Only using very large
$b$ may bring the preprocessing time larger than the single graph algorithm
execution time, but (1) PG benefits from low $b \in \{1, 2\}$ in any case, and
(2) the PG representation of a graph has to be computed only once, and it can
be then freely used with any considered graph algorithms.

\section{\hl{Beyond Bloom Filter and MinHash}}

\hl{The generic nature of PG enables using other probabilistic representations in
place of BF and MH. As an example, we discuss how to use PG with \emph{K
Minimum Values} (KMV)~\mbox{\cite{bar2002counting}}, another sketch that was
originally developed to accelerate counting distinct elements in a data stream.
To construct a KMV representation~\mbox{$\mathcal{K}_X$} of a set~\mbox{$X$},
one evaluates the associated hash function \mbox{$h: X \to (0;1]$} for all
elements of~\mbox{$X$}. Then, one selects \mbox{$k$} \emph{smallest} hashes
that constitute the final KMV representation \mbox{$\mathcal{K}_X$} of the
set~\mbox{$X$}. One can then estimate \mbox{$|X|$} with
\mbox{$\widehat{|X|}_{KMV}$} \mbox{$= \frac{k-1}{\max{\mathcal{K}_X}}$}. Note
that \mbox{$\mathcal{K}_X$} differs from a MH \mbox{$\mathcal{M}_X$} because,
as opposed to \mbox{$\mathcal{M}_X$}, it contains hashes.}

\hl{Now, one can use KMV to also estimate \mbox{$|X \cap Y|$}, and then use it
within PG.  For this, one constructs a KMV \mbox{$\mathcal{K}_{X \cup Y}$} by
taking the \mbox{$k$} smallest elements from \mbox{$K_X \cup K_Y$}. Then, by
the KMV properties, we have  \mbox{$\widehat{|X \cup Y|}_{KMV}$} \mbox{$=
\frac{k-1}{\max{\mathcal{K}_{X \cup Y}}}$}. Finally, \mbox{$\widehat{|X \cap
Y|}_{KMV}$} \mbox{$= |X| + |Y| - \widehat{|X \cup Y|}_{KMV}$}, which can be
directly used in PG formulations of graph algorithms. We provide concentration bounds for all the KMV estimators defined above in the Appendix.}

\marginparsep=1em
\marginpar{\vspace{-10em}\colorbox{yellow}{\textbf{R-2}}\\ \colorbox{yellow}{\textbf{R-4}}}

\iftr
\section{Related Work: Summary}

\enlargeSQ

We summarize related work; some parts are already covered in
Sections~\ref{sec:intro} and~\ref{sec:theory}.
First, there exist more {set-related probabilistic data structures}, for
example HyperLogLog~\cite{flajolet2007hyperloglog}. \emph{ProbGraph embraces
such data structures:} while we focus on BF~\cite{bloom1970space}
and MH~\cite{broder1997resemblance}, one could
easily extend ProbGraph with other structures; we leave details for future
work.

Next, there are many approximate graph algorithms~\cite{riondato2016fast,
borassi2016kadabra, riondato2018abra, geisberger2008better,
bader2007approximating, chazelle2005approximating, dumbrava2018approximate,
iyer2018asap, besta2020substream, chechik2014better, roditty2013fast,
slota2014complex, roditty2013fast, boldi2011hyperanf, echbarthi2017lasas}.
ProbGraph differs from them as it can approximate any algorithm or scheme that
uses $|X \cap Y|$, set membership query, and others, where $X$ and $Y$ are
arbitrary sets of vertices or edges (all our theoretical and most of empirical
results are directly applicable to any sets). Moreover, ProbGraph {is simple}:
all one has to do is to plug in a selected set representation. 

A few existing general approaches for {approximate graph computations} 
usually target specific problems or they do not come with
guarantees on the quality of outcomes~\cite{shang2014auto, Rahman:2014,
singh2018scalable}. As shown in Section~\ref{sec:theory}, unlike
ProbGraph, specific {schemes for TC} do not offer strong
accuracy guarantees~\cite{tsourakakis2009doulion, pagh2012colorful,
Rahman:2014, iyer2018asap, iyer2018bridging, bandyopadhyay2016topological,
besta2019slim}.

ProbGraph's probabilistic representations of vertex sets are a form of graph
compression~\cite{besta2018survey}, and they could be used to extend existing
compressed graph representations and paradigms~\cite{besta2018log,
besta2019slim}.

There exist a few works on {using BF or MH specific single graph problems,
usually in the context of evolving graphs~\cite{ediger2010massive,
galhotra2015tracking, besta2019practice, saha2019reachability,
bandyopadhyay2016topological}, which is outside PG's scope.

Approximating the triangle count in time
less then linear in the size of the input was shown in
\cite{eden2017approximately}. This has been later generalized to
approximating the number of $k$-cliques in a graph \cite{eden2020approximating}.
Improved bounds are known when the arboricity of the graph is small
\cite{eden2020faster}. Assuming we can sample edges uniformly, better
algorithms are also known \cite{assadi2018simple,biswas2021towards}.
Unlike PG, these schemes are specific to selected graph problems and
graphs with certain properties such as low arboricity.

There are many works on counting or finding different graph patterns
(also called motifs, graphlets, or
subgraphs)~\cite{besta2021graphminesuite, besta2021sisa, chakrabarti2006graph,
washio2003state, lee2010survey, rehman2012graph, gallagher2006matching,
ramraj2015frequent, jiang2013survey, aggarwal2010managing, tang2010graph,
leicht2006vertex, besta2017push, liben2007link, ribeiro2019survey, lu2011link,
al2011survey, bron1973algorithm, cazals2008note, DBLP:conf/isaac/EppsteinLS10,
DBLP:journals/tcs/TomitaTT06, danisch2018listing, jabbour2018pushing, besta2021motif}.
PG can be used as a subroutine in different such works, offering speedups
while preserving high accuracy. 

Counting and listing simple patterns such as triangles have been recently used
to enhance the design of numerous models in Graph Neural
Networks~\cite{wu2020comprehensive, zhou2020graph, scarselli2008graph,
zhang2020deep, chami2020machine, hamilton2017representation,
bronstein2017geometric, kipf2016semi, gianinazzi2021learning,
bronstein2021geometric, wu2020comprehensive, zhou2020graph, besta2022parallel,
zhang2020deep, chami2020machine, hamilton2017representation,
bronstein2017geometric, zhang2019graph}. Such models could use PG to
accelerate expensive graph mining preprocessing costs. 

The straightforward parallelism in computing BF based estimators
implies that other architectures that offer massive parallelism may provide
even higher benefits. This includes FPGAs~\cite{de2018designing, kuon2008fpga,
besta2019graph, besta2020substream}, CGRAs~\cite{cong2014fully}, or
processing in-memory~\cite{mutlu2019, mutlu2020modern, ghose2019processing,
seshadri2013rowclone, gomez2021benchmarking, oliveira2021damov,
hajinazar2021simdram, seshadri2017ambit, ahn2015scalable_tes}. We leave these
studies for future work.

\iftr
Next, there are many \textbf{approximate graph
algorithms}~\cite{halldorsson1993still, de2004approximate, khot2008vertex,
de2004approximate, wang1995algorithms, riondato2018abra, gianinazzi2018communication}. ProbGraph differs
from them as \emph{it is general}: it can approximate any algorithm or scheme
that uses $|X \cap Y|$, set membership query, and others, where $X$ and $Y$ are
arbitrary sets of vertices or edges (all our theoretical and most of empirical
results are directly applicable to any sets). Moreover, ProbGraph \emph{is
simple}: all one has to do is to plug in a selected set representation and
implementations of $|X \cap Y|$, a set membership query, and any other related
schemes.
\fi

\fi
\section{Conclusion}

\enlargeSQ
\enlargeSQ
\enlargeSQ

We propose ProbGraph, a parallel graph representation that enables simple,
general, and high-performance approximate graph computations.
The key idea is to sketch sets of vertices, and the cardinality of the
intersection of such sets, with probabilistic set representations such as Bloom
filters or MinHash. Such representations usually offer much higher performance
than exact set representations, while only requiring small additional storage.
Importantly, they can be treated as a black box and seamlessly incorporated
into graph algorithms. We show that ProbGraph is simple to use while offering
speedups of more than 50$\times$ for some graphs and retaining high accuracy of
more than 90\% for problems such as Triangle Counting, when comparing to tuned
exact parallel baselines on 32 cores.

We support ProbGraph with an in-depth theoretical underpinning, in which we
derive novel statistical concentration bounds on the accuracy of ProbGraph
approximations. Our bounds are the first exponential or polynomial quality
bounds for the accuracy of Bloom filters and MinHash.  As such, they are of
interest to the broader audience beyond graph analytics.
We also use the work-depth formal analysis to show that ProbGraph has also theoretical
advantages over parallel baselines in parallel computational complexity.

Set algebra is common in many graph problems. Hence, we expect that ProbGraph
and its set-centric approach for approximate graph analytics may be used for
other problems.

\vspace{1em}

%

\macb{Acknowledgements:}
We thank Hussein Harake, Colin McMurtrie, Mark Klein, Angelo Mangili, and the
whole CSCS team granting access to the Ault and Daint machines, and for their
excellent technical support.
We thank Timo Schneider for immense help with computing infrastructure at SPCL.
We thank Aryaman Fasciati and Sebastian Leisinger for help
in the early stages of the project.
This research received funding from 
the European Research Council
\includegraphics[height=1em]{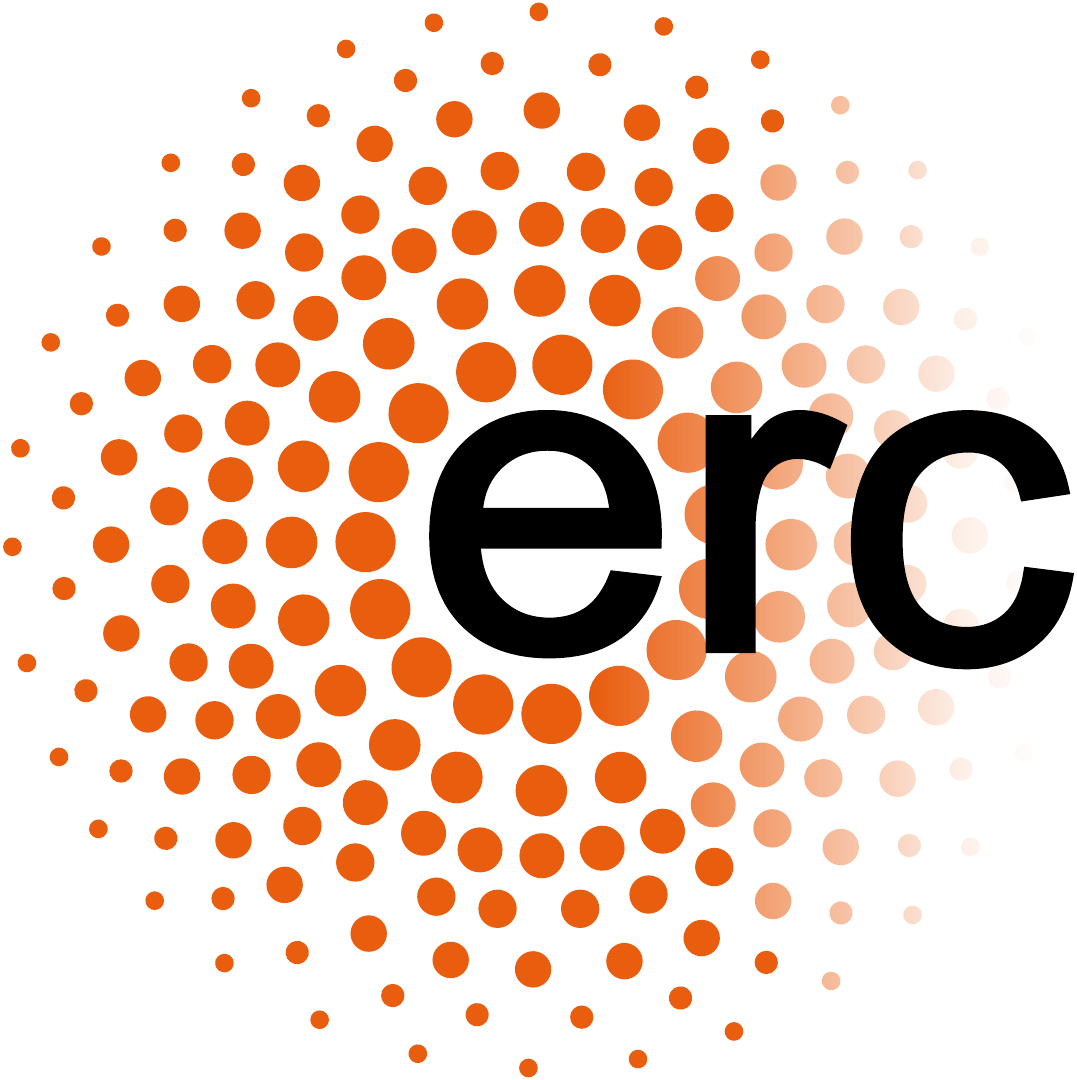}
(Project DAPP, No.~678880; Project
PSAP, No.~101002047), Huawei, and MS Research
through its Swiss Joint Research Centre.
Jakub Tětek was supported by the VILLUM Foundation grant 16582.

\normalsize


\bibliographystyle{IEEEtran}
\bibliography{references}

\iftr 
\appendix

\vspaceSQ{2em} 

Here, we provide proofs and various details omitted from the main part of the manuscript.

\subsection{Probability Distributions}

Consider a sequence of $n$ trials (experiments).  A single trial is a Bernoulli
trial, i.e., it gives a \emph{success} or a \emph{failure} outcome with a
probability $p$ or $1-p$, respectively.  In the context of PG, a single
trial will correspond to some property of a given set representation, for
example whether a given bit in a BF has 1 or 0.
Now, if all the trials are independent (i.e., obtaining a specific outcome does
not impact the number of such potential outcomes in future trials), the
resulting distribution is \textbf{binomial} (commonly denoted as $Bin (n, p)$).
Otherwise (i.e., obtaining a specific outcome decreases by one the number of
such potential outcomes in future trials), it is \textbf{hypergeometric}
(usually denoted as $Hyper (N, K,n)$). Both distributions enable deriving specific \emph{probabilities} for the number of either success or failure trials.

\subsection{Plug-In Principle}

Some concentration bounds that we present in this paper depend on a given set 
size (i.e., $|X|$ appears in the formulation of the bound). Thus if we want to
obtain an estimate of the upper bound, we need to substitute the
estimator~$\widehat{|X|}$ instead of~$|X|$ whenever $|X|$ appears. This
procedure is known as the plug-in principle and it is well established in
statistics. However this method is safe to use only if the estimator we substitute
is at least consistent for the parameter of interest (i.e. $|X|$). Indeed this is the
case for all the BF and MinHash estimators presented in this paper.

\subsection{BF Sketches for Single Sets}\label{app_single_sets}

We provide extended results for BF for single sets.

\subsubsection{Concentration Bound for BF Single Sets}

\begin{prop}\label{bound_ss_BF}
	Let $\widehat{|X|}_S$ be the estimator defined in
	Eq.~(\ref{eq:est_ss_bf}). For $B_X,b \in \mathbb{N}$ such that $b = o(\sqrt{B_X})$, and a set $X$ such that $b |X| \leq 0.499 B_X \log B_X$ the following holds:
	
	\small
	
	\[
	E\fSB{\fRB{\widehat{|X|}_S - |X|}^2} \leq (1+o(1))\left(e^{\frac{|X| b }{ (B_X-1)}} \frac{B_X}{b^2} - \frac{B_X}{b^2} - \frac{|X|}{b}\right)
	\]
	\normalsize
\end{prop}

\begin{proof}
	We now prove Proposition~\ref{bound_ss_BF}. Before bounding the mean squared error of $\widehat{|X|}_S$, we need to prove several simple bounds. Let $\mu = E[B_{0,X}] = B_X\left(1-\frac{1}{B_{X}}\right)^{b|X|}$. It holds:
	
	\begin{align*}
		\mu &\geq B_X\left(1-\frac{1}{B_{X}}\right)^{0.499 B_X \log B_X} \\&\geq B_X \exp\fRB{- \frac{0.499 \log B_X}{1-\frac{1}{B_X}}} \\
		&= B_X^{0.501-o(1)}
	\end{align*}
	
	Let us fix some $\epsilon>0$. Let $\mathcal{E}$ be the event that $B_{0,X} \geq \mu/(1+\epsilon)$. \cite[Theorem~2]{Kamath1995} prove that:\\
	$P(\mathcal{E}) \geq 1 - \exp(-\Omega(\mu^2/B_X)) \geq 1 - \exp\fRB{-B_X^{\Omega(1)}}$.\\
	We have $\widehat{|X|} = -\frac{B_X}{b} \log (B_{X,0}/B_X + \mathbb{I}[B_{X,0} = 0]) \leq B_X \log B_X$ and by our assumption, $|X| \leq b |X| \leq 0.499 B_X \log B_X$. It thus holds $(\widehat{|X|} - |X|)^2 \leq O(B_X^2 \log^2 B_X)$. Let $\kappa = - \frac{B_X}{b} \log \left(1-\frac{1}{B_{X}}\right)^{b|X|} = - B_X |X| \log \left(1-\frac{1}{B_{X}}\right)$. Moreover for $B_X \rightarrow \infty$, we have $\log(1-1/B_X) = -1/B_X + O(1/B_X^2)$. Therefore, it holds $\kappa = |X| + o(1)$.\\

	Now we are able to bound the mean squared error of $\widehat{|X|}_S$. We present each step of the derivation as a unique figure (see Fig. \ref{fig:app:bound_ss_BF} below) to improve the clarity of the content.\\

\begin{figure*}[t]
\begin{align}
		E[&(\widehat{|X|} - |X|)^2]\\ 
		=& E[(\widehat{|X|} - |X|)^2 | \mathcal{E}] P(\mathcal{E}) + E[(\widehat{|X|} - |X|)^2 | \neg \mathcal{E}] P(\neg \mathcal{E}) \\
		\leq& (1+\epsilon)E[(\widehat{|X|} - \kappa)^2 | \mathcal{E}] + \frac{1+\epsilon}{\epsilon}E[(\kappa -|X|)^2 | \mathcal{E}] + O(B_X^2 \log^2 B_X) \cdot \exp(-B_X^{\Omega(1)}) \label{eq:sort_of_triangle_ineq}\\
		\leq& \frac{(1+\epsilon) B_X^2}{b^2} E[(\log(B_{X,0}/B_X) - \log(1-1/B_X)^{b |X|})^2 | \mathcal{E}] + O((\kappa -|X|)^2) + \exp(-B_X^{\Omega(1)}) \\
		\leq& \frac{(1+\epsilon) B_X^2}{b^2} E[(\log(B_{X,0}/B_X) - \log(1-1/B_X)^{b |X|})^2 | \mathcal{E}] + O(|X|/B_X)\\
		\leq& \frac{(1+\epsilon)^2 B_X^2}{b^2} e^{2b |X| / B_X} E[(B_{X,0}/B_X - (1-1/B_X)^{b |X|})^2 | \mathcal{E}] + O(|X|/B_X) \label{eq:log_is_lipschitz}\\
		\leq& \frac{(1+\epsilon)^2 B_X^2}{b^2} e^{2b |X| /( B_X-1)} \cdot E[(B_{X,0}/B_X - (1-1/B_X)^{b |X|})^2]/P[\mathcal{E}] + O(|X|/B_X)\\
		=& \big((1+\epsilon)^2+o(1)\big)\frac{B_X^2}{b^2} e^{2b |X| /( B_X-1)} \cdot E[(B_{X,0}/B_X - (1-1/B_X)^{b |X|})^2] + O(|X|/B_X)\\
		=& \big((1+\epsilon)^2+o(1)\big)\frac{e^{2b |X| / (B_X-1)}}{b^2} Var[B_{X,0}]  + O(|X|/B_X)\label{eq:rewrite_as_variance}\\
		\leq& \big((1+\epsilon)^2+o(1)\big) e^{2b |X| / (B_X-1)} \cdot \left(e^{-\frac{b |X|}{B_X}}\frac{B_X}{b^2} - B_X/b^2 - |X|/b\right)  + O(|X|/B_X)\label{eq:balls_and_bins_variance}\\
		\leq&  \big((1+\epsilon)^2+o(1)\big)\left(e^{|X| b / (B_X-1)} \frac{B_X}{b^2} - B_X/b^2 - |X|/b\right) + O(|X|/B_X)\\
		\leq&  \big((1+\epsilon)^2+o(1)\big)\left(e^{|X| b / (B_X-1)} \frac{B_X}{b^2} - B_X/b^2 - |X|/b\right) \label{eq:removing_error_term}
	\end{align}
\caption{Detailed steps of the derivation of an upper bound for the mean squared error of the BF estimator of the single set size.}
\label{fig:app:bound_ss_BF}
\end{figure*}

	In particular, \cref{eq:sort_of_triangle_ineq} in Figure \ref{fig:app:bound_ss_BF} holds because for any $a,b,c \in \mathbb{R}$ and $\epsilon>0$, it holds\footnote{This inequality is equivalent to $(1+\epsilon)(a-c)^2 + \frac{1+\epsilon}{\epsilon}(c-b)^2 - (a-b)^2 \geq 0$. The left-hand side can be simplified to $\frac{(\epsilon a+b-c (1+\epsilon))^2}{\epsilon}$ and the inequality thus holds.} $(a-b)^2 \leq (1+\epsilon)(a-c)^2 + \frac{1+\epsilon}{\epsilon}(c-b)^2$. Eq. (\ref{eq:log_is_lipschitz}) holds because on $\mathcal{E}$, given $B_{X,0} \geq \mu/(1+\epsilon)$, $\log(B_{X,0}/B_X)$ is $c$-lipschitz for $c = (1+\epsilon)B_X/\mu \leq (1+\epsilon)e^{\frac{2 b |X|}{B_X(1-1/B_X)}} = (1+\epsilon)e^{\frac{2 b |X|}{B_X-1}}$. 
	Eq. (\ref{eq:rewrite_as_variance}) holds because $E[B_{X,0}/B_X] = (1-1/B_X)^{b |X|}$ and \cref{eq:balls_and_bins_variance} holds because
	$Var(B_{X,0})\sim B_X e^{-\frac{b |X|}{B_X}}-B_X \left(\frac{b |X|}{B_X}+1\right) e^{-\frac{2 b |X|}{B_X}}$ \cite{stackexchange}.
	By sending $\epsilon \rightarrow 0$, we get that\footnote{It is well known that if $f(x) \leq (1+\epsilon) g(x)$ for any $\epsilon>0$, then $f(x) \leq (1+o(1))g(x)$.}:
	
	\small
	\[
	E[(\widehat{|X|} - |X|)^2] \leq (1+o(1))\left(e^{\frac{|X| b }{ (B_X-1)}} \frac{B_X}{b^2} - \frac{B_X}{b^2} - \frac{|X|}{b}\right)
	\]
	\normalsize
\end{proof}

\vspaceSQ{-0.5em}
\subsubsection{Class of Estimators with General Bounds}
\label{sec:single-class}

The bound on the MSE presented in Proposition~\ref{bound_ss_BF} holds up to some assumptions (i.e. $b = o(\sqrt{B_{X}})$ and $b |X| \leq 0.499 B_{X} \log B_{X}$) and an $o(1)$ term. To derive a concentration bound for the MSE which does not depend on these conditions and that enhance the interpretability, we develop a \emph{class of
estimators} which encompasses the one by Swamidass et al.~\cite{swamidass2007mathematical}. To introduce this framework, we first propose a new limiting estimator called $\widehat{|X|}_L$ which belongs to this class. We obtain $\widehat{|X|}_L$ by simplifying the estimator from Eq.~(\ref{eq:est_ss_bf}) and taking the limit for $B_X \rightarrow \infty$:

\vspaceSQ{-0.5em}
\ifsqEQ\small\fi

\begin{align} 
  \widehat{|X|}_L &\equiv \lim_{B_X\to\infty} \widehat{|X|}_S = \lim_{B_X\to\infty}- \frac{B_X}{b} \log \left ( 1 - \frac{B_{X,1}}{B_X} \right ) \nonumber \\
  &= \log \left ( \lim_{B_X\to\infty} \left(1 - \frac{B_{X,1}}{B_X}\right)^{-\frac{B_X}{b}} \right)
  \nonumber \\
  & = \log\left(\exp \left(\frac{B_{X,1}}{b}\right)\right) =\frac{B_{X,1}}{b} \label{eq:limit_xs}
\end{align}

\normalsize

We can perform this simplification thanks to the continuity of the logarithm in
$(0,\infty)$ that allows us to safely move the limit inside log, knowing that
$B_X,b \in \mathbb{N}$ by construction. This result tells us that, as $B_X$
increases, $\widehat{|X|}_S$ \emph{rescales the number of ones in the BF} by
the quantity $\frac{1}{b}$ because $\widehat{|X|}_S \sim \frac{B_{X,1}}{b}$ for $X,b$ fixed and $B_X \rightarrow \infty$.
We can also prove that $\widehat{|X|}_S \leq \frac{\log B_X}{b} B_{X,1}$ thus implying that $\widehat{|X|}_S$ can \emph{inflate} the number of ones \emph{at most} by the factor $\frac{\log B_X}{b}$. These interesting insights motivate us to propose a general class of estimators. The key idea is to define any BF estimator as a \emph{function of a random variable} (i.e. $B_{X,1}$, the number of ones in a BF). Specifically, we have $\widehat{|X|}_\bullet \equiv \delta_{B_X, b}(B_{X,1})$, where $\delta(\cdot)$ is a given non-negative
function of $B_X$, $b$, and $B_{X,1}$. We choose to denote $\delta_{B_X, b}(B_{X,1})$ and $\delta_{B_X, b}$ instead of the usual $\delta(B_X,b,B_{X,1})$ and $\delta(B_X,b)$ to clearly separate the deterministic BF design parameters
$B_X$ and $b$ from the unique random component $B_{X,1}$. The key benefit of this formulation is that (1) it generalizes both $\widehat{|X|}_S$ and $\widehat{|X|}_L$, and (2) we can use it to provide concentration bounds that are applicable to $\widehat{|X|}_L$, and {many} other estimators within the proposed class depending on the functional form of $\delta_{B_X, b}(B_{X,1})$. To obtain $\widehat{|X|}_S$, we set: 

\vspaceSQ{-0.5em}
$$\widehat{|X|}_S \equiv \delta_{B_X, b}(B_{X,1}) = \frac{B_X}{b} \log \left
( 1 - \frac{B_{X,1}}{B_X} \right ).$$

To recover $\widehat{|X|}_L$, we first (with a slight abuse of notation) fix $\delta_{B_X, b}(B_{X,1})$ to be linear in $B_{X,1}$ and then set it to be specifically equal to $\frac{1}{b}$:

\begin{equation}\label{eq:est_simple_bf}
  \widehat{|X|}_L \equiv \delta_{B_X, b} \cdot B_{X,1} = \frac{B_{X,1}}{b}
\end{equation}     

We underline that if $\delta_{B_X, b}(B_{X,1})$ is linear in $B_{X,1}$ we are
implicitly imposing, depending on the values of $B_X$ and $b$, either a
\emph{deflation} or an \emph{inflation} of the observed \emph{number of ones}
in the BF. For example, we have already seen that, when $B_X\to\infty$ for fixed $X,b$, we have
$\delta_{B_X, b}(B_X) = \frac{B_X}{b}$ in Eq.~(\ref{eq:est_ss_bf}). We now
show that any estimator that can be written as $\widehat{|X|}_\bullet$ with
$\delta_{B_X, b}(B_{X,1})$ linear in $B_{X,1}$ has a bounded MSE. 

\begin{prop}\label{bound_ss_simple_BF}
  Let $\widehat{|X|}_\bullet \equiv \delta_{B_X, b} \cdot B_{X,1}$. For $B_X,b \in \mathbb{N}$, the following holds:
\vspaceSQ{-0.5em}
  
\begin{equation}
\begin{aligned}
  E\fSB{\fRB{\widehat{|X|}_\bullet - |X|}^2}  \leq \left[|X| - \;\delta_{B_X, b}\; B_X \left(1 - e^{-\frac{|X| b}{B_X}}\right)\right]^{2} \nonumber \\
  + \; \delta_{B_X, b}^{2} \; B_X \left[e^{-\frac{|X| b}{B_X}} - \left(1 + \frac{|X| b}{B_X}\right)e^{-\frac{2|X| b}{B_X}}  \right] \nonumber
\end{aligned}
\end{equation}
  \normalsize
\end{prop}

\noindent
We use Chebyshev's inequality to get the final
concentration bound: 

\vspaceSQ{-0.5em}
\small
\begin{equation*}\label{cbound_ss_simple_BF}
\begin{aligned}
P \left( \left| \widehat{|X|}_\bullet - |X| \right| \geq t \right) \leq \frac{\left[|X| - \;\delta_{B_X, b}\;B_X \left(1 - e^{-\frac{|X| b}{B_X}}\right)\right]^{2}}{t^2} \\ 
+ \frac{\delta_{B_X, b}^{2} \; B_X \left[e^{-\frac{|X| b}{B_X}} - \left(1 + \frac{|X| b}{B_X}\right)e^{-\frac{2|X| b}{B_X}}  \right] }{t^{2}} 
\end{aligned}
\end{equation*}
\vspaceSQ{-0.5em} 
\normalsize

By fixing $\delta_{B_X, b} = \frac{1}{b}$, we obtain a valid bound for
$\widehat{|X|}_L$ which is the limiting estimator we present in our evaluation (Section~\ref{sec:eval}).

\begin{proof}
	We provide a proof of proposition \ref{bound_ss_simple_BF}. We start by the well known MSE decomposition: 
	\small
	\begin{equation}\label{app_mse_simple_bf}
		E\fSB{\fRB{\widehat{|X|}_\bullet - |X|}^2} = E\fSB{\fRB{\widehat{|X|}_\bullet - |X|}}^{2} +\; Var(\widehat{|X|}_\bullet)
	\end{equation} 
\normalsize
	Now notice that $E[B_{0,X}] = B_X\left(1-\frac{1}{B_{X}}\right)^{b|X|}$. Then, since $\widehat{|X|}_\bullet = \delta_{B_X, b} \; B_{X,1}$, we can easily derive:
	\begin{align*}
		E\left[\delta_{B_X, b} \; B_{X,1}\right] &= E\left[\delta_{B_X, b} \; (B_{X} - B_{X,0})\right]\\ 
		&= \delta_{B_X, b} \; B_X \left[1 - \left(1 - \frac{1}{B_X} \right)^{b|X|}\right]
	\end{align*}
	
	On the other hand, to bound the variance of the simplified estimator, we follow the same reasoning outlined in Proposition \ref{bound_ss_BF}. Indeed it holds that $Var(B_{X,0}) \sim B_X \left[e^{-\frac{|X| b}{B_X}} - \left(1 + \frac{|X| b}{B_X}\right)e^{-\frac{2|X| b}{B_X}}  \right] $ as shown in \cite{stackexchange}. Now notice that $Var(B_{X,1}) = Var(B_{X} - B_{X,0}) = Var(B_{X,0})$. At this point we can substitute in eq. (\ref{app_mse_simple_bf}) the squared bias and variance of $\widehat{|X|}_\bullet$ to conclude that:
	
	\footnotesize
	\begin{equation}
		\begin{aligned}
			E\fSB{\fRB{\widehat{|X|}_\bullet - |X|}^2}  \leq \left\{|X| - \delta_{B_X, b} \; B_X \left[1 - \left(1 - \frac{1}{B_X} \right)^{b|X|}\right]\right\}^{2} \nonumber \\
			+ \delta_{B_X, b}^{2} \; B_X \left[e^{-\frac{|X| b}{B_X}} - \left(1 + \frac{|X| b}{B_X}\right)e^{-\frac{2|X| b}{B_X}}  \right] \nonumber
		\end{aligned}
	\end{equation}  
	
	\normalsize
	
	which ends the proof. To improve the interpretability of the bound, we use the fact that $\left(1 - \frac{1}{B_X} \right)^{b|X|} \sim e^{-\frac{|X| b}{B_X}}$ in the statement of Proposition~\ref{bound_ss_simple_BF}. 
	
\end{proof}

\vspaceSQ{0.25cm}
\subsubsection{Enhancing the Estimator by Swamidass~\cite{swamidass2007mathematical}}

The estimator by Swamidass et al.~\cite{swamidass2007mathematical},
is \emph{divergent}\footnote{\scriptsize An estimator whose moments are not
finite. In the case of the estimator Swamidass et
al.~\cite{swamidass2007mathematical}, the expectation of $\widehat{|X|}$, and
thus also the higher moments, diverge, which happens for $B_{X,1} = B_X$} in
its original form.
To alleviate this, we replace $B_{X,1}$ with $\widetilde{B}_{X,1} \equiv
B_{X,1} - \mathbb{I}[B_{X,1} = B_X]$, where, for a given proposition~$P$,
$\mathbb{I}[P]$ is $1$ if $P$ holds, and $0$ otherwise.  $\widetilde{B}_{X,1}$
only differs from $B_{X,1}$ in the unlikely case of $B_{X,1} = B_X$.
Thanks to this modification, our estimator~$\widehat{|X|}$ has, unlike
Swamidass et al.'s, a finite expectation (as it is bounded).

The final form of the estimator is

\begin{equation*}\label{eq:est_ss_bf___e}
\widehat{\widetilde{|X|}} = - \frac{B_X}{b} \log \left ( 1 - \frac{\widetilde{B_{X,1}}}{B_X} \right )
\end{equation*}

\subsubsection{Proof of consistency and asymptotic unbiasedness}

We need to show that $\widehat{|X|}_S = - \frac{B_X}{b} \log \left ( 1 - \frac{B_{X,1}}{B_X} \right )$ is consistent and asymptotically unbiased as $B_X\to\infty$. We provide here an intuitive formulation based on the false positive probability which can be easily made more rigorous by direct application of the definition of consistency which we omit for the sake of simplicity. First of all, as shown in eq.(\ref{eq:est_simple_bf}), we can notice that $\widehat{|X|}_S \sim \widehat{|X|}_L$ as the Bloom Filter size diverges. This means that the proof is valid for both estimators because they are asymptotically equivalent. Now we can look at the probability of false positives as $B_X \rightarrow \infty$ for fixed and finite $b$ and $|X|$: 

$$\lim_{B_X\to\infty} \left[1 -\left(1-\frac{1}{B_{X}}\right)^{b|X|}\right]^{b} \to 0$$

The result above tells us that false positive matches cannot happen anymore in the limit. Each element of $|X|$ will then be hashed in a \emph{personal} bit and counting the number of ones in $B_X$ (and dividing by $b$ in case of multiple hash functions) will always deliver $|X|$ at a given precision as $|X|$ is fixed and $B_X \rightarrow \infty$. Thus we can conclude that $\frac{B_{X,1}}{b} \overset{p}{\to} |X|$ which proves consistency. Asymptotic unbiasedness follows from consistency in our case as the variance of both estimators is bounded (see the proof of Proposition \ref{bound_int_BF}). The same reasoning can be easily extended to show consistency and asymptotic unbiasedness also for $\widehat{|X \cap Y|}_{AND}$ and $\widehat{|X \cap Y|}_{OR}$ presented in section \ref{sec:int_bf_est}.      

\subsection{Proposition~\ref{bound_int_BF}}

\begin{proof}
	To prove Proposition~\ref{bound_int_BF} from Section~\ref{sec:int_bf_est} we can extend in a straightforward way the proof presented for Proposition \ref{bound_ss_BF}. Indeed we just need to substitute $|X|$ with $|X \cap Y|$ and $B_{X}$ with $B_{X \cap Y}$ to obtain the desired result. 
	\end{proof}

\vspaceSQ{-0.5em}
\subsection{MinHash Sketches for Set Intersection}\label{sec:app:int-minhash}

\subsubsection{Expectation formula}

Since in the case of $k$-hash, $|M_X \cap M_Y| \sim Bin (\;k\;,\;J_{X,Y}\;)$, and for $1$-hash, $| M_X^1 \cap M_Y^1| \sim Hypergeometric (|X \cup Y|,|X \cap Y|,k)$, we have: 

\vspaceSQ{-0.5em}
\footnotesize\begin{gather}\label{exp_int_mh}
	\E[\widehat{|X \cap Y|}_{kH}] = (|X|+|Y|) \sum\limits_{s=0}^{k} {k \choose s} (J_{X,Y})^s (1-J_{X,Y})^{k-s} \frac{s}{k + s}\\
	\E[\widehat{|X \cap Y|}_{1H}] = (|X|+|Y|) \sum\limits_{s=0}^{k} \frac{{ |X \cap Y| \choose s} { |X \cup Y| - |X \cap Y| \choose k - s}}{{ |X \cup Y| \choose k}} \frac{s}{k + s}
\end{gather}
\normalsize
\vspaceSQ{-0.5em}

There exists an involved closed form expression for equation
(\ref{exp_int_mh}) which is beyond the scope of this paper. We refer the
interested reader to \cite{chao1972negative} for a clear derivation of a
similar problem.

\subsubsection{Proof of consistency and asymptotic unbiasedness}
We start to show that $\widehat{|X \cap Y|}_{kH}$ is consistent. This follows respectively from Proposition \ref{bound_int_mh} statement. Indeed by taking the limit for $k \rightarrow \infty$ with fixed and finite $|X|$ and $|Y|$ we obtain:

\small
\begin{equation*}
	\lim_{k\to\infty} P \left(\fVB{\widehat{|X \cap Y|}_{kH} - |X \cap Y|} \geq t \right) \leq \lim_{k\to\infty} 2 e^{-\frac{2\;k\;t^2}{(|X| + |Y|)^2}} \to 0
\end{equation*}
\normalsize

The above implies that $\widehat{|X \cap Y|}_{kH} \overset{p}{\to} |X \cap Y|$. On the other hand, for $\widehat{|X \cap Y|}_{1H}$ we are in the \textit{sampling without replacement} scheme. This means that the population size (i.e. $|X \cup Y|$) is finite and by taking the limit for $k \rightarrow |X \cup Y|$ in Proposition \ref{bound_int_one_h}, with fixed and finite $|X|$ and $|Y|$, we have already sampled the entire population contrarily to the $k$-Hash case. Thus $\widehat{|X \cap Y|}_{1H}$ is also a consistent estimator of $|X \cap Y|$. Then, for both estimators, the asymptotic unbiasedness follows from consistency and by noticing that both $\widehat{|X \cap Y|}_{kH}$ and $\widehat{|X \cap Y|}_{1H}$ have a bounded variance.

\subsubsection{Sub-Gaussian preliminaries}\label{app:subgaussian}
We recall some key notions of sub-gaussian random variables as they are necessary for the following proofs. First of all, we define $\psi_{X}(\lambda) = \log(\mathbb{E}[e^{\lambda X}])$ as the logarithmic moment generating function (i.e. cumulant) of a generic random variable $X$. For example, if $Z$ is a centered normal random variable with variance $\sigma^2$, we have that $\psi_{Z}(\lambda) = \frac{\lambda^2 \sigma^2}{2}$. It can be shown, we refer the interested reader to chapter 2 of \cite{boucheron2013concentration} for a detailed explanation, that Chernoff's inequality in this case implies, for all $t>0$, that:

\begin{equation}\label{chernorm}
	P(Z \geq t) \leq e^{-\frac{t^2}{2\sigma^2}}    
\end{equation}

The bound above, characterize the decay of the tail probabilities of a centered normal random variable. If the tail probabilities of a generic centered random variable $X$, decrease at least as rapidly as the ones in (\ref{chernorm}) then $X$ is \textit{sub-gaussian}. More formally, a centered random variable $X$ is said to be \textit{sub-gaussian} with variance factor $\sigma^2$ if:

\begin{equation}\label{subgauss}
	\psi_{X}(\lambda) \leq \frac{\lambda^2 \sigma^2}{2}\;\;\;\;\; \forall \lambda \in \mathbb{R}
\end{equation}

We underline that (\ref{subgauss}) only requires $Var(X) \leq \sigma^2$. Moreover, if we call $\mathcal{G}(\sigma^2)$ the collection of random variables for which (\ref{subgauss}) holds (e.g. all bounded random variables belongs to $\mathcal{G}(\sigma^2)$), we can state that:\\

\begin{lemma}\label{lem:sum_subgauss}
	Let $X_{1},\dots,X_{n}$ be sub-gaussians random variables so that $X_{i} \in \mathcal{G}(\sigma^2_{i})$ for every $i \in \{1,\dots,n\}$. Then $\sum_{i=1}^{n} X_{i} \in \mathcal{G}((\sum_{i=1}^{n} \sigma_{i})^2)$. Moreover if the random variables are independent, then $\sum_{i=1}^{n} X_{i} \in \mathcal{G}(\sum_{i=1}^{n} \sigma^2_{i})$.
\end{lemma}

This is due to the fact that (\ref{subgauss}) implies a bound on the moment generating function whose properties, together with the Hölder inequality, help to verify the above statement. For a detailed proof of lemma \ref{lem:sum_subgauss} see theorem 2.7 in \cite{rivasplata2012subgaussian} while for alternative characterizations of sub-gaussianity in terms of growth of moments, we refer to chapter 2 of \cite{boucheron2013concentration}.   

\vspaceSQ{0.5em}

\subsubsection{Concentration bounds for k-Hash and 1-Hash}

We present below the proof of Propositions \ref{bound_int_mh} and \ref{bound_int_one_h}. First, we show the following lemma which we will also use later.
\begin{lemma} \label{lem:hash_est_concentration}
	\begin{gather}
		P(|\hat{J}_1 - J| \geq t),P(|\hat{J}_k - J| \geq t) \leq 2e^{-2t^2k}
	\end{gather}
\end{lemma}
\begin{proof}
	The random variables $k\hat{J}_1$ and $k\hat{J}_k$ follow the hypergeometric and binomial distributions, respectively. Applying the Hoeffding's inequalities in the binomial case, we get the desired inequality. The Serfling's bound can be applied in the case of the hypergeometric distribution. The Serfling's bound always gives better bounds than the Hoeffding's, proving the inequality for $\hat{J}_1$.
\end{proof}


We now show concentration of the sum of the estimators and, therefore also of the individual estimators (by fixing $n=1$).

\begin{theorem} \label{thm:mh_bounds}
	Let $Y_1 = \sum_i^n C_i\frac{\hat{J}_1}{1+\hat{J}_1}, Y_k = \sum_i^n C_i\frac{\hat{J}_k}{1+\hat{J}_k}$. Then for any non-negative constants $C_i$ and $S = \sum_{i=1}^n C_i \frac{J}{1+J}$ 
	\small
	\begin{gather}
		P(|Y_1 - S|>t), P(|Y_k - S|>t) \leq 2 \exp \left( - \frac{ 2 \;k \;t^2}{(\sum_i^n C_i)^2}\right)
	\end{gather}
\normalsize
\end{theorem}
\begin{proof}
	The function $\frac{X}{1+X}$ is 1-Lipschitz and it, therefore, holds that $\forall \; X,X' \in [0,1]$
	\[
	\left|\frac{X}{1+X} - \frac{X'}{1+X'}\right| \leq |X - X'|
	\]
	
	The concentration result from Lemma \ref{lem:hash_est_concentration} then also holds for $\frac{\hat{J_1}}{1+\hat{J_1}}$ and $\frac{\hat{J_k}}{1+\hat{J_k}}$. The random variables
	\begin{gather*}
		\frac{\hat{J_1}}{1+\hat{J_1}} - \frac{J}{1+J}\\
		\frac{\hat{J_k}}{1+\hat{J_k}} - \frac{J}{1+J}
	\end{gather*}
	are therefore sub-gaussian with variance factor $\sigma^2 = \frac{1}{4k}$ (see eq. (\ref{chernorm}), (\ref{subgauss}) and in general section \ref{app:subgaussian}). Now we multiply by $C_i$ each variable, we take the sum and, thanks to lemma \ref{lem:sum_subgauss}, we get that $Y_1 - S$ and $Y_k - S$ are sub-gaussian with variance factor 
	\[
	\frac{(\sum_{i=1}^n C_i)^2}{4k}
	\]
	The theorem follows from the definition of a sub-gaussian random variable (see section \ref{app:subgaussian} for further references).
\end{proof}

We stress here that $\widehat{|X \cap Y|}_{kH}$ derived with $k$-hash can also be interpreted as a \emph{maximum likelihood estimator (MLE)} (cf.~\cref{sec:estimators}) for $|X \cap Y|$ because of the invariance property
outlined in~\cref{sec:back-properties} and detailed in Chapter~7 of~\cite{casella2002statistical}. Indeed since $|M_X \cap M_Y| \sim Bin (k,J_{X,Y})$ we have that $\widehat{J_{X,Y}}_{kH} = \frac{|M_X \cap M_Y|}{k}$ is the maximum likelihood estimator of $J_{X,Y}$ if we assume that the $k$ hash functions are independent and perfectly random (a usual assumption). Then our estimator $\widehat{|X \cap Y|}_{kH}$ is just a function of
$\widehat{J_{X,Y}}_{kH}$ and, because of the invariance of the MLE (see
\textit{Theorem 7.2.10} in \cite{casella2002statistical}), this implies that $\widehat{|X \cap Y|}_{kH}$ inherits all the properties of this class of estimators. In particular, it is consistent and asymptotically efficient since it reaches the \textit{Cramér-Rao Lower Bound} (see \textit{Theorem 7.3.9} in
\cite{casella2002statistical}) meaning that no other estimator can have a lower variance. It is also normally distributed, as $k$ increase, which is useful in general to derive confidence intervals.
 
\subsection{Results \& Derivations for Triangle Counts}
\label{sec:app_triangle}

\subsubsection{Proof of consistency and asymptotic unbiasedness}
Any estimator for triangle count analyzed in PG, is simply a sum of
cardinalities $\widehat{|X \cap Y|}$ for different neighborhoods $X$ and $Y$
(cf.~Section~\ref{sec:algs}):

\begin{gather*}
	\widehat{TC}_{\star} = \frac{1}{3} \sum_{(u,v) \in E} \widehat{|N_u \cap N_v|}_{\star}
\end{gather*}

\noindent
where $\star$ indicates a specific $\widehat{|X \cap Y|}_{\star}$
estimator (cf.~Table~\ref{tab:estimators-summary}). Since we have already proven consistency and asymptotic unbiasedness for each of the $\widehat{|X \cap Y|}_{\star}$ estimators presented in PG, we now can address jointly the consistency of the triangle count estimators. To do so we just need to acknowledge the fact that a sum of consistent estimators is itself a consistent estimator. Indeed this is a direct consequence of the more general \emph{Slutsky theorem} which enable us to state that $\widehat{TC}_{\star} \overset{p}{\to} TC$. The asymptotic unbiasedness then follows from consistency and by noticing that all $\widehat{TC}_{\star}$ estimators have a bounded variance (see all the proofs presented below).       

\subsubsection{Bloom Filters}
We first present the estimator~$\widehat{|X \cap Y|}_{OR}$. This estimator
was introduced before~\cite{swamidass2007mathematical} and uses the single set estimator evaluated on the set union:
\vspaceSQ{-0.5em}
\begin{equation}\label{eq:bf_int_union}
	\widehat{|X \cap Y|}_{OR} = |X| + |Y| + \frac{B_{X \cup Y}}{b} \log \left( 1 - \frac{B_{X \cup
			Y,1}}{B_{X \cup Y}} \right) 
\end{equation}
\vspaceSQ{-0.5em}

%
Note that, to obtain the expression above, we also use the fact that $|X \cup Y| = |X| + |Y| - |X \cap Y|$.\\ 
We now prove the triangle count bound for BF stated in theorem \ref{thm:tc_deviations} for the triangle count OR estimator (the proof is of course valid also for $\widehat{TC}_{AND}$).
\begin{proof}
We first define the mean squared error (mse) as follows:
\begin{align*}
	E[(TC - \widehat{TC}_{OR})^2] = (E[\widehat{TC}_{OR}] - TC)^2 + Var(\widehat{TC}_{OR}) 
\end{align*}
where the equality is a standard identity.
	
Now we bound the first component of the mse which is the squared bias of our estimator. In order to ease the notation from now on we denote $|N_u \cup N_v| \; \forall \; (u,v) \in E$ as $|X|_{i} \; \forall \; i=1,..,m$ where $m$ is the number of edges. In the same fashion, we denote $\widehat{|X|}_{i} \; \forall \; i=1,..,m$ as the estimator of $|N_u \cup N_v| \; \forall \; (u,v) \in E$. Thus we can write:

\begin{align}
	& (E[\widehat{TC}_{OR}] - TC)^2 \nonumber\\
	&= \frac{1}{9}\;\left[\sum_{i=1}^{m} E(\widehat{|X|}_{i}) - |X|_{i}\right]^2 \\
	&\leq  \frac{1}{9}\;\left\{\sum_{i=1}^{m} \sum_{j=1}^{m} \left|[E(\widehat{|X|}_{i}) - |X|_{i}] [E(\widehat{|X|}_{j}) - |X|_{j}]\right|\right\}\\
	&\leq  \frac{1}{9}\;\left\{\sum_{i=1}^{m} \sum_{j=1}^{m} \left|[E(\widehat{|X|}_{i}) - |X|_{i}]\right| \left|[E(\widehat{|X|}_{j}) - |X|_{j}]\right|\right\} \label{cauchy_swartz}\\
	&= \frac{1}{9}\;\left\{\sum_{i=1}^{m} \sum_{j=1}^{m} \sqrt{[E(\widehat{|X|}_{i}) - |X|_{i}]^{2}} \sqrt{[E(\widehat{|X|}_{j}) - |X|_{j}]^{2}}\right\} \\
	&\leq \frac{m^2}{9}\; (1+o(1)) \left(e^{2\Delta b / (B_X-1)} \frac{B_X}{b^2} - \frac{B_X}{b^2} - \frac{2\Delta}{b}\right) \label{bias_bound}
\end{align} 

where (\ref{cauchy_swartz}) follows by Cauchy–Schwarz inequality and (\ref{bias_bound}) by Proposition \ref{bound_ss_BF} which in general bounds the mse (and thus the squared bias by Jensen's inequality) if $2 b \Delta \leq 0.499 B_X \log B_X$ where $B_X = min(B_{N_u \cup N_v})$ with $(u,v) \in E$. In particular we underline that any bound obtained by Proposition \ref{bound_ss_BF}, which is valid for a given set size, is also automatically valid for all the set sizes smaller than that \textit{a fortiori}. Thus we can notice that $2\Delta \geq |X|_{i} \; \forall \; i=1,..,m$ where $\Delta$ is the maximum degree of the input graph which justifies (\ref{bias_bound}).
At this point, it remains to bound the second component of the mse which is the variance of our estimator. Indeed we can write:

\begin{align}
	& Var(\widehat{TC}_{OR}) \nonumber\\
	&= \frac{1}{9}\;Var\left[\sum_{i=1}^{m} \widehat{|X|}_{i}\right] \\
	&=  \frac{1}{9}\;\sum_{i=1}^{m}\sum_{j=1}^{m} Cov(\widehat{|X|}_{i},\widehat{|X|}_{j}) \\
	&\leq  \frac{1}{9}\;\sum_{i=1}^{m}\sum_{j=1}^{m} \sqrt{Var(\widehat{|X|}_{i})}\sqrt{Var(\widehat{|X|}_{j})} \label{cov_ineq} \\
	&\leq \frac{m^2}{9}\; (1+o(1)) \left(e^{2\Delta b / (B_X-1)} \frac{B_X}{b^2} - \frac{B_X}{b^2} - \frac{2\Delta}{b}\right) \label{var_bound}
\end{align} 
where (\ref{cov_ineq}) holds because of the covariance inequality and (\ref{var_bound}) by Proposition \ref{bound_ss_BF} which in general bounds the mse (and thus the variance \textit{a fortiori}) if $2 b \Delta \leq 0.499 B_X \log B_X$. The justification of (\ref{var_bound}) is similar to the one outlined before for the squared bias case. 

Now, again assuming $2 b \Delta \leq 0.499 B_X \log B_X$, we can obtain the bound for the mean squared error of the OR estimator:

\small
\begin{align*}
	& E[(TC - \widehat{TC}_{OR})^2] = \\
	& (E[\widehat{TC}_{OR}] - TC)^2 + Var(\widehat{TC}_{OR}) \\
	&\leq \frac{2m^2}{9}\; (1+o(1)) \left(e^{2\Delta b / (B_X-1)} \frac{B_X}{b^2} - \frac{B_X}{b^2} - \frac{2\Delta}{b}\right) 
\end{align*}    
\normalsize
The above bound is valid also for $\widehat{TC}_{AND}$ since $\Delta \geq |N_u \cap N_v| \; \forall \; (u,v) \in E$ however it can be made tighter for the same reason. Indeed if $b \Delta \leq 0.499 B_X \log B_X$ (where now $B_X = min(B_{N_u \cap N_v})$ with $(u,v) \in E$) by Chebychev inequality:

\begin{align*}
	& P\left(\left|TC - \widehat{TC}_{AND}\right| \geq t\right) \leq \\  
	&\frac{2m^2 (1+o(1)) \left(e^{\Delta b / (B_X-1)} \frac{B_X}{b^2} - \frac{B_X}{b^2} - \frac{\Delta}{b}\right)}{9\;t^{2}}
\end{align*}

which is the statement of Theorem \ref{thm:tc_deviations} for the BF case. 
\end{proof}

\vspaceSQ{-0.5em}

\subsubsection{MinHash}

We can show the concentration of the sum of the set intersection estimators using
theorem \ref{thm:mh_bounds} presented in Appendix \ref{sec:app:int-minhash}. Then for the edge $e_i = uv$, we define $C_i = \text{deg}(u) + \text{deg}(v)$ thus giving us $S = \frac{1}{3} \;\sum_{i=1}^n C_i \frac{J}{1+J} = TC$. We will not consider the scaling factor $\frac{1}{3}$ till the final expressions of the bounds to ease the notation. Thus we can write:
\begin{gather*}
	\sum_{i=1}^m C_i = \sum_{i=1, e_i = uv}^m \text{deg}(u) + \text{deg}(v) = \\
	= \sum_{v \in V} \text{deg}(v)^2
\end{gather*}

Combining the above result with theorem \ref{thm:mh_bounds}, we obtain the triangle count bound for MinHash presented in Theorem \ref{thm:tc_deviations}.

However this bound can be improved if we assume more independence,
which will be satisfied in the case of triangle counting when the maximum
degree is not too large. We now prove a tighter bound under these conditions.

\begin{theorem} \label{thm:mh_bounds2}
	Let $Y_1 = \sum_i^n C_i\frac{\hat{J}_1}{1+\hat{J}_1}, Y_k = \sum_i^n C_i\frac{\hat{J}_k}{1+\hat{J}_k}$, and assume we partition the set of estimators into groups $\mathcal{X}_1, \cdots, \mathcal{X}_\chi$ such that estimators from each set are mutually independent. Then for any non-negative constants $C_i$ and $S = \sum_{i=1}^n C_i \frac{J}{1+J}$ 
	{\footnotesize
		\begin{align*}
			P(|Y_1 - S|>t), P(|Y_k - S|>t) &\leq 2 \exp \left( - \frac{k (\max(0,t-2S/k))^2}{2 (\sum_i^\chi \sqrt{\sum_{d \in \mathcal{X}_i} C_d^2})^2}\right) \\
			&\leq 2 \exp \left( - \frac{k (\max(0,t-2S/k))^2}{2 \chi \sum_i^n C_i^2}\right)
		\end{align*}
	}
\end{theorem}

\begin{proof}
	We modify the proof of Theorem \ref{thm:mh_bounds} by instead considering the random variables
	\begin{gather*}
		\frac{\hat{J_1}}{1+\hat{J_1}} - \frac{J}{1+J} - \mu_1\\
		\frac{\hat{J_k}}{1+\hat{J_k}} - \frac{J}{1+J} - \mu_k
	\end{gather*}
	where $\mu_1$ and $\mu_k$ are chosen so as to make this random variable have mean zero.
	
	We then first sum estimators from each group separately using Lemma \ref{lem:sum_subgauss}, which gives us subgaussian coefficient (i.e. the square root of the variance factor) of $\sqrt{\sum_{d \in \mathcal{X}_i} C_d^2}$. Adding the groups together, we get using again Lemma \ref{lem:sum_subgauss}, that the subgaussian coefficient is $\sigma_\mathcal{X} = \sum_i^\chi \sqrt{\sum_{d \in \mathcal{X}_i} C_d^2}$. 
	To finish the proof of the first inequality, we have to show a bound on $\sum
	C_i \mu_1$ and $\sum C_i \mu_k$. We show the argument for the case of 1-hash,
	the argument for $k$-hash is analogous. Note that $\mu_1$ is the jensen gap of
	$\frac{\hat{J_1}}{1+\hat{J_1}}$. Since $\text{Var}(\hat{J}_1) \leq J/k$, by
	Theorem 1 from \cite{Liao2019}, we have $-J / k \leq
	E[\frac{\hat{J_k}}{1+\hat{J_k}}]- \frac{J}{1+J} \leq 0$. We can bound $-J / k
	\geq 2/k \frac{J}{1+J}$. Therefore, we can bound
	
	\[
	-2S/k \leq \sum C_i \mu_1 \leq 0
	\]
	
	To prove the second inequality, we define the following optimization problem
	\begin{align*}
		\text{maximize} \quad & \sum_{i=1}^n \sqrt{x_n}\\
		\text{subject to} \quad & \sum x_i = c
	\end{align*}
	
	Set $x_i = \sum_{d \in \mathcal{X}_i} C_d^2$ and $c = \sum C_i^2$. We see that
	for every possible assignment of the estimators to the sets
	$\{\mathcal{X}_i\}_{i=1}^\chi$, we have a feasible solution with objective
	value equal to the subgaussian coefficient $\sigma_\mathcal{X}$. Therefore,
	the subgaussian coefficient for any assignment to the groups is dominated by
	the optimum of this optimization problem.
	
	Optimum of this optimization problem is when all $x_i$'s have the same value --
	otherwise one can pick $i,j$ such that $x_i < x_j$ and $0 < \epsilon \leq (x_j
	- x_i)/2$ and then replace $x_i$ by $x_i+\epsilon$ and similarly $x_j$ by
	$x_j-\epsilon$, increasing the objective while retaining feasibility. This
	gives us objective value of $\chi \sqrt{\sum_{i=1}^n C_i^2 / \chi} = \sqrt{\chi
		\sum_{i}^n C_i^2}$
\end{proof} 

To show the final expression of the bound, we use \Cref{thm:mh_bounds2}. Then by Vizing's theorem, $\chi \leq \Delta +1$ and by the same substitution done for the first bound, we have:
\begin{gather*}
	\sum_{i=1}^m C_i^2 = \sum_{i=1, e_i = uv}^m (\text{deg}(u) + \text{deg}(v))^2 \leq \\
	\leq \sum_{i=1, e_i = uv}^m 2(\text{deg}(u)^2 + \text{deg}(v)^2) = 2 \sum_{v \in V} \text{deg}(v)^3
\end{gather*}

Indeed combining the above result with \Cref{thm:mh_bounds2}, we obtain the triangle count bound for MinHash presented in Theorem \ref{thm:tc_deviations} if the maximum degree is $\Delta$.

\subsection{KMV Sketches}\label{sec:app:kmv}

\subsubsection{Single Sets}

\label{sec:kmv_ss}

We state an existing result on the KMV sketching; we use it later to
provide a KMV sketch for $|X \cap Y|$~\cite{david_nagaraja_2003}.
The hash function used with a KMV maps elements from $X$ to real numbers in
$(0,1]$ u.a.r.\footnote{uniformly at random}. Thus, the hashes should be evenly
spaced and one can estimate $|X|$ by dividing the size $k-1$ of $K_X$ by the largest
hash in $K_X$.

\vspaceSQ{-0.5em}
\begin{equation}\label{eq:est_ss_kmv}
	\widehat{|X|}_K = \frac{k-1}{\max \{x | x \in K_X\}}    
\end{equation}{}
\vspaceSQ{-0.5em}

As noted in \cite[§2.1]{david_nagaraja_2003}, the $k$-th smallest value follows
the beta distribution $Beta(\alpha,\beta)$ with shape parameters $\alpha = k$ and $\beta=|X|-k+1$. Now we can get concentration bounds for the estimator: indeed, following~\cite{on_synopses}, we can show that:

\begin{prop}\label{bound_ss_KMV}
	Consider~$\widehat{|X|}_K$ in Eq.~(\ref{eq:est_ss_kmv}), then the probability
	of deviation from the true set size, at a given distance $t \geq 0$, is
	
	\vspaceSQ{-0.5em}
	{
		\ifsq\footnotesize\fi
		\begin{eqnarray*}
			P\left(\left|\widehat{|X|} - |X|\right| \leq t\right) &=& I_{u(|X|,k,
				t/|X|)}(k, |X|-k+1) - \\ 
			& & I_{l(|X|,k,t/|X|)}(k, |X|-k+1)
	\end{eqnarray*}}
	\vspaceSQ{-0.5em}
	
	\noindent
	where $u(|X|,k,t/|X|) = \frac{k-1}{|X|-t}$ and $l(|X|,k,t/|X|) =
	\frac{k-1}{|X|+t}$ and $I_x(a,b)$ is the regularized incomplete beta
	function. \end{prop}

In the case of a KMV estimator bound, we can evaluate:

\ifsq\small\fi
$$I_x(k, |X|-k+1) = \sum_{i=k}^{|X|} {|X| \choose i} x^i (1-x)^{|X| - i}$$
\normalsize

\subsubsection{Set Intersection $|X \cap Y|$}

Given $\mathcal{K}_X$ and $\mathcal{K}_Y$ of size $k_X$ and $k_Y$, one can
construct a KMV $\mathcal{K}_{X \cup Y}$ by taking the $k = \min\{k_X, k_Y\}$
smallest elements from $K_X \cup K_Y$.  $\widehat{|X \cup Y|}_K$,
$\widehat{|X|}_K$ and $\widehat{|Y|}_K$ can be computed using the following equations
(note that the second one uses the exact sizes of $X,Y$ instead of their estimators).

\vspaceSQ{-0.5em}
\begin{equation}\label{eq:est_inter_set_kmv}
	\widehat{|X \cap Y|}_K = \widehat{|X|}_K + \widehat{|Y|}_K - \widehat{|X \cup Y|}_K
\end{equation}{}
\vspaceSQ{-1em}
\begin{equation}\label{eq:est_inter_set_kmv_exactsets}
	\widehat{|X \cap Y|}_K = |X| + |Y| - \widehat{|X \cup Y|}_K
\end{equation}{}
\vspaceSQ{-0.5em}

We present now a simple upper bound (using the union bound) on the probability that $\widehat{|X \cap Y|}_K$ deviates by more than $t$ from the true value. 

\vspaceSQ{-0.5em}
\begin{prop}
	Let $\widehat{|X \cap Y|}_K$ be the estimator defined in (\ref{eq:est_inter_set_kmv}), then the following upper bound for the probability of deviation from the true intersection set size, at a given distance $t \geq 0$, holds:
	
	\vspaceSQ{-0.5em}
	\small
	\begin{equation*}
		\begin{aligned}
			P\left(|\widehat{|X \cap Y|}_K - |X \cap Y|| \geq t \right) \leq P(|\widehat{|X|}_K - |X|| \geq t/3) +\\
			P(|\widehat{|Y|}_K - |Y|| \geq t/3) + P(|\widehat{|X\cup Y|}_K - |X\cup Y|| \geq t/3)
		\end{aligned}
	\end{equation*}
	\vspaceSQ{-0.5em}
	
	\noindent
	where the probabilities on the right can be evaluated with Proposition~\ref{bound_ss_KMV}.
\end{prop}

Yet, if we know the exact size of $X$ and $Y$ (a reasonable assumption for
graph algorithms as the degrees can be easily precomputed), we can get a
considerably better bound.

\begin{prop}
	Let $\widehat{|X \cap Y|}_K$ be the estimator from (\ref{eq:est_inter_set_kmv_exactsets}), then
	
	{\small
		\vspaceSQ{-0.5em}
		\begin{align*}
			& P\left(|\widehat{|X \cap Y|}_K - |X \cap Y|| \geq t\right) = \\
			&I_{u(|X\cup Y|,k,t/|X\cup Y|)}(k, |X\cup Y|-k+1)\\ 
			& - I_{l(|X\cup Y|,k,t/|X\cup Y|)}(k, |X\cup Y|-k+1)
		\end{align*}
		\vspaceSQ{-0.5em}
	}
	
\end{prop}

The bound presented above is a simple application of the identity $|X\cap Y| = |X| + |Y| - |X\cup Y|$ and Proposition \ref{bound_ss_KMV} on the estimator of $|X \cup Y|$.
\fi

\end{document}